\documentclass[11pt, a4paper]{article}
\usepackage[a4paper, portrait, margin=1.2in]{geometry}
\usepackage{amsfonts,amsmath,amssymb,amsthm}
\usepackage[hidelinks]{hyperref}
\hypersetup{colorlinks,citecolor=blue,linkcolor=blue,urlcolor=black}
\usepackage[hypcap]{caption}
\usepackage[utf8]{inputenc}
\usepackage[english]{babel}
\usepackage{graphicx}
\usepackage{natbib}
\usepackage{booktabs}	
\usepackage[capposition=top]{floatrow}
\usepackage{cases}
\usepackage{float}
\usepackage{lscape}
\usepackage{tikz}
\usepackage{pgfplots}
\pgfplotsset{compat=1.18}
\usetikzlibrary{calc,positioning,backgrounds,fit,decorations.pathreplacing}
\usepackage{adjustbox}
\usepackage{rotating}
\usepackage{comment}
\usepackage{color}
\usepackage{multirow}
\usepackage{calc}
\usepackage[nameinlink]{cleveref}
\usepackage{setspace} 
\usepackage[shortlabels]{enumitem}
\usepackage{epstopdf}
\usepackage[final]{pdfpages}
\usepackage{csquotes}
\usepackage{mathtools}
\setlistdepth{9}

\let\emptyset\varnothing

\usepackage[multiple]{footmisc}
\def\signed #1{{\leavevmode\unskip\nobreak\hfil\penalty50\hskip2em
  \hbox{}\nobreak\hfil(#1)%
  \parfillskip=0pt \finalhyphendemerits=0 \endgraf}}

\newsavebox\mybox

\theoremstyle{plain}
\newtheorem{thm}{Theorem}

\newtheorem{prop}{Proposition}
\newtheorem{coro}{Corollary}

\newtheorem*{remark*}{Remark}

\theoremstyle{definition}
\newtheorem{defi}{Definition}

\theoremstyle{remark}

\newtheorem*{claim*}{Claim}

\theoremstyle{definition}
\newtheorem{example}{Example}[section]

  \theoremstyle{plain}
  
  \theoremstyle{remark}
  
\theoremstyle{plain}

\crefname{thm}{Theorem}{Theorems}
\Crefname{thm}{Theorem}{Theorems}

\crefname{lem}{Lemma}{Lemmas}
\Crefname{lem}{Lemma}{Lemmas}

\crefname{prop}{Proposition}{Propositions}
\Crefname{prop}{Proposition}{Propositions}

\crefname{coro}{Corollary}{Corollaries}
\Crefname{coro}{Corollary}{Corollaries}

\crefname{defi}{Definition}{Definitions}
\Crefname{defi}{Definition}{Definitions}

\crefname{ex}{Example}{Examples}
\Crefname{ex}{Example}{Examples}

\crefname{example}{Example}{Examples}
\Crefname{example}{Example}{Examples}

\title{Visibly Fair Mechanisms\thanks{We thank Mustafa Oguz Afacan, Samson Alva, Yuan Gao, Sambuddha Ghosh, Xi Jin, Vikram Manjunath, Debasis Mishra and Yu Zhou for helpful comments. We are grateful for the feedback at several conferences and seminars, including Indian Statistical Institute Delhi, Queen's University of Belfast, University of York, University of Manchester, 14th Conference on Economic Design at University of Essex, 2025 CCBEF-IES Workshop at SWUFE, Greater Bay Area Market Design Workshop at University of Macau, Lisbon Meetings, and SAET Conference.}}

\author{Inácio Bó\thanks{Department of Economics, Faculty of Social Sciences, University of Macau, Macau SAR. Email: inaciobo@um.edu.mo}
\and Gian Caspari\thanks{Department of Market Design, ZEW — Leibniz Centre for European Economic Research, Mannheim 68161, Germany. Email: gian.caspari@zew.de}
\and Manshu Khanna\thanks{Peking University HSBC Business School, Shenzhen 518055, China. Email: manshu@phbs.pku.edu.cn}}

\begin{document}
\date{May, 2026}
\maketitle

\begin{abstract}	
Priority-based allocation often requires eliminating justified envy, making serial dictatorship (SD) the only non-wasteful direct mechanism with that property. However, SD's outcomes can conflict with the policymaker's objectives. We introduce visible fairness, a framework where fairness is evaluated using coarser information. This is achieved by designing message spaces that strategically conceal information that could render desired allocations unfair. We characterize these mechanisms as generalizations of SD, establish conditions for strategy-proofness, and show how to implement distributional constraints. This creates a new trade-off: achieving distributional goals may require limiting preference elicitation, forgoing efficiency gains even when compatible with the constraints. A simulation exercise shows that merging zones, and thereby eliciting more comparisons, raises welfare substantially on average. We apply the framework to India's cadre allocation mechanism, diagnosing failures of strategy-proofness and showing how three alternative modular-upper-bound specifications can each address the reform's distributional objectives.
\bigskip 
        
\noindent  \parbox[t][11mm]{.16\linewidth}{\textbf{Keywords:}}
	\parbox[t][11mm]{.84\linewidth}{Matching Theory, Market Design, Indirect Mechanisms} 

    \vspace{-0.2cm}
    
	  \noindent  \parbox[t]{.16\linewidth}{ \textbf{JEL:}}
	\parbox[t]{.84\linewidth}{C78, D47, D78}
\end{abstract}

\newpage
\onehalfspacing

\section{Introduction}

Priority-based assignments are pervasive in a wide range of real-world matching contexts, including university admissions, public-sector placements, and cadet-branch allocations in military academies \citep[see, e.g.,][]{balinski1999tale,sonmez2013matching}. In a typical priority-based system, participants are strictly ranked—based on exam scores or a merit list for instance—and are assigned to positions accordingly. The central fairness requirement in these contexts is that no lower-ranked participant should occupy a seat that a higher-ranked participant strictly prefers; otherwise, the latter has a legitimate grievance, known as \emph{justified envy} \citep{abdulkadirouglu2003school}. Such fairness concerns become even more pronounced when the priorities at stake represent strongly protected interests—like property rights or national exam rankings—where even a single instance of justified envy can trigger legal and administrative challenges.\footnote{One example of these legal challenges is the Federal University of Bahia hiring suit (Brazil, 2025), where a federal judge \emph{blocked} the university from hiring a lower-scoring quota applicant and ordered the single vacancy awarded to the exam’s top scorer \citep{ufba2025}. Another was the Italian national residency “fiasco”, in which the Regional Administrative Tribunal of Lazio annulled a ruling that forced higher-ranked doctors to forfeit more-preferred specialties while lower-ranked peers advanced \citep{lazio2016}. Both rulings treat the harm as a breach of the merit order—i.e., a violation of justified envy—showing that such breaches readily provoke litigation.}

Under the standard approach of designing \emph{direct mechanisms}—where each participant reports a complete ranking over positions—\emph{Serial Dictatorship (SD)} is in fact the \emph{only} mechanism that can satisfy non-wastefulness---i.e., not leaving desirable positions unfilled---and no-justified-envy. In SD, the highest-priority participant is assigned their top choice, then the next highest-priority participant is assigned their top choice from the remaining positions, and so on. This procedure prevents any lower-ranked participant from ending up in a spot that a higher-ranked participant strictly prefers, ensuring no-justified-envy. However, SD can produce allocations that are misaligned with policy goals, such as excessive clustering of top-ranked participants in a small set of elite locations, or  undesirable regional or demographic distributions. When considering standard direct mechanisms, there is \emph{no alternative} to SD that both respects strict priorities and operates on full preference lists. Thus, policymakers appear to face a dilemma: given the requirement for ``fairness by priority'', how can one construct rules in pursuit of better distributional outcomes?

When we look at real-world priority-based assignments, we see many depart from the fully “direct” approach, as illustrated by the following cases:

\medskip

\noindent \textbf{Case 1.}~In the Indian Administrative Services (IAS), officers were formerly assigned to state cadres\footnote{In the IAS context, a \emph{cadre} is the state-level administrative unit to which an officer is allocated and within which the officer ordinarily serves and develops her career. Some cadres correspond to a single state, while others group states or union territories.} in a priority-driven process aligned with exam-based merit. This arrangement, which was essentially a serial dictatorship with some modifications, produced undesirable allocations exhibiting \emph{homophily}, i.e., a propensity for officers to serve in or near their home regions. Such geographic clustering was seen as compromising the national integration objective of the service \citep{thakurMatchingCivilService2023}. A 2017 reform imposed a zone-based scheme: all states were partitioned into five geographic zones, and each officer now submits a separate ranking of states \emph{within} each zone and a separate ranking of the zones, rather than submitting a single list over the entire country (see \Cref{fig:iasRank}). Under the revised mechanism, no lower-ranked officer receives a state preferred by a higher-ranked officer \emph{inside the zone where they are ultimately matched}. At the same time, the zonal structure allows officers to be more evenly dispersed across India, advancing distributional goals without overriding the preferences participants actually report.

\begin{figure}[h!]
    \centering
\adjustbox{width=\textwidth}{\begin{tikzpicture}[font=\footnotesize]

\draw[thick] (0,0) rectangle (3,4.5);
\draw[thick] (0,3.5) -- (3,3.5);
\node at (1.5,4) {\textbf{Zone 3}};
\node[
    anchor=north west,
    align=left,
    text width=2.5cm
] at (0.1,3.3) {
    Gujarat\\
    Madhya Pradesh\\
    Chhattisgarh\\
    Maharashtra
};

\node at (3.5,2) {\Large $>$};

\draw[thick] (4,0) rectangle (7.5,4.5);
\draw[thick] (4,3.5) -- (7.5,3.5);
\node at (5.5,4) {\textbf{Zone 1}};
\node[
    anchor=north west,
    align=left,
    text width=3cm
] at (4.1,3.3) {
    AGMUT\\
    Punjab\\
    Jammu \& Kashmir\\
    Himachal Pradesh\\
    Uttarakhand\\
    Haryana\\
    Rajasthan
};

\node at (8,2) {\Large $>$};

\draw[thick] (8.5,0) rectangle (12,4.5);
\draw[thick] (8.5,3.5) -- (12,3.5);
\node at (10,4) {\textbf{Zone 4}};
\node[
    anchor=north west,
    align=left,
    text width=3cm
] at (8.6,3.3) {
    Manipur\\
    West Bengal\\
    Sikkim\\
    Nagaland\\
    Assam-Meghalaya\\
    Tripura
};

\node at (12.5,2) {\Large $>$};

\draw[thick] (13,0) rectangle (16,4.5);
\draw[thick] (13,3.5) -- (16,3.5);
\node at (14.5,4) {\textbf{Zone 2}};
\node[
    anchor=north west,
    align=left,
    text width=2.5cm
] at (13.1,3.3) {
    Uttar Pradesh\\
    Odisha\\
    Bihar\\
    Jharkhand
};

\node at (16.5,2) {\Large $>$};

\draw[thick] (17,0) rectangle (20,4.5);
\draw[thick] (17,3.5) -- (20,3.5);
\node at (18.5,4) {\textbf{Zone 5}};
\node[
    anchor=north west,
    align=left,
    text width=2.5cm
] at (17.1,3.3) {
    Andhra Pradesh\\
    Telangana\\
    Kerala\\
    Tamil Nadu\\
    Karnataka
};

\end{tikzpicture}
}
    \caption{Example of a preference ranking in the 2017 IAS Mechanism}
    \label{fig:iasRank}
\end{figure}

\medskip

\noindent \textbf{Case 2.}~ Under the U.S. Military Academy matching process of cadets to military branches used in 2006, each cadet (i) ranked the branches and (ii) stated, for every branch, whether they would accept a longer service obligation in exchange for a priority boost \citep{sonmez2013matching,greenbergRedesigningUSArmy2024}. This elicitation did not allow them to report \emph{cross-branch} trade-offs. Omitting this information let the Academy honour the official order-of-merit list while making it possible to steer more cadets into longer commitments that would have been rejected under full preference elicitation. As \citet{sonmezMinimalistMarketDesign2024} explains, this design kept such priority violations hidden:

\begin{quote}
“Several years later, in 2019, I finally learned why the Army initially did not pursue a potential reform of the USMA-2006 mechanism. (...) Any failure of the no-justified-envy axiom rooted in this first issue was also ‘invisible’ to the Army. When a cadet receives his first-choice branch at the increased price but prefers his second choice at the base price, this information was simply unavailable under the strategy space of the USMA-2006 mechanism.”
\end{quote}

\medskip

\noindent \textbf{Case 3.}~In the Chinese college admissions system, applicants submit a structured rank-order list in which majors are nested under colleges, effectively enforcing a lexicographic hierarchy: once a college is deemed higher-ranked, every major it offers is treated as strictly preferred to any program at a lower-ranked college. Although this structure is known to generate numerous cases of justified envy in practice, none of these can be challenged under the restricted message space \citep[see][]{huchinesecolleges}.\footnote{As of January 2025, 23 out of 31 provinces in China retain the nested rank-order procedure \citep{huchinesecolleges}, including Shanghai, Beijing, Tianjin, Hainan, Jiangsu, Fujian, Hubei, Hunan, Guangdong, Heilongjiang, Gansu, Jilin, Anhui, Jiangxi, Guangxi, Shanxi, Henan, Shaanxi, Ningxia, Sichuan, Yunnan, Tibet, and Xinjiang.  We provide screenshots from Fujian and Shanghai's official college-major list sample form in Appendix~\ref{appendixindirect}.}

\medskip

As the preceding cases show, limiting what participants may report can hide genuine priority breaches. Following this, we say a mechanism is \textbf{visibly fair} when, given the elicited (partial) preferences, no outcome appears to violate priority. For example, if seats are partitioned into zones and participants may only rank seats \emph{within} each zone, allocating by priority inside every zone looks perfectly fair—even though a cross-zone comparison (never elicited) might reveal a lower-ranked participant holding a seat a higher-ranked participant prefers. By restricting the scope of reported preferences, policy makers can pursue goals such as geographic diversity while keeping any latent violations invisible. A closely related idea appears already in \citet{greenbergRedesigningUSArmy2024}, who introduced the notion of \emph{detectable priority reversals}, a concept that corresponds precisely to visible (un)fairness in the context of the US Army's branching mechanism.

Inspired by these observations, \emph{we examine the design problem of assigning officers to positions under a strict priority ordering while maintaining visible fairness}. In contrast to standard models that fix a preference-reporting format, in this framework \textbf{the policy maker chooses both the message space} (the form of partial preferences agents can report) \textbf{and the outcome rule}. Our analysis provides a framework and results on how to configure these two elements together so as to achieve desired policy objectives.

\subsection*{Summary of Results}\label{subsec:summary-results}

Our analysis delivers four main sets of results.

First, we pin down the precise structure that visible fairness imposes on allocation rules.  \Cref{thm:CharacterizationVisiblyFairMechanisms} shows that any visibly–fair mechanism must operate as an \emph{$m$-queue allocation}: officers are processed in strict priority order, and each is assigned a state that is undominated, among the remaining states, within the partial ranking they are allowed to report.  When each officer's message space induces an officer-specific partition of the state space into zones, this characterization yields the more specific results in \Cref{thm:CharacterizationIncontestableMechanisms2,thm:CharacterizationIncontestableMechanisms3}. In this setting, the only visibly-fair rules are partitioned priority mechanisms. Moreover, when rankings are permitted across each officer's zones, the only visibly-fair rules are ranked-partitioned priority mechanisms. While visible fairness implies serial dictatorship (and therefore strategy-proofness) when using direct mechanisms (\Cref{cor:SDUniqueVisiblyFairDM}), that is not the case for general message spaces. We show in \Cref{thm:Truthful1} that strategy-proofness is obtained exactly when the mechanism also satisfies \emph{expressiveness} and  \emph{(weak) availability}, two properties that rule out profitable deviations when the mapping from message profiles to outcomes is more general.

Second, we introduce a flexible way to encode distributional goals through \emph{modular upper-bounds}.  A quota system is modular when every bound caps groups of officers within an arbitrary subset of states (\Cref{def:ModularUpperBounds}). Modular bounds induce \emph{constraint-induced message spaces} with a partition of states. This provides a direct link between policy objectives and elicitation: the distributional constraint itself determines which distinctions among states are relevant, and hence which preference information the mechanism should request. Building on this structure, the \emph{Modular Priority Mechanism} (\Cref{def:modularPriorityMechanism}) processes officers by priority while dynamically closing zones whose relevant caps have just filled.  \Cref{modulartheorem} proves that this mechanism simultaneously respects every modular bound, remains visibly fair, and is strategy-proof.

Third, we show that the choice of message space creates a genuine elicitation problem, not merely a way to hide priority violations. Coarser message spaces enlarge the set of allocations that can be implemented without visible unfairness, but they also reduce the mechanism's ability to condition on welfare-relevant preference information. Conversely, richer message spaces make the mechanism more responsive to participants, but they may reveal comparisons that constrain what can be done while preserving visible fairness and modular upper-bounds. We illustrate the non-monotonicity of individual effects, quantify the average welfare value of merging zones through simulations, and then prove in \Cref{thm:NoMechanismAllThree} that the trade-off has a formal limit: for some modular upper-bound systems, no static mechanism can simultaneously satisfy visible fairness, respect the caps at every message profile, and achieve constrained Pareto efficiency. The force behind the impossibility is that message spaces must be fixed preemptively: comparisons that could ever make fairness and feasibility incompatible must be excluded, even when, at the realized preference profile, those same comparisons would identify cap-respecting Pareto improvements.

Fourth, we bring the full toolkit to bear on the cadre allocation mechanism for India's All India Services (\Cref{sec:IAS}).  We show that the 2017 zone-based mechanism is visibly fair when insider quotas are removed, but that its interleaving rule violates strategy-proofness. Turning to the distributional intent of the reform, we formulate three modular-upper-bound specifications---zonal homophily caps, cross-zone balance constraints, and graduated home-proximity caps---each capturing a different plausible reading of the policymaker's objectives. Notably, the cross-zone-balance specification endogenously recovers the actual five-zone partition, while the home-proximity specification produces a richer, type-specific four-zone partition through nested constraints.

\subsection*{Related Literature}

A foundational theory for mechanisms with restricted message spaces comes from \cite{greenPartiallyVerifiableInformation1986},  who study settings where participants are restricted to a limited message space that depends on their true state. Their key insight is that by constraining the information a participant can reveal, the set of implementable outcomes can be expanded beyond what is possible in standard direct mechanisms. 

Beyond the already discussed cases of the Indian Administrative Service (IAS) cadre allocation \citep{thakurMatchingCivilService2023}, the U.S.\ Military Academy’s cadet-branch matching \citep{sonmez2013matching,sonmez2013bidding,greenbergRedesigningUSArmy2024}, and Chinese college admissions \citep{huchinesecolleges}, other real-world mechanisms also limit the extent of preference reporting. Relatedly, \citet{aygunTurhan2023AffirmativeActionIndia} document a restricted preference-language problem in Indian affirmative-action allocations, where applicants rank institutions even though their relevant preferences may be over institution--vertical-category positions. In school choice, for example, some systems cap the number of schools an applicant may rank \citep{haeringer2009constrained, calsamigliaConstrainedSchoolChoice2010}, while others allow applicants to “bundle” schools into groups without inter-group comparisons \citep{huang2025bundled}. 

Dynamic procedures often also restrict preferences and by doing so, might reduce complexity \citep{pyciaTheorySimplicityGames2023}. \cite{boIterativeDeferredAcceptance2022} propose Iterative Deferred Acceptance by letting participants choose from menus, obtaining stable results without demanding complete rankings, and \cite{boPickanObjectMechanisms2024,mackenzieMenuMechanisms2022} extend this idea to mechanisms that sequentially offer feasible outcomes—enhancing privacy and performing well in controlled experiments. Even small constraints, like limiting the length of rank-ordered lists, can disrupt classical incentive properties: \cite{haeringer2009constrained} show how capping the number of ranked schools compromises the usual strategy-proofness of the Deferred Acceptance procedure. Meanwhile, \cite{Caspari2024} propose precise conditions for stability and incentives with non-standard preference formats. Collectively, these studies highlight how restricting preference elicitation can open new design possibilities while preserving key desiderata --- an insight we leverage in defining and deploying visible fairness.
 \cite{decerfIncontestableAssignments2024} is perhaps the work most conceptually related to ours. Their notion of \emph{incontestable assignments} describes an environment where participants cannot fully observe others’ preferences or placements, leaving them unable to identify certain envy issues. This parallels our concept of \emph{visible fairness}, in which certain violations become undetectable. The key difference, however, is that \cite{decerfIncontestableAssignments2024} derive their informational constraints from the participants’ limited ability to view the full outcome, whereas in our framework, these constraints are intentionally \emph{designed} by the policy maker. Specifically, in our setup the policy maker restricts what participants can report so as to preclude distributional tensions that might otherwise manifest as visible grievances.\footnote{In many real-life applications, such as the IAS hiring and public sector hiring contests in Italy and Brazil, transparency requirements imply the public disclosure of information such as exam papers, scoring sheets, and even interview recordings. These are not considered private and must be accessible to ensure administrative and social oversight \citep{ANAC2025,cgu_parecer_1103_2023,AbizadaHiringPool2021}. Under these informational circumstances, incontestability might become equivalent to standard elimination of justified envy, and therefore under single priority imply serial dictatorship.} In addition, while \cite{decerfIncontestableAssignments2024} accommodate a variety of school-specific priorities, our model considers a single, strict priority ranking that orders all participants.

Another paper that considers mechanisms using alternative message spaces is \cite{doganCavalloGeography2024}. The authors analyze Italy’s nationwide teacher‑mobility scheme, in which teachers may rank entire municipalities, districts, or provinces—nested geographic units that bundle many schools into a single item on their list. They show that the tie‑breaking rule used to resolve these coarse rankings might allocate lower‑priority teachers ahead of higher‑priority ones, creating detectable instances of justified envy. They also show that these result in legal challenges: Italian courts have repeatedly upheld merit‑based claims, and parliamentary testimony records more than 1,000 lawsuits filed each year, on average, over such priority violations.

Partial preferences have also been studied in other contexts within market design.
One strand lets participants declare indifference classes directly: \cite{ERDIL2017268,manjunathStrategyproofExchangeTrichotomous2021,anderssonOrganizingTimeExchanges2021} build mechanisms that treat weak orders --- strict ranks punctuated by ties --- as the primitives, and then exploit those ties to recover efficiency and strategy‐proofness. A second strand considers problems in which not every pair of outcomes can be compared, evaluating which adaptation of standard properties, such as stability, can nonetheless sustain strategy-proofness under suitable conditions \citep{Caspari2024,kuvalekarMatchingIncompletePreferences2023}. Typically, these frameworks rely on “weak stability,” where participants who are indifferent or indecisive simply cannot block assignments. In contrast, our approach presumes participants \emph{do} have complete preferences but are deliberately constrained from revealing them in full.

The design of matching markets with distributional constraints through quota systems has emerged as a critical area of research in market design, balancing equity objectives with efficiency and stability considerations \citep{echeniqueHowControlControlled2015, abdulkadiroglu2023market}.\footnote{Practical implementations in education markets reveal both the potential and complexity of quota systems. \cite{combe2022design} quantified these trade-offs through France’s teacher assignment reforms, where experience-based distribution constraints reduced novice teacher concentrations in disadvantaged schools by 18\% without significant efficiency losses. \cite{combe2022market} extended this through a reassignment algorithm that prioritized understaffed schools, demonstrating how temporal flexibility in constraints (allowing multi-year adjustment periods) could mitigate short-term displacement costs. } \cite{kamadaEfficientMatchingDistributional2015} introduced the idea of matching with distributional constraints, showing that conventional stable matching algorithms can break down under strict regional quotas. To address these deficiencies, they introduced new mechanisms that ensure such constraints are respected while preserving or improving upon stability, efficiency, and incentive alignment. Subsequent research refines and generalizes these insights: for instance, \cite{kamadaStabilityStrategyproofnessMatching2018} identify structural conditions enabling strategy-proof and stable mechanisms under distributional constraints, and their more recent work \citep{kamadaFairMatchingConstraints2024} extends stability ideas to increasingly nuanced affirmative action policies.\footnote{Additional contributions include \cite{azizMatchingDiversityConstraints2020}, who introduce the principle of “cutoff stability” for diversity-constrained matching, and \cite{kojimaJobMatchingConstraints2020}, who identify conditions ensuring that distributional constraints do not undermine substitutability in job-matching markets.} 

At the core of this field lies the tension between rigid distributional quotas and the flexible preferences of participants. \cite{fragiadakis2017improving} demonstrated this through military cadet matching, where static reservation systems created inefficiencies by locking seats for specific groups prematurely. Their dynamic quota mechanism represented a paradigm shift, adjusting reservation targets based on revealed preferences while maintaining strategy-proofness. This approach inspired subsequent innovations like the Adaptive Deferred Acceptance (ADA) mechanism by \cite{gotoDesigningMatchingMechanisms2017}, which introduced hereditary constraints—rules where satisfying a constraint automatically satisfies all its subsets. The ADA mechanism’s success in Japanese medical residency matching showed that carefully designed constraints need not sacrifice core market principles like nonwastefulness and strategy-proofness.

\bigskip

\noindent \textbf{Structure of the paper.}
In \Cref{sec: Model and Definitions}, we introduce the model and definitions, covering partial preferences, feasible allocations, and our message-space framework. \Cref{sec: Visibly Fair Mechanisms} then characterizes visibly fair mechanisms, identifying them as queue-allocation variants and examining particular message spaces, such as zonal message spaces. \Cref{sec: Incentives} turns to incentives, specifying exactly when these mechanisms are strategy-proof via the conditions of expressiveness and availability. \Cref{sec:DistributionalObjectives} presents our results on distributional objectives, including modular upper-bounds and the Modular Priority Mechanism. \Cref{sec:efficiency} explains why preference elicitation remains valuable even when the designer restricts messages, uses simulations to quantify the average welfare value of merging zones, and establishes an impossibility result: in general, no one-shot mechanism can combine visible fairness and constrained Pareto efficiency. \Cref{sec:IAS} applies the framework to the Indian Administrative Service cadre allocation mechanism, analyzing its fairness and incentive properties and developing three modular-upper-bound specifications that capture different distributional readings of the 2017 reform. \Cref{sec: Conclusion} concludes. Proofs of the main-text results are relegated to \Cref{sec: Appendix Proofs}, while additional simulation details appear in \Cref{app:simulations}.

\section{Model and Definitions}
\label{sec: Model and Definitions}

A problem consists of:

\begin{enumerate}
    \item a finite set of \textbf{officers} $I=\{i_1,i_2,\ldots, i_n\}$,
    \item a finite set of \textbf{states} $S=\{s_1,s_2,\ldots, s_m\}$,
    \item a positive integer \textbf{capacity} $q_s \in \mathbb{N}:=\{1,2,\ldots\}$ for each state $s\in S$, such that $\sum_{s\in S}q_s \ge n$,
    \item a strict \textbf{preference} (\emph{asymmetric, complete, and transitive}) for each officer $(\succ_i)_{i\in I}$ over states $S$,\footnote{We denote by  $\succsim_{i}$ the associated weak preference---that is, $s\succsim_{i}s'\iff s\succ_{i}s'\text{ or }s=s'$.} and
    \item a \textbf{priority} ranking $\pi$ of officers $I$, where officer $i$ is ranked higher than officer $j$ if $\pi(i) < \pi(j)$.
\end{enumerate}

For any given problem, the goal is to produce an allocation of officers to states.  
Formally, an \textbf{allocation} $a=(a_i)_{i\in I}$ is a list specifying a state $a_i \in S$ for each officer $i\in I$.  
An allocation is \textbf{feasible} if, for each $s\in S$, we have $|\{ i\in I: a_i=s\}|\leq q_s$. We denote the set of all feasible allocations by $\mathcal{A}$. Furthermore, without loss of generality, we assume that officers with lower subscripts have a higher priority, i.e., $\pi(i_1)< \pi(i_2) < \cdots < \pi(i_n)$.

While the allocation decision is based on officers' reported preferences and assigned priorities, in our setup officers do not necessarily communicate their full preferences directly. Instead, they provide partial preference information from a message space.

Let $M_i$ denote the \textbf{message space} of officer $i\in I$. Each \textbf{message} $m_i\in M_i$ induces an expressed preference relation $\succ_{m_i}$ on $S$. We assume that $\succ_{m_i}$ is a \emph{strict partial order}, that is, it is \emph{irreflexive} and \emph{transitive}. Some message spaces below are described by listing directly stated comparisons and then closing them under transitivity. In all cases, however, the object induced by a message and used in every definition is the resulting strict partial order $\succ_{m_i}$.

Throughout the paper, all notions of comparability, truthfulness, visible fairness, maximality, expressiveness, and visible efficiency are defined using $\succ_{m_i}$.

We say that states $s$ and $s'$ are \textbf{comparable} under $m_i$ if
\[
s=s',\quad s\succ_{m_i}s',\quad\text{or}\quad s'\succ_{m_i}s.
\]
We denote by $\succeq_{m_i}$ the associated weak relation:
\[
s\succeq_{m_i}s'
\quad\Longleftrightarrow\quad
s=s' \text{ or } s\succ_{m_i}s'.
\]
We denote message space profiles and message profiles by $M=(M_i)_{i\in I}$ and
$m=(m_i)_{i\in I}$.

Given officer $i$'s true preference $\succ_i$ over states, a message $m_i\in M_i$ is \textbf{truthful} for $\succ_i$ if
\[
\succ_{m_i}\subseteq \succ_i,
\]
that is, for all $s,s'\in S$,
\[
s\succ_{m_i}s' \Longrightarrow s\succ_i s'.
\]
Let
\[
\mathcal{T}_i(\succ_i):=\{m_i\in M_i:\succ_{m_i}\subseteq \succ_i\}.
\]

Throughout the paper, we restrict attention to message spaces that allow for truthful messages: for every officer $i\in I$ and every strict preference $\succ_i$, the set $\mathcal{T}_i(\succ_i)$ is nonempty. That is, message spaces cannot force any comparison that would be inconsistent with some possible true strict preference.

Finally, having defined both allocations and messages, a \textbf{M-mechanism} is a function from message profiles to allocations, $\psi: M \to \mathcal{A}$.  
We will use the shorthand \textbf{mechanism} when the space of message profiles is clear from the context.

Throughout, we index objects by officers $i \in I$. When officers act in sequence, we write $i_k$ for the officer who moves at step $k$, and use the shorthand $a_k \equiv a_{i_k}$, $m_k \equiv m_{i_k}$, $\mathcal{C}_k \equiv \mathcal{C}_{i_k}$, $Z_k \equiv Z_{i_k}$, and analogously for any other object indexed by officers. The forms $x_k$ and $x_{i_k}$ refer to the same object.

\section{Visibly Fair Mechanisms}
\label{sec: Visibly Fair Mechanisms}

\subsection{Visible Fairness}

We now introduce the key notion of visible fairness:

\begin{defi}
\label{def:visibleFairness}
An allocation $a$ is \textbf{visibly unfair under $m$} if for some $i\in I$ either

\begin{itemize}
    \item[i)]there is a $j\in I$ such that $a_i\not=a_j$, $\pi(i) <  \pi(j)$, and $a_j \succ_{m_i} a_i$, or
    \item[ii)] there is an $s\in S$ such that $a_i\neq s$, $|\{ j\in I: a_j=s\}|< q_s$, and $s \succ_{m_i} a_i$.
\end{itemize}
A mechanism $\psi$ is \textbf{visibly fair} if there does not exist $m\in M$ such that $\psi(m)$ is visibly unfair under $m$. 
\end{defi}

At first sight, visible fairness appears to be a combination of standard non-wastefulness and elimination of justified envy. The distinction, however, lies in the fact that $m_i$, in general, is incomplete, and therefore some existing wastefulness or justified envy is ``invisible'' due to the limits to the expression of the associated preferences imposed by the message space.

Let $G(X,m_i)$ be the set of all $m_i$-maximal elements of $X\subseteq S$, that is
\[
G(X,m_i)=\{s\in X:\nexists s'\in X\text{ such that }s'\succ_{m_i}s\}.
\]

If $X\neq\emptyset$, then $G(X,m_i)\neq\emptyset$, because $\succ_{m_i}$ is a strict partial order and $S$ is finite.\footnote{This follows because every finite strict partial order has a maximal element on each nonempty subset---see \citet[Theorem 2.6]{bossert2010consistency}.} We next introduce a new family of mechanisms:

\begin{defi}
\label{def:mQueueMechanisms}
A mechanism $\psi$ is an \textbf{$\mathbf{m}$-queue allocation mechanism} if there exist deterministic selectors
\[
g_k:(2^S\setminus\{\emptyset\})\times M\to S,\qquad k=1,\ldots,n,
\]
such that, for every nonempty $X\subseteq S$ and every $m\in M$,
\[
g_k(X,m)\in G(X,m_k),
\]
and $\psi(m)$ is the outcome produced by the following procedure.

\begin{itemize}
   \item[] \textbf{Step $0$}: Set $S^1= S$.
   \item[] \textbf{Step $k$ ($1\leq k\leq n$)}:  Set $a_k=g_k(S^k,m)$. If the number of officers assigned to $a_k$ reaches $q_{a_k}$, that is $|\{\ell\in\{1,\ldots,k\}:a_\ell=a_k\}|=q_{a_k}$, then $S^{k+1}=S^k\setminus\{a_k\}$. Otherwise, $S^{k+1}=S^k$.
\end{itemize}
\end{defi}

The $m$-queue allocation mechanism is the natural partial-preference analogue of the serial dictatorship mechanism. Starting from the highest-ranked officer, officers are matched to undominated states among those not yet filled to capacity. The difference is that there may be multiple undominated states, so a selection criterion among them must also be specified.

\begin{thm}
\label{thm:CharacterizationVisiblyFairMechanisms}
A mechanism $\psi$ is visibly fair if and only if it is a $m$-queue allocation mechanism. 
\end{thm}

When considering complete message spaces, Theorem \ref{thm:CharacterizationVisiblyFairMechanisms} gives us the following corollary:

\begin{coro}
\label{cor:SDUniqueVisiblyFairDM}
    If for every $i\in I$, every $m_i\in M_i$ and every $s,s'\in S$, states $s$ and $s'$ are comparable under $m_i$, then Serial Dictatorship is the unique visibly fair $M$-mechanism. 
\end{coro}

Theorem \ref{thm:CharacterizationVisiblyFairMechanisms} provides a full characterization of visibly fair mechanisms as the set of \(m\)-queue allocation mechanisms. While this definition is compact, it captures an extensive class of mechanisms. Designing an \(m\)-queue mechanism involves two layers of choices. First, one must specify the message spaces officers can use. Second, given a message profile, one must determine which of the undominated, or \(m_i\)-maximal, states each officer is matched to. This second choice need not be made officer by officer in isolation: the selection among an officer's \(m_i\)-maximal states may depend on the entire message profile. Consequently, once messages are coarsened, the designer obtains a large combinatorial family of visibly fair mechanisms, corresponding to different ways of selecting among undominated states as a function of reported messages. In contrast, as indicated in the corollary above, when officers are allowed to submit complete preferences over all states, visible fairness pins down serial dictatorship as the unique mechanism. The characterization in Theorem~\ref{thm:CharacterizationVisiblyFairMechanisms} therefore illustrates how relaxing the message space dramatically expands the set of mechanisms that satisfy visible fairness, enabling flexibility in accommodating distributional goals.

\subsection{Two Special Message Spaces}

We now turn to two natural families of message spaces, which are used in practice and serve as the foundation for implementing distributional objectives in \Cref{sec:DistributionalObjectives}.

\subsubsection{Zonal Message Space}

We first focus on a family of message spaces that induce partitions of the set of states into ``zones.''

For each officer \(i\in I\), fix a partition of the set of states \(S\), $Z_i=\{z^i_1,\ldots,z^i_{\ell_i}\}$, where \(\bigcup_{j=1}^{\ell_i} z^i_j=S\) and \(z^i_j\cap z^i_{j'}=\emptyset\) for \(j\neq j'\). We call \(M_i\) a \textbf{zonal message space} for officer \(i\) if it is associated with such a partition \(Z_i\) and, under any \(m_i\in M_i\), two states are comparable if and only if they belong to the same zone in \(Z_i\).

In particular, we focus on zonal message spaces where officer \(i\) can rank the states within each of her zones in any desired way. That is, \(M_i\) has the following two properties:

\textbf{Within-zone completeness}: For every strict total order \(\succ^\ast\) on \(S\), there exists a message \(m_i \in M_i\) such that, for every zone \(z\in Z_i\) and all states \(s,s'\in z\) with \(s\neq s'\),
\[
s \succ^\ast s' \quad\Longleftrightarrow\quad s \succ_{m_i} s'.
\]

\textbf{Across-zone incomparability}: For every pair of distinct zones \(z,z'\in Z_i\), and for every message \(m_i\in M_i\), no state in \(z\) and state in \(z'\) are ever ordered under \(\succ_{m_i}\). That is, the message space never allows officer \(i\) to rank a state in one of her zones relative to a state in another of her zones.

In words, from officer \(i\)'s perspective, submitting a message in a zonal message space amounts to choosing a complete strict ordering over states within each zone in \(Z_i\), while leaving no comparison defined across zones in \(Z_i\).

\begin{example}
Let the set of states be \(S=\{s_1, s_2, s_3, s_4\}\). For a given officer \(i\), fix the following partition:

\[
Z_i=\{z^i_1,z^i_2\},\qquad
z^i_1 = \{s_1, s_2\}, \quad z^i_2 = \{s_3, s_4\}.
\]
Any message in officer \(i\)'s associated zonal message space is of the following form:

\begin{itemize}
    \item Within \(z^i_1\), rank \(s_1\) above \(s_2\) (or vice versa). 
    \item Within \(z^i_2\), rank \(s_3\) above \(s_4\) (or vice versa).
\end{itemize}
\end{example}

Notice that given officer \(i\)'s partition \(Z_i\), officer \(i\)'s preference \(\succ_i\), and a zonal message space \(M_i\) associated with \(Z_i\), there exists a unique truthful message \(m_i\in M_i\): within-zone completeness pins down the ranking of every pair of states that lie in the same zone of \(Z_i\) in agreement with \(\succ_i\), and across-zone incomparability leaves no further freedom.

Suppose each officer \(i\) has a zonal message space associated with an officer-specific partition \(Z_i\).

Consider a zone selection function for officer \(i\), $\mathcal{C}_i:(2^S\setminus\{\emptyset\})\times M\to Z_i$, which is a function such that for any nonempty \(X\subseteq S\) and any \(m\in M\), \(X\cap \mathcal{C}_i(X,m)\neq\emptyset\).\footnote{That is, as long as there are states with spare capacity available, \(\mathcal{C}_i\) must choose one of officer \(i\)'s zones with at least one such state.}

\begin{defi}
\label{def:PartitionPriorityMech}
A mechanism $\psi$ is a \textbf{partitioned priority mechanism} if there exists a zone selection function profile $\left(\mathcal{C}_i\right)_{i\in I}$ such that for any message profile $m$, $\psi(m)$ is the outcome produced by the following procedure:

\begin{itemize}
\item[] \textbf{Step $0$}: Set $S^1= S$.
   \item[] \textbf{Step $k$ ($1\leq k\leq n$)}:  $a_k \in G(S^k, m_k) \cap \mathcal{C}_k(S^k,m)= \{s^k\}$.\footnote{It is easy to see that within each zone, there is a unique $m_k$-maximal element among any set of remaining states with spare capacity, since all states in a zone are comparable under $m_k$.} If the number of officers assigned to $s^k$ reaches $q_{s^k}$, that is $\left|\{\ell\in\{1,\ldots,k\}: a_\ell=s^k\}\right|= q_{s^k}$, then $S^{k+1}\equiv S^{k}\backslash\{s^k\}$. Otherwise, $S^{k+1}= S^{k}$.
\end{itemize} 
\end{defi}

Zone selection functions constitute the essential component of the definition that results in the large variety of these mechanisms. They indicate, for each profile of messages, which zone will be used to determine an officer's outcome. This zone can depend on some exogenous parameter, on the allocations of higher-ranked officers and/or the preferences stated by other officers, as well as their own message. Once a zone is determined, however, the state that will be matched to the officer depends only on the reported preferences between the remaining states in that zone.\footnote{Notice, moreover, that the definition of the mechanism requires that $\mathcal{C}_k$ chooses a zone with states with spare capacity, which by assumption always exists.}

\begin{thm}
\label{thm:CharacterizationIncontestableMechanisms2}
For a zonal message space $M$, $\psi$ is visibly fair if and only if it is a partitioned priority mechanism.
\end{thm}

\subsubsection{Zonal Message Space with Ranking over the Zones}

We now enrich the idea of a zonal message space by allowing officers to report a strict ordering \emph{across} their own zones. As before, officer \(i\)'s message space is associated with a partition $Z_i=\{z^i_1,\ldots,z^i_{\ell_i}\}$ of \(S\). Moreover, when \(X\) is contained in a single zone \(z\in Z_i\), write \(\max(X,m_i)\) and \(\min(X,m_i)\) for the unique maximal and minimal elements of \(X\) under the within-zone strict total order induced by \(m_i\).

A \textbf{zonal message space with ranking over the zones} for officer \(i\) is a message space \(M_i\) in which a message \(m_i\in M_i\) consists of:

\begin{enumerate}
    \item for each zone \(z\in Z_i\), a strict total order over the states in \(z\); and
    \item a strict total order \(\triangleright_{m_i}\) over the set of zones \(Z_i\).
\end{enumerate}

The expressed relation \(\succ_{m_i}\) is defined as the smallest strict partial order on \(S\), with respect to set inclusion, satisfying the following two requirements. First, within each zone \(z\in Z_i\), \(\succ_{m_i}\) coincides with the strict total order reported for that zone. Second, for every pair of distinct zones \(z,z'\in Z_i\),
\[
z\triangleright_{m_i}z'
\quad\Longrightarrow\quad
\max(z,m_i)\succ_{m_i}\min(z',m_i).
\]
No other cross-zone comparison is imposed directly. Thus, any further comparison in \(\succ_{m_i}\) is present only because it is implied by transitivity.

In words, these rankings augment zonal message spaces in a minimal sense. By ranking zone \(z\) above zone \(z'\), officer \(i\) directly states only that the best state in \(z\) is preferred to the worst state in \(z'\). This convention preserves the interpretation that the ranked-higher zone contains at least one state that is visibly better than at least one state in the ranked-lower zone, without imposing a full lexicographic ordering across zones.\footnote{To see the distinction, suppose that \(Z_i=\{z^i_1,z^i_2\}\), \(z^i_1=\{s_1,s_2\}\), and \(z^i_2=\{s_3,s_4\}\). Let \(m_i\) be such that \(s_1\succ_{m_i}s_2\) and \(s_3\succ_{m_i}s_4\). If the defining cross-zone comparison associated with \(z^i_1\triangleright_{m_i}z^i_2\) is \(s_1\succ_{m_i}s_4\), then only the best state in \(z^i_1\) is directly compared with the worst state in \(z^i_2\). Using \(s_4\succ_{m_i}s_1\) as the defining comparison for the same zone ranking would instead make the worst state in \(z^i_2\) dominate the best state in \(z^i_1\), which is inconsistent with the intended interpretation of \(z^i_1\triangleright_{m_i}z^i_2\).}

Notice, moreover, that zonal message spaces with rankings over zones satisfy the truthful-message existence assumption. Given any strict preference \(\succ_i\), officer \(i\) can truthfully rank states within each zone according to \(\succ_i\), and can rank zones according to their \(\succ_i\)-best elements.

\begin{example}
Let the set of states be \(S=\{s_1, s_2, s_3, s_4\}\). For a given officer \(i\), fix the following partition:

\[
Z_i=\{z^i_1,z^i_2\},\qquad
z^i_1 = \{s_1, s_2\}, \quad z^i_2 = \{s_3, s_4\}.
\]

Any message in officer \(i\)'s associated zonal message space with rankings over zones is of the following form:

\begin{itemize}
    \item Within \(z^i_1\), rank \(s_1\) above \(s_2\) (or vice versa).
    \item Within \(z^i_2\), rank \(s_3\) above \(s_4\) (or vice versa).
    \item In addition, choose a ranking over the two zones. This ranking determines the direct cross-zone comparison as follows:
        \begin{itemize}
            \item If \(s_1\succ_{m_i}s_2\) and \(s_3\succ_{m_i}s_4\), then \(z^i_1\triangleright_{m_i}z^i_2\) imposes \(s_1\succ_{m_i}s_4\), whereas \(z^i_2\triangleright_{m_i}z^i_1\) imposes \(s_3\succ_{m_i}s_2\).
            \item If \(s_1\succ_{m_i}s_2\) and \(s_4\succ_{m_i}s_3\), then \(z^i_1\triangleright_{m_i}z^i_2\) imposes \(s_1\succ_{m_i}s_3\), whereas \(z^i_2\triangleright_{m_i}z^i_1\) imposes \(s_4\succ_{m_i}s_2\).
            \item If \(s_2\succ_{m_i}s_1\) and \(s_3\succ_{m_i}s_4\), then \(z^i_1\triangleright_{m_i}z^i_2\) imposes \(s_2\succ_{m_i}s_4\), whereas \(z^i_2\triangleright_{m_i}z^i_1\) imposes \(s_3\succ_{m_i}s_1\).
            \item If \(s_2\succ_{m_i}s_1\) and \(s_4\succ_{m_i}s_3\), then \(z^i_1\triangleright_{m_i}z^i_2\) imposes \(s_2\succ_{m_i}s_3\), whereas \(z^i_2\triangleright_{m_i}z^i_1\) imposes \(s_4\succ_{m_i}s_1\).
        \end{itemize}
\end{itemize}

\end{example}

Suppose each officer \(i\) has a zonal message space with ranking over her zones, associated with partition \(Z_i\).
Define a \textbf{ranked zone selection function} for officer \(i\) to be a mapping $\mathcal{C}_i:\bigl(2^S\setminus\{\emptyset\}\bigr)\times M\to Z_i$ that, for each nonempty subset of states \(X\subseteq S\) and each message profile \(m\), selects a zone \(\mathcal{C}_i(X,m)\in Z_i\) such that:

\begin{enumerate}
    \item \(X \cap \mathcal{C}_i(X,m) \neq \emptyset\), and
    \item either
   \[
     X \cap \mathcal{C}_i(X,m) 
     \neq \bigl\{\min\!\bigl(\mathcal{C}_i(X,m),m_i\bigr)\bigr\},
   \]
   or there is no zone \(z\in Z_i\) with \(z\triangleright_{m_i}\mathcal{C}_i(X,m)\) and \(\max(z,m_i)\in X\).
\end{enumerate}

A ranked zone selection function encodes the restrictions that visible fairness imposes when zones are ranked: if officer \(i\)'s selected zone contains only the least-preferred state in that zone, then no zone that officer \(i\) ranks higher may still have its most-preferred state available.

\begin{defi}\label{def:RankedPartitionPriorityMech}
A mechanism $\psi$ is a \textbf{ranked partitioned priority mechanism} if there exists a ranked zone selection function profile $(\mathcal{C}_i)_{i\in I}$ such that for any message profile $m$, $\psi(m)$ is the outcome produced by the following procedure:

\begin{itemize}
\item[] \textbf{Step $0$}: Set $S^1= S$.
   \item[] \textbf{Step $k$ ($1\leq k\leq n$)}:  $a_k \in G(S^k, m_k) \cap \mathcal{C}_k(S^k,m)= \{s^k\}$.\footnote{Ranked zone selection function makes sure that within the selected zone, a $m_k$-maximal element exists. To see why, let $z = \mathcal{C}_k(S^k, m)$ be the selected zone. We have two cases from the ranked zone selection function definition:

\textbf{Case 1:} $S^k \cap z \neq \emptyset$ and $S^k \cap z \neq \{\min(z, m_k)\}$

Then $S^k \cap z$ contains non-minimal elements of zone $z$.  Since all states in a zone are comparable under $m_k$, the $m_k$-maximal state within $S^k \cap z$ is unique. By the across-zone ranking property, only $\min(z, m_k)$ can be dominated by states in different zones, so the $m_k$-maximal state within $S^k \cap z$ is also $m_k$-maximal in all of $S^k$. Therefore, $G(S^k, m_k) \cap \mathcal{C}_k(S^k,m))$ is non-empty and a singleton.

\textbf{Case 2:} $S^k \cap z \neq \emptyset$ and no zone $z'\in Z_k$ with $z' \triangleright_{m_k} z$ has $\max(z', m_k) \in S^k$

The condition ensures that all maximal elements of higher-ranked zones are unavailable in $S^k$. Therefore, $m_k$-maximal state within $S^k \cap z$ faces no domination from higher-ranked zones and is $m_k$-maximal in $S^k$. Therefore, again $G(S^k, m_k) \cap \mathcal{C}_k(S^k,m)$ is non-empty and a singleton.} If the number of officers assigned to $s^k$ reaches $q_{s^k}$, that is $\left|\{\ell\in\{1,\ldots,k\}: a_\ell=s^k\}\right|= q_{s^k}$, then $S^{k+1}\equiv S^{k}\backslash\{s^k\}$. Otherwise, $S^{k+1}= S^{k}$.
\end{itemize}
\end{defi}

\begin{thm}
\label{thm:CharacterizationIncontestableMechanisms3}
For zonal message space with ranking over zones, $\psi$ is visibly fair if and only if it is a ranked partitioned priority mechanism.
\end{thm}

The presence of ranking over zones implies some restrictions on the zones that the zone selection function can determine, as shown in the example below.

\begin{example}
\label{ex:ZoneRankRestriction}
\frenchspacing
There are three states $S=\{s_1,s_2,s_3\}$ each with capacity $q_{s}=1$ and two
zones  

\[
  z_1=\{s_1\},\qquad z_2=\{s_2,s_3\}.
\]

There are two officers, $i_1$ and $i_2$.
We will consider visibly fair mechanisms in which the message space is zonal with rankings over these zones for both officers.

Officers $i_1$ and $i_2$ both submit the same message:

\medskip
\noindent
\[
  z_1\;\triangleright_{m_i}\;z_2,
  \qquad
  s_2\succ_{m_i}s_3 .
\]

Because $z_2$ contains \emph{two} available states, the ranked‑zone
selection function could, in this scenario, place $i_1$ in \emph{either}
zone $z_1$ or $z_2$. Suppose that it places on $z_2$. Given $m_1$, $i_1$ is matched to $s_2$. 

Now $z_1$ still has $s_1$ free, and $z_2$ only $s_3$. Since $z_1$ remains vacant and is the \emph{highest‑ranked} for $i_2$, the mechanism must choose $z_1$ and assign $s_1$. Selecting $z_2$ would contradict $z_1\triangleright_{m_2}z_2$
while $z_1$ still offers an available seat, and is therefore not allowed.
\end{example}

Ordinary partitioned priority mechanisms (without the zone ranking) would leave the planner free to swap $i_2$ between $z_1$ and $z_2$, illustrating how adding cross‑zone orderings tightens the designer’s hands, in comparison.

\section{Incentives}
\label{sec: Incentives}

Recall from Section~2 that a message $m_i$ is \emph{truthful} for a preference $\succ_i$ whenever every comparison it expresses agrees with $\succ_i$, and that the truthful-message existence assumption guarantees $\mathcal{T}_i(\succ_i)\neq\varnothing$ for every officer.
For a given officer $i\in I$, let $\mathcal{Q}_i$ be the \textbf{set of all preferences} over $S$.

When it comes to incentives, a key desideratum is that an officer who submits a truthful message should never receive a less preferred allocation than if they were to report any other message. 
Formally, a mechanism $\psi$ is \textbf{strategy-proof} if, for every officer $i\in I$, every strict preference $\succ_i\in\mathcal{Q}_i$ with associated weak preference $\succsim_i$, every truthful message $m_i\in\mathcal{T}_i(\succ_i)$, every profile of other officers' messages $m_{-i}\in M_{-i}$, and every alternative message $\hat{m}_i\in M_i$, we have
$$
\psi(m_i,m_{-i})_i \succsim_i \psi(\hat{m}_i,m_{-i})_i.
$$

Note that strategy-proofness implies that if an officer has multiple truthful messages, then they cannot lead to different outcomes. 
We next define two conditions that are sufficient for a visibly fair mechanism to be strategy-proof. Weakening one of the two conditions is also necessary for a visibly fair mechanism to be strategy-proof.

The first condition, expressiveness, requires that whenever an officer changes her message and thereby obtains another assignment, this new assignment must be comparable to the officer’s originally assigned state under the original message. 

\begin{defi}
\label{def:expressiveness}
Let $\psi$ be a mechanism and let $m=(m_i,m_{-i})\in M$ be a message profile. The interest of officer $i$ in state $s \in S$ is \textbf{expressed} under $m$ if $s$ and $\psi(m)_i$ are comparable under $m_i$.

A mechanism $\psi$ satisfies \textbf{expressiveness} if, for every officer $i\in I$, every message profile $m=(m_i,m_{-i})\in M$, and every alternative message $\hat{m}_i\in M_i$, the interest in state $\psi(\hat{m}_i,m_{-i})_i$ is expressed by officer $i$ under $m$. Equivalently, for every such $i$, $m$, and $\hat{m}_i$,
$$
\psi(\hat{m}_i,m_{-i})_i=\psi(m)_i
\quad\text{or}\quad
\psi(\hat{m}_i,m_{-i})_i\succ_{m_i}\psi(m)_i
\quad\text{or}\quad
\psi(m)_i\succ_{m_i}\psi(\hat{m}_i,m_{-i})_i.
$$
\end{defi}

The second condition, availability, requires that whenever an officer changes her message and thereby obtains another assignment, this new assignment must always correspond to a state that was already available to her under the original message. 

\begin{defi}
\label{def:availability}
Let $\psi$ be a mechanism and let $m\in M$ be a message profile. A state $s \in S$ is \textbf{available} to officer $i$ under $m$ if
$$
\left|\{j\in I:\psi(m)_j=s\text{ and }\pi(j)<\pi(i)\}\right|<q_s.
$$

A mechanism $\psi$ satisfies \textbf{availability} if, for every officer $i\in I$, every message profile $m=(m_i,m_{-i})\in M$, and every alternative message $\hat{m}_i\in M_i$, the state $\psi(\hat{m}_i,m_{-i})_i$ is available to officer $i$ under $m$.
\end{defi}

Interestingly, availability is too strong of a condition for ensuring strategy-proofness. Indeed, while availability ensures that an officer cannot manipulate the availability of states by submitting a different message, weak availability only requires that an officer cannot manipulate the availability of weakly preferred states, evaluated at the original message, relative to her assignment under the original message. This condition, together with expressiveness, is necessary for a visibly fair mechanism to be strategy-proof.

\begin{defi}
\label{def:weakavailability}
A mechanism $\psi$ satisfies \textbf{weak availability} if, for every officer $i\in I$, every preference $\succsim_i\in\mathcal{Q}_i$, every truthful message $m_i\in\mathcal{T}_i(\succ_i)$, every profile of other officers' messages $m_{-i}\in M_{-i}$, and every alternative message $\hat{m}_i\in M_i$,
$$
\psi(\hat{m}_i,m_{-i})_i \succsim_i \psi(m_i,m_{-i})_i
\quad\Longrightarrow\quad
\psi(\hat{m}_i,m_{-i})_i \text{ is available to officer } i \text{ under } (m_i,m_{-i}).
$$
\end{defi}

Now we are ready to formally state our main result on incentives.

\begin{thm}
\label{thm:Truthful1} 

A visibly fair mechanism is strategy-proof if and only if it satisfies
expressiveness and weak availability.

\end{thm}

The following corollary is immediate: 
\begin{coro}
\label{coro:Truthful1}
    A visibly fair mechanism is strategy-proof if it satisfies expressiveness and availability. 
\end{coro}

To give some intuition behind the results we give two examples. The first one is a strategy-proof mechanism violating availability --- but satisfying weak availability and expressiveness. 

\begin{example}[Strategy-proof visibly fair mechanism violating availability]
\label{example:Violation availability}
\quad
Consider a problem with two officers $I = \{i_1, i_2\}$. Without loss of generality, we let officer $i_1$ have higher priority than $i_2$, i.e., $\pi(i_1) < \pi(i_2)$. There are two states $S = \{s_1, s_2\}$, each with capacity $q_s = 1$.
Consider a mechanism $\psi$ where $i_1$ can only submit a single message $m_{i_1}$ without any preference information and $i_2$ can either submit message $m_{i_2}: s_1 \succ_{m_{i_2}} s_2$ or $m'_{i_2}: s_2 \succ_{m'_{i_2}}s_1$.
Finally, let $\psi(m_{i_1},m_{i_2})=(s_2,s_1)$ and $\psi(m_{i_1},m'_{i_2})=(s_1,s_2)$.

Clearly, the mechanism is strategy-proof as $i_1$ cannot influence the outcome by submitting a different message, and $i_2$ always gets her top choice when submitting preferences truthfully. 
Moreover, the mechanism is visibly fair as $i_1$ gives no preference information and $i_2$ gets her top choice. 

It is easy to see that availability is violated as, e.g., consider messages $m_{i_2}$ and $m'_{i_2}$. Note that, $m'_{i_2}$ leads to a different outcome $\psi(m_{i_1},m'_{i_2})_{i_2}=s_2 \neq \psi(m)_{i_2}=s_1$ which is not available under $m$. 

Note that weak availability is not violated. Under $m_{i_2}$, which is a truthful message for $s_1 \succ_{i_2} s_2$, we have $\psi(m)_{i_2} \succ_{i_2} \psi(m_{i_1},m'_{i_2})_{i_2}$. Analogously, under $m'_{i_2}$, which is a truthful message for $s_2 \succ_{i_2} s_1$, we have $\psi(m_{i_1},m'_{i_2})_{i_2} \succ_{i_2} \psi(m)_{i_2}$.

Overall this example illustrates how availability is too strong a requirement for strategy-proofness. 
\end{example}\medskip

The second example shows a mechanism satisfying expressiveness but violating weak availability and thus strategy-proofness.

\begin{example}[Expressive but not weakly available visibly fair mechanism]
\label{example:Violation weak availability}
\quad
Consider a problem with two officers $I = \{i_1, i_2\}$, where, without loss of generality, officer $i_1$ has higher priority than $i_2$, i.e., $\pi(i_1) < \pi(i_2)$. There are two states $S = \{s_1, s_2\}$, each with capacity $q_s = 1$.
Consider a mechanism $\psi$ where $i_1$ can only submit a single message $m_{i_1}$ without any preference information and $i_2$ can either submit message $m_{i_2}: s_1 \succ_{m_{i_2}} s_2$ or $m'_{i_2}: s_2 \succ_{m'_{i_2}} s_1$.
Finally, let $\psi(m_{i_1},m_{i_2})=(s_1,s_2)$ and $\psi(m_{i_1},m'_{i_2})=(s_2,s_1)$. 

Clearly, the mechanism is not strategy-proof as $i_2$ always gets his second choice when submitting a truthful message and can get his first choice by simply reporting the opposite message. 
The mechanism is visibly fair as $i_1$ does not give any preference information, and $i_2$ only reports the preferred state, which is given to the higher priority officer $i_1$. 

The mechanism also satisfies expressiveness as under any message $i_2$ gives full preference information, while $i_1$ has a single message automatically satisfying expressiveness. 
At the same time the mechanism violates weak availability as e.g. under $m_{i_2}$  which is a truthful message for $s_1 \succ_{i_2} s_2$ and another message $m'_{i_2}$, where $\psi(m_{i_1},m'_{i_2})_{i_2} \succsim_{i_2} \psi(m)_{i_2} $ we have that $\psi(m_{i_1},m'_{i_2})_{i_2}$ is not available under $m$.

\end{example}

\subsection{Two special message spaces revisited}

In this section we show that, neither the class of partition priority mechanism nor the ranked partition priority mechanisms are always strategy-proof.

In \Cref{example:Violation s-p pp} we show that mechanisms using zonal message spaces can fail weak availability and expressiveness, and therefore not be strategy-proof.
Moreover, in \Cref{example:Violation s-p rpp} we show that zonal message spaces with ranking over zones can be sufficient conditions for visibly fair mechanisms to fail expressiveness, and therefore strategy-proofness.

\begin{example}[Partitioned‐priority can violate strategy-proofness]
\label{example:Violation s-p pp}
\quad
Consider a problem with three officers $I = \{i_1, i_2, i_3\}$. There are three states $S = \{s_1, s_2, s_3\}$, each with capacity $q_s = 1$. The zonal message space for all officers consists of two zones: $z_1 = \{s_1, s_2\}$ and $z_2 = \{s_3\}$.
Suppose each officer’s true preference is $s_1 \succ_i s_2 \succ_i s_3$. For each officer $i$, let $m_i$ be the message that ranks $s_1 \succ_{m_i} s_2$, and let $m_i'$ reverse it: $s_2 \succ_{m_i'} s_1$.

Consider a zone selection function $\mathcal{C}$ that includes the following:

\[
\begin{aligned}
  &\mathcal{C}_{i_1}(S,\,m) \;=\; z_1, 
    &&\mathcal{C}_{i_1}(S,\,(m_{-i_3},m_{i_3}')) \;=\; z_1,\\
  &\mathcal{C}_{i_2}(\{s_2,s_3\},\,m) \;=\; z_1,
    &&\mathcal{C}_{i_2}(\{s_2,s_3\},\,(m_{-i_3},m_{i_3}')) \;=\; z_2,\\
  &\mathcal{C}_{i_3}(\{s_3\},\,m) \;=\; z_2,
    &&\mathcal{C}_{i_3}(\{s_2\},\,(m_{-i_3},m_{i_3}')) \;=\; z_1.
\end{aligned}
\]
\smallskip

Notice that these conditions are consistent with zone selection function that induces a partitioned priority mechanism. Consider the following successful manipulation for $i_3$: Under truthful report $m$, the allocation is $(s_1, s_2, s_3)$; but when $i_3$ flips her internal ranking $(m_{i_3}\!\to\!m_{i_3}')$, the mechanism assigns $(s_1, s_3, s_2)$, giving $i_3$ a preferred assignment as $s_2 \succ_{i_3} s_3$.  

Weak availability is violated. Under the truthful message $m_{i_3}$, officer $i_3$ is assigned $s_3$. By deviating to $m_{i_3}'$, the officer is assigned $s_2$, and $s_2\succ_{i_3}s_3$. However, $s_2$ is not available to $i_3$ under the original message profile $m$, because the higher-priority officer $i_2$ is assigned to $s_2$ under $m$.

Similarly, expressiveness is violated as e.g. officer $i_3$ does not express any preference information regarding $s_2$ and $s_3$, even though $s_2$ is allocated to $i_3$ under message $m_{i_3}'$ and $s_3$ is allocated to $i_3$ under message $m_{i_3}$.

\end{example}\medskip

Finally, we give an example showing that non-strategy-proofness can result from the message space itself: with a zonal message space that also elicits a ranking over zones, no visibly fair mechanism using that message space is strategy-proof.

\begin{example}[Ranked‐partitioned priority can violate Expressiveness]
\label{example:Violation s-p rpp}
\quad

Consider a setting with three officers $I = \{i_1, i_2, i_3\}$. The set of states is $S = \{s_1, s_2, s_3\}$, each with a capacity of $q_s = 1$.
The message space is zonal with rankings over zones for all officers, with two zones: $z_1 = \{s_1, s_2\}$ and $z_2 = \{s_3\}$. Each officer's message must provide a full ranking over the states within zone $z_1$, i.e., between $s_1$ and $s_2$, and additionally indicate whether $z_1 \triangleright_{m_i} z_2$ or $z_2 \triangleright_{m_i} z_1$.

\noindent
\emph{True preferences for $i_1$:} $s_1 \succ_{i_1} s_2 \succ_{i_1} s_3$.  
This yields a \emph{unique} truthful message $m^\rhd_{i_1}: (s_1 \succ s_2 \rhd_{} s_3)$,\footnote{A message $m^{\rhd}_{i_1}: s_1 \succ s_2 \rhd s_3$ is our abbreviation for the preference message with $s_1\succ_{m_{i_1}} s_2$ and $z_1 \rhd_{m_{i_1}} z_2$.} ensuring $i_1$ obtains $s_1$ in any visibly fair mechanism.

\smallskip
\smallskip
\smallskip

\noindent
\emph{Two preferences for $i_2$:}
\begin{enumerate}
\item First, suppose $i_2$ has $s_3 \succ_{i_2} s_2 \succ_{i_2} s_1$.  
The unique truthful message $m_{i_2}: (s_3\,\rhd\, s_2 \succ s_1)$  forces $i_2$ to end up with $s_3$ under any visibly fair, strategy‐proof mechanism. To see this, suppose that $i_2$ is assigned $s_2$ instead. Then, by submitting $m^{\rhd \prime}_{i_2}: s_3 \rhd s_1 \succ s_2$, any visibly fair mechanism must assign $s_3$ to $i_2$, leading to a successful manipulation. 

\item Next, consider $\succ_{i_2}: \,s_2 \succ s_3 \succ s_1$.   Here, \emph{two} truthful messages are possible.  One message 
$m_{i_2}: (s_3 \rhd s_2 \succ s_1)$ again, as we have just argued, assigns $s_3$ to $i_2$.  
Another message 
$m^{\rhd \prime \prime}_{i_2}: (s_2 \succ s_1 \rhd s_3)$ must assign $s_2$ to $i_2$.
\end{enumerate}

\noindent

Here, expressiveness is violated as e.g. officer $i_2$ does not express any preference information regarding $s_2$ and $s_3$ given message $m_{i_2}$, even though $s_2$ is allocated to $i_2$ under message $m^{\rhd \prime \prime}_{i_2}$ and $s_3$ is allocated to $i_2$ under message $m_{i_2}$.
On the other hand, weak availability does not pose a problem in the above example. 

\end{example}\medskip

Given \Cref{thm:Truthful1}, combined with \Cref{example:Violation s-p pp} and \Cref{example:Violation s-p rpp} the following corollary is immediate: 

\begin{coro}
\label{coro: s-p}
Consider the zonal message spaces with and without rankings. Both partitioned priority mechanisms and ranked partitioned priority mechanisms may fail strategy-proofness.
\end{coro}

\section{Achieving Distributional Objectives}
\label{sec:DistributionalObjectives}

When using direct mechanisms, our earlier discussion shows that only serial dictatorship (SD) achieves visible fairness. By eliciting less information about preferences—through carefully designed restricted message spaces—we expand the set of allocations that can be deemed visibly fair for a given problem. This relaxation provides the policy maker with additional flexibility, allowing for the implementation of a broader array of allocation rules that still adhere to this notion of fairness.

There are, in principle, many distinct distributional objectives that can be accommodated within this broader framework. In this section, we introduce one family of such objectives, which we denote \emph{Modular Upper-Bounds}, and give complete instructions on how to design mechanisms that are visibly fair, strategy-proof, and respect these bounds. Modular upper-bounds model distributional objectives by imposing limits on the number of officers of certain types assigned to specific subsets of states. The following section, \Cref{sec:efficiency}, examines why preference elicitation remains valuable despite these restrictions and analyzes the efficiency consequences of the message spaces required to guarantee the designer's objectives.

\subsection{Modular Upper-Bounds}

We extend our original model by saying that each officer $i$ has a type $t$ from a finite set of types $T$, where $t_i$ denotes the type of officer $i$. We will also use $t_k$ to refer to the type of officer $i_k$. The distributional goals of the designer are modeled through a collection of type-specific upper-bounds: for each relevant set of types and subset of states, the designer may impose a cap on how many officers of those types can be assigned there.

\begin{defi}
\label{def:ModularUpperBounds}  
A \textbf{modular upper‑bound system} is a finite collection
\[
H \;=\;\bigl\{\,(\Xi_h,\;S_h,\;k_h)\bigr\},
\]
where for each element $h\in H$:\footnote{In our notation, $H$ contains a set of upper-bounds, each of which represented by the letter $h$. In this context, $\Xi_h$ for example, is the first component of $h$.}  

\begin{itemize}
    \item $\varnothing\neq \Xi_h\subseteq T$ is the set of types covered by the quota,  
    \item $S_h\subseteq S$ is a subset of states, and  
    \item $k_h\in\mathbb Z_{\ge 0}$ is the cap.
\end{itemize}

\noindent For every type $t\in T$ we write  
\[
H^{t}\;:=\;\bigl\{\,(\Xi_h,S_h,k_h)\in H : t\in\Xi_h\bigr\}.
\]

\noindent For every state $s \in S$ and type $t\in T$, we write its \textbf{upper-bound signature} as
\[
    H^{s,t}=\{\,(\Xi_h,S_h,k_h)\in H : s\in S_h,\;t\in\Xi_h\}.
\]
\end{defi}

\begin{defi}
\label{def:respectsModularUpperBounds}
An allocation $a \in \mathcal{A}$ \textbf{respects the upper-bounds} $H$ if, for every $h \in H$:
\[
    \left|\{\,i\in I:\;a_i\in S_h,\;t_i\in\Xi_h\}\right| \;\le\; k_h.
\]   
\end{defi}

We say that the modular upper-bound $h$ is \textbf{binding} at allocation $a \in \mathcal{A}$ if the constraint is satisfied with equality. 

Since in our model officers cannot be left unmatched, we need to guarantee that these upper-bounds are compatible with that restriction while using visibly fair mechanisms. Formally, hereafter we will restrict our attention to modular upper-bound systems that satisfy the following property of \textit{sequential solvency}.

Given a subset $J \subseteq I$, a \textbf{partial allocation} on $J$ is a list $a^J = (a_j)_{j \in J}$ specifying a state $a_j \in S$ for each $j \in J$. We say $a^J$ is \textbf{feasible} if $\bigl|\{\,j \in J : a_j = s\,\}\bigr| \le q_s$ for every $s \in S$, and that it \textbf{respects} the modular upper-bound system $H$ if $\bigl|\{\,j \in J : a_j \in S_h,\ t_j \in \Xi_h\,\}\bigr| \le k_h$ for every $h \in H$. We write $\mathcal{A}^J$ for the set of feasible partial allocations on $J$.

\begin{defi}\label{def:seq-solvency}
A modular upper-bound system $H$ satisfies \textbf{sequential solvency} if, for every officer $i \in I$, every subset $J \subseteq I \setminus \{i\}$, and every partial allocation $a^J = (a_j)_{j \in J} \in \mathcal{A}^J$ that respects $H$, there exists a state $s^* \in S$ such that
\begin{enumerate}
    \item[(i)] $\bigl|\{\,j \in J : a_j = s^*\,\}\bigr| < q_{s^*}$, and
    \item[(ii)] $\bigl|\{\,j \in J : a_j \in S_h,\ t_j \in \Xi_h\,\}\bigr| < k_h$ for every $h \in H^{s^*,\, t_i}$.
\end{enumerate}
\end{defi}

Starting from any feasible, upper-bound-respecting partial assignment of the other officers, officer $i$ can always be placed at some state $s^*$ without violating capacity or any modular upper-bound relevant to her type.\footnote{The empty partial assignment ($J = \emptyset$) is included in the quantifier, so the condition has bite even when no full respecting allocation exists.} The intuition for the definition above is simple. Regardless of which capacities or collection of upper-bounds bind, there will always be a compatible state for every remaining officer. That is, while matching officers one at a time, modular upper-bounds can restrict \textit{where} officers are matched, but \textit{not whether} they are matched. Since it relies on the particular number of agents of each type in $I$, it allows for interesting and practical constraints, as we will show in examples that will follow. In \Cref{app:seq-solvency} we verify sequential solvency for all the examples in the paper that rely on modular upper-bounds.\footnote{Notice that it is crucial that the definition depends on the profile of types of officers. Otherwise, the upper-bounds would have to be satisfied when all agents have the same type, making only bounds that never bind compatible with not leaving officers unmatched.}

\begin{example}
Consider a problem with five states \(S=\{s_1,s_2,s_3,s_4,s_5\}\), each having capacity \(q_s=1\). There are two officer types: \(t_1\) and \(t_2\).
    \begin{itemize}

        \item For type \(t_1\), the upper-bound system is 
        \[
        H^{t_1}=\Bigl\{ \bigl(\{t_1\},\{s_1,s_2,s_3\},\,2\bigr) \Bigr\},
        \]
        meaning that at most 2 type-\(t_1\) officers may be assigned to states in \(\{s_1,s_2,s_3\}\).
        \item For type \(t_2\), the upper-bound system is 
        \[
        H^{t_2}=\Bigl\{ \bigl(\{t_2\},\{s_3,s_4,s_5\},\,1\bigr) \Bigr\},
        \]
        meaning that at most 1 type-\(t_2\) officer may be assigned to states in \(\{s_3,s_4,s_5\}\).
    \end{itemize}
\end{example}\medskip

The literature has proposed several ways to formalize “quota–type” constraints.\footnote{An incomplete list includes \cite{echeniqueHowControlControlled2015}, \cite{kamadaEfficientMatchingDistributional2015}, \cite{gotoDesigningMatchingMechanisms2017}, \cite{kamadaStabilityStrategyproofnessMatching2018}, \cite{azizMatchingDiversityConstraints2020}, \cite{kojimaJobMatchingConstraints2020}, and \cite{kamadaFairMatchingConstraints2024}.} An example of a very \emph{permissive} notion is the hereditary family of \citet{gotoDesigningMatchingMechanisms2017}: write a matching as a vector that counts, for every state, how many officers of each type are assigned there; a subset of vectors is \emph{hereditary} when it is closed under coordinate-wise decrements. Any system of pure caps clearly has this property, so every modular upper-bound instance fits inside the hereditary domain. A \emph{tighter} specification is the hierarchical (laminar) system analysed by \citep{kamadaEfficientMatchingDistributional2015,kamadaStabilityStrategyproofnessMatching2018}: here the subsets that carry quotas must form a tree—any two are either disjoint or one contains the other. Laminar caps are useful when the policy maker wants, say, regional caps that line up neatly with district caps, but they rule out overlapping constraints such as “no more than ten officers in the Northeastern states \emph{and} no more than eight in the coastal states.” All laminar systems are modular, yet the converse is false.

\subsection{Constraint-Induced Message Spaces}

To design mechanisms that respect modular upper-bounds, we define \emph{constraint-induced message spaces}. The idea is that the message space is itself derived from the distributional constraints. The modular upper-bounds determine which states must be treated as jointly relevant for the policy objective, and these states are grouped into common zones. The resulting message space therefore elicits precisely the coarse comparisons across states that are relevant for respecting the upper-bounds.

\begin{defi}
\label{def:constraintInducedMessageSpaces}
For any type $t\in T$, define an equivalence relation $\sim_{t}$ on $S$ such that, for all $s, s' \in S$,
\[
    s \sim_{t} s' \quad \text{if and only if} \quad H^{s,t} = H^{s',t}.
\]

The \textbf{Constraint-Induced Message Space} associated with the modular upper-bounds $H^t$ is the zonal message space $M_{t}$ with zones $Z = \{ z^t_1, z^t_2, \dots, z^t_\ell \}$, where each zone $z^t_j$ is an equivalence class under $\sim_{t}$, that is,
\[
    z^t_j = \{ s \in S : s \sim_{t} s_j \},
\]
for some representative state $s_j \in S$.
\end{defi}

\noindent This construction ensures that:
\begin{enumerate}[i)]
    \item Zones are disjoint and partition the set of states: $\bigcup_{j} z^t_j = S$ and $z^t_j \cap z^t_{j'} = \emptyset$ for $j \neq j'$.
    \item All states within the same zone are involved in exactly the same set of upper-bounds for the type $t$. Therefore, if some upper-bound is binding for some state in a zone, then it binds for all states in that zone.
\end{enumerate}

\subsection{Modular Priority Mechanism}

Each type \(t \in T\) induces a \emph{constraint-induced message space} \(M_{t}\), as in Definition~\ref{def:constraintInducedMessageSpaces}, where for each type $t\in T$, \(S\) is partitioned into zones $Z^t=\{z^t_1,z^t_2,\ldots\}$ according to equivalence classes of states under the same upper-bound signature. Thus, for an officer \(i\) of type \(t_i\), the mechanism \emph{offers} the zonal space \(M_{t_i}\), with associated zones $z^{t_i}_1,z^{t_i}_2\ldots$, requiring her to submit a message \(m_i\in M_{t_i}\). That message ranks all states within each zone but cannot compare states across different zones.

\medskip

 As part of the mechanism design, each officer \(i\) is also assigned an \emph{exogenous} ranking
\[
  z^{t_i}_1 \,\blacktriangleright_i\, z^{t_i}_2 \,\blacktriangleright_i\,\cdots\,\blacktriangleright_i\,\ldots 
\]

over the same zones, independent of the message \(m_i\). These exogenous rankings---which can encode policy priorities such as emphasizing certain zones first or last---do not depend on agents' reports. 

\begin{defi}\label{def:modularPriorityMechanism}
Given a sequentially solvent modular upper-bound system $H$, an exogenous zone ranking $\blacktriangleright_i$ for each officer $i$, and a profile of messages $m=(m_i)_{i\in I}$, each $m_i$ in the constraint-induced message space $M_{t_i}$, the \textbf{Modular Priority Mechanism} $\psi$ proceeds as follows.

\smallskip
\noindent\textbf{Initialization.}
\begin{itemize}
  \item For each $s\in S$, set the remaining capacity
        $q_{s}^{\mathrm{rem}}=q_{s}$.
  \item For each type $t\in T$ and each zone $z_{\ell}^{t}$, set a flag
        $B_{\ell}^{t}=\textsc{False}$.
  \item Set $a_{i}=\varnothing$ for every $i\in I$.
\end{itemize}

\smallskip
\noindent\textbf{Upper-bound status update.}
At any stage of the procedure, given the current partial allocation $a$
and the current flags $(B_{\ell}^{t})$, update the flags as follows. For
each $h\in H$, let
\[
N_{h}(a)=
\bigl|\{\,i\in I:a_{i}\in S_{h}\text{ and } t_{i}\in\Xi_{h}\,\}\bigr|.
\]
If $N_{h}(a)=k_{h}$, then for every $t\in\Xi_{h}$ and every $\ell$ such
that $S_{h}\cap z_{\ell}^{t}\neq\varnothing$, set
$B_{\ell}^{t}=\textsc{True}$.

\smallskip
\noindent

Apply the upper-bound status update once before the sequential assignment
begins, so that any upper-bound $h\in H$ with $k_{h}=0$ is registered as
binding.

\smallskip
\noindent\textbf{Sequential assignment.}
Process the officers in the order $(i_{1},\dots,i_{n})$. For each
$k=1,\dots,n$:
\begin{enumerate}
  \item Let $t_{{k}}$ be the type of officer $i_{k}$, and let
        $z_{1}^{t_{{k}}},\dots,z_{K}^{t_{{k}}}$ be the zones
        induced by $m_{i_{k}}$.
  \item Following the exogenous ranking
        $\blacktriangleright_{{k}}$, let
        $z_{\ell^{\ast}}^{t_{{k}}}$ be the highest-ranked zone such
        that
        \begin{itemize}
          \item $B_{\ell^{\ast}}^{t_{{k}}}=\textsc{False}$, and
          \item there exists $s\in z_{\ell^{\ast}}^{t_{{k}}}$ with
                $q_{s}^{\mathrm{rem}}>0$.\footnote{Existence follows by applying sequential solvency to officer $i_k$ and to the partial allocation $(a_{i_1},\dots,a_{i_{k-1}})$ of already assigned officers.}
        \end{itemize}
        
  \item Let $s^{\ast}$ be the $m_{i_{k}}$-most-preferred state in
        $z_{\ell^{\ast}}^{t_{{k}}}$ with
        $q_{s^{\ast}}^{\mathrm{rem}}>0$. Set $a_{i_{k}}=s^{\ast}$ and
        decrement
        $q_{s^{\ast}}^{\mathrm{rem}}\leftarrow
         q_{s^{\ast}}^{\mathrm{rem}}-1$.
  \item Apply the upper-bound status update.
\end{enumerate}

\smallskip
\noindent\textbf{Outcome.}
After processing $i_{1},\dots,i_{n}$, the mechanism returns the
allocation $a=(a_{i})_{i\in I}$.
\end{defi}
\medskip

The Modular Priority Mechanism is, therefore, a partitioned priority mechanism in which each officer type is associated with a zonal message space of states that share the same upper-bound constraints. An officer’s final assignment is determined by two key factors: (1) a counter that tracks remaining capacity for the relevant states, and (2) a flag indicating whether any upper-bound restrictions in that zone have become binding. Because all states in a given zone are governed by the same set of constraints, a single triggered bound applies uniformly across the entire zone. Below, we present the main result for this mechanism and two examples that illustrate how it operates in practice.

\begin{thm}\label{modulartheorem}
    Suppose that $H$ is a sequentially solvent modular upper-bound system. Then the Modular Priority Mechanism is visibly fair, strategy-proof, and respects upper-bounds.
\end{thm}

\medskip

\begin{center}
\captionsetup{font=footnotesize}
\begin{tikzpicture}[
  font=\small,
  regbox/.style 2 args={
    rounded corners=8pt, draw=#1, line width=0.6pt,
    fill=#2, minimum width=3.5cm, minimum height=3.5cm
  },
  rural/.style={
    circle, draw=teal!60, fill=teal!12,
    line width=0.5pt, minimum size=9mm, font=\footnotesize
  },
  urban/.style={
    rectangle, rounded corners=2pt, draw=violet!40!gray, fill=violet!8,
    line width=0.5pt, minimum size=9mm, font=\footnotesize
  },
  annot/.style={font=\scriptsize, text=gray!60!black},
  rlabel/.style={font=\small\bfseries}
]

\node[regbox={red!35!gray}{red!5}]     (R1) at (0,0)   {};
\node[regbox={blue!35!gray}{blue!5}]   (R2) at (4.8,0) {};
\node[regbox={green!40!gray}{green!5}] (R3) at (9.6,0) {};

\node[rlabel] at (0,2.05)   {$R_1$};
\node[rlabel] at (4.8,2.05) {$R_2$};
\node[rlabel] at (9.6,2.05) {$R_3$};

\node[rural] (s1) at (0,0.55)      {$s_1$};
\node[urban] (s2) at (-0.65,-0.5)  {$s_2$};
\node[urban] (s3) at (0.65,-0.5)   {$s_3$};

\node[rural] (s4) at (4.8,0.55)    {$s_4$};
\node[urban] (s5) at (4.15,-0.5)   {$s_5$};
\node[urban] (s6) at (5.45,-0.5)   {$s_6$};

\node[rural] (s7) at (9.6,0.55)    {$s_7$};
\node[urban] (s8) at (8.95,-0.5)   {$s_8$};
\node[urban] (s9) at (10.25,-0.5)  {$s_9$};

\begin{scope}[on background layer]
  \node[fit=(s2)(s3)(s5)(s6)(s8)(s9),
    inner sep=4pt, rounded corners=5pt,
    draw=violet!35, line width=0.55pt, densely dashed,
    fill=violet!4] (Uband) {};
\end{scope}

\node[annot] at (0,-2.1)   {$9$ doctors, type $t{=}1$};
\node[annot] at (4.8,-2.1) {$9$ doctors, type $t{=}2$};
\node[annot] at (9.6,-2.1) {$9$ doctors, type $t{=}3$};

\node[rural,  minimum size=4mm, inner sep=0pt] at (2.8,-2.85) {};
\node[annot, anchor=west] at (3.1,-2.85) {Rural};
\node[urban, minimum size=4mm, inner sep=0pt] at (5.3,-2.85) {};
\node[annot, anchor=west] at (5.6,-2.85) {Urban};

\end{tikzpicture}
\captionof{figure}{Structure of \Cref{ex:example2}. Nine hospitals ($S$) across three regions, each containing one rural and two urban hospitals, with $U=\{s_2,s_3,s_5,s_6,s_8,s_9\}$. Modular upper-bounds impose a regional cap $(\{t\}, R_t, 6)$ limiting local doctors in their home region and an urban cap $(\{1,2,3\}, U, 19)$ across all urban hospitals.}
\label{fig:doctorRegions}
\end{center}

\begin{example}[Distributing doctors across regions and urban/rural divides]\label{ex:example2}
Consider nine hospitals $S = \{s_1,s_2,s_3,s_4,s_5,s_6,s_7,s_8,s_9\}$ partitioned into regions (see \Cref{fig:doctorRegions}) 
\[
R_1 = \{s_1,s_2,s_3\},\quad R_2 = \{s_4,s_5,s_6\},\quad R_3 = \{s_7,s_8,s_9\},
\]
with each hospital having capacity \(q_s=4\).  In each region, the first hospital in each region is \emph{rural} (\(s_1,s_4,s_7\)) and the others are \emph{urban} (\(s_2,s_3,s_5,s_6,s_8,s_9\)).  There are 27 doctors, 9 of each type \(t\in\{1,2,3\}\), indicating their home region.  

Two types of modular upper‐bounds apply:
\[
H^{t} = \{(\{t\},R_{t},6)\}
\quad\text{for each }t=1,2,3,
\quad\text{and}\quad
H^{\mathrm{U}} = \{(\{1,2,3\},U,19)\}
\]
where 
\(
U=\{s_2,s_3,s_5,s_6,s_8,s_9\}
\)
is the set of all urban hospitals, and \((U,19)\) is a universal cap across all doctors.  

\medskip

A doctor of type~1, for instance, faces the following \emph{signatures}:
\[
H^{s,1}=
\begin{cases}
\{(\{1\},R_1,6)\}, & \text{if $s\in R_1$ is rural (here, $s_1$)},\\
\{(\{1\},R_1,6),(\{1,2,3\},U,19)\}, & \text{if $s\in R_1$ is urban (here, $s_2,s_3$)},\\
\{(\{1,2,3\},U,19)\}, & \text{if $s$ is urban but not in $R_1$ (e.g.\ $s_5,s_6,s_8,s_9$)},\\
\varnothing, & \text{if $s$ is rural outside $R_1$ (e.g.\ $s_4,s_7$).}
\end{cases}
\]
Hence the constraint-induced partition for a type-1 doctor is:
\[
z^1_1=\{s_1\},\quad
z^1_2=\{s_2,s_3\},\quad
z^1_3=\{s_4,s_7\},\quad
z^1_4=\{s_5,s_6,s_8,s_9\}.
\]
Within each zone, the doctor ranks hospitals \emph{fully}, yet makes no cross‐zone comparisons. Message spaces following the same principle are constructed for the remaining types.

\emph{Exogenous zone ranking:} each doctor of type~\(1\) has the fixed ordering 
\[
z^1_2 \;\blacktriangleright_{i}\; z^1_1\blacktriangleright_{i}\; z^1_4\blacktriangleright_{i}\; z^1_3,
\]
so that own-region hospitals are ranked before hospitals outside the home region, and within each of these two groups urban hospitals are ranked before rural hospitals. Exogenous zone rankings for other types of doctors are defined analogously.

A modular--priority mechanism enforcing these caps proceeds as in Definition~\ref{def:modularPriorityMechanism}, using the zonal message spaces defined above and the exogenous rankings over zones $\left(\blacktriangleright_{i}\right)_{i\in I}$: each type-\(t\) doctor is assigned to her top feasible and non-binding hospital within the partition induced by \((R_t,6)\) and \((U,19)\).
\end{example}\medskip

As noted previously, the specification of the exogenous rankings over zones \(\left(\blacktriangleright_{i}\right)_{i\in I}\) does not affect the theoretical properties established in our results. However, if we know that, for instance, most doctors typically prefer urban hospitals over rural ones, defining the rankings as in the examples above tends to improve efficiency of the final allocations, without harming the other objectives. Naturally, if the actual preferences of participants substantially diverge from such assumptions, employing this approach could lead to less desirable outcomes.

It is also worth noting that the mechanisms introduced above---and the overarching notion of visible fairness---are ill-suited to implementing \emph{within}-state affirmative action quotas. Such policies, which often appear in the matching literature on diversity (e.g.\ majority quotas or type-specific caps in each state), could be encoded as modular upper-bounds with one state in each upper-bound. However, a scenario in which, for example, each state can admit a maximum number of agents of a certain type, the resulting zonal message space would have one state per zone. This, of course, would eliminate a role for preferences in the allocation. If we attempted to compensate for this by using zonal message spaces with ranking over zones, we would easily conclude that no visibly fair mechanism with these characteristics would respect upper bounds.

\section{Why Preference Elicitation Matters}
\label{sec:efficiency}

The previous section shows how a designer can use restricted message spaces to implement distributional objectives while preserving visible fairness and strategy-proofness. At this point, however, one might ask why the designer should elicit preferences at all. If the designer has an objective that conflicts with ordinary serial dictatorship, and if the designer can restrict the comparisons that participants are allowed to reveal, why not elicit no preference information and simply implement whichever allocation the designer prefers?

The answer is that preference elicitation is not merely a constraint on the designer. It is also the instrument through which the allocation can be made responsive to participants. Standard market-design models already give the designer substantial discretion: the designer chooses the mechanism, and therefore chooses which allocation is produced at each preference profile. A mechanism with no preference reports can still use priorities, capacities, types, or administrative objectives. What it cannot do is distinguish between two economies that differ only in what participants want. Once preference information disappears, so does the possibility of making the allocation depend on that information.

This conditionality matters for two reasons. First, it is often part of the institutional purpose of the mechanism. Even when the designer has distributional objectives, the resulting assignment is usually expected to bear some relation to the claims made by participants. In our framework, this relation is represented by a family of possible outcomes indexed by the messages participants submit. If no preference information is elicited, that family collapses: the mechanism can only select among allocations using non-preference information. Second, when using visibly fair mechanisms, preference information has a direct welfare value. Conditional on satisfying the designer's distributional objectives, one allocation may generate substantially more welfare than another. A mechanism cannot systematically select such allocations unless it asks participants for at least some of the preference comparisons that distinguish them.

Visible fairness clarifies the informational nature of this trade-off. By \Cref{thm:CharacterizationVisiblyFairMechanisms}, visibly fair mechanisms are sparse versions of serial dictatorship: officers are processed by priority and receive states that are maximal (in the submitted partial order) among the states that remain available according to the preferences they are allowed to express. At one extreme, when all comparisons are elicited, \Cref{cor:SDUniqueVisiblyFairDM} implies that visible fairness pins down ordinary serial dictatorship, which is Pareto efficient. At the other extreme, if no preference comparisons are elicited, outcomes are necessarily unrelated to preference-based welfare. Between these extremes lies the design problem studied in this paper.

The relationship, however, is not monotone. Expanding the message space changes both the information available to the mechanism and the visible-fairness constraints the mechanism must satisfy. The following example shows that merging zones can make most officers worse off, even though the merged message space contains strictly more preference information.

\begin{example}[More elicitation need not benefit every officer]
\label{ex:zone-merge-worse-off}
Fix $n\geq 3$, officers $I=\{i_1,\ldots,i_n\}$, states $S=\{s_1,\ldots,s_n\}$ (unit capacities), and priority $i_1,i_2,\ldots,i_n$. Start with two zones $z_1=\{s_1\}$ and $z_2=\{s_2,\ldots,s_n\}$, and a partitioned priority rule that assigns $i_1$ through $z_1$ and all others through $z_2$, always taking the best available state in the selected zone. Let preferences satisfy the following. Officer $i_1$ has $s_2$ as her most-preferred state and $s_2\succ_{i_1}s_1$. For each $k=2,\ldots,n-1$, officer $i_k$ ranks $s_k$ first and $s_{k+1}$ second, with $s_k\succ_{i_k}s_{k+1}\succ_{i_k}s$ for every $s\in S\setminus\{s_k,s_{k+1}\}$. Officer $i_n$ ranks $s_n$ first and $s_1$ second, with $s_n\succ_{i_n}s_1\succ_{i_n}s$ for every $s\in S\setminus\{s_n,s_1\}$. All remaining comparisons can be completed arbitrarily. Then the two-zone outcome is $a=(s_1,s_2,\ldots,s_n)$. After merging the zones, all states become comparable and visible fairness implies serial dictatorship, yielding $a'=(s_2,s_3,\ldots,s_n,s_1)$. Hence $i_1$ is better off while every $i_2,\ldots,i_n$ is worse off, so the number of officers who are worse off can be made arbitrarily large by increasing $n$.
\end{example}

\Cref{ex:zone-merge-worse-off} should not be read as a comparative-static theorem against elicitation. The two mechanisms are defined on different message spaces, and there is no canonical way to compare all visibly fair mechanisms before and after a change in the partition. A different zone-selection rule in the original message space would define a different baseline. The example instead makes a narrower point: more elicitation changes the allocation problem, and the gains from information need not accrue officer by officer. The relevant question is therefore not whether every individual benefits from every refinement, but whether additional preference information tends to improve the allocation landscape available to the designer.

\subsection{Simulations: the average value of merging zones}
\label{sec:preference-elicitation-simulations}

To assess the average effect of eliciting more comparisons, we simulate a family of zonal priority mechanisms. Each simulated market has $n=m=50$ officers and states, with unit capacities. We use cardinal utilities, which are generated according to

\[
    u_{is}=\sqrt{p_1}\,c_s+\sqrt{1-p_1}\,\varepsilon_{is},
\]
where $c_s$ and $\varepsilon_{is}$ are independent standard normal draws.\footnote{The square-root coefficients keep $\mathrm{Var}(u_{is})=1$ for every $p_1$, while delivering $\mathrm{Corr}(u_{is},u_{js})=p_1$ for any two officers $i\neq j$ at a given state $s$.} Thus $p_1\in[0,1]$ measures the commonality in officers' evaluations of states: when $p_1=0$, utilities are independent across officers; when $p_1=1$, all officers have the same cardinal utility vector.

For each utility profile, we draw multiple priority orders and compare two visibly fair mechanisms: a baseline zonal priority mechanism and an otherwise analogous mechanism whose message space is enriched by random zone mergers. The richer mechanism is obtained from the same initial partition and zone-order construction, so differences in welfare isolate the informational effect of eliciting additional within-merged-zone comparisons.\footnote{Full details on the partition generation, zone-order construction, merger procedure, and implementation of both mechanisms are provided in \Cref{app:simulations}.}

For each run, we compute four welfare levels. Let $W^0$ be total welfare under the original zonal message space. Let $W^{p_2}$ be total welfare after pairwise zone mergers with probability $p_2$. Let $W^1$ be total welfare under the full-merger benchmark, where all states lie in one zone and the priority rule is ordinary serial dictatorship with full preference rankings. Finally, let $W^{FB}$ be the first-best welfare, i.e., the maximum total welfare over all one-to-one assignments, ignoring the priority order. 

The main figure reports two normalized welfare measures. The first measures what percentage of the welfare gain between the baseline and the full-merger benchmark is already achieved at merger probability $p_2$. The second measures what percentage of the welfare gap between the baseline and the unconstrained first-best benchmark is closed at merger probability $p_2$. Both measures are computed cell by cell and then averaged over values of $p_1<1$. The case $p_1=1$ is excluded from these averages because all officers have identical utilities; with $n=m$ and unit capacities, every complete assignment has the same total welfare, so the denominators are zero up to simulation noise.

\begin{center}
\begin{minipage}{0.86\textwidth}
\centering
\includegraphics[width=\textwidth]{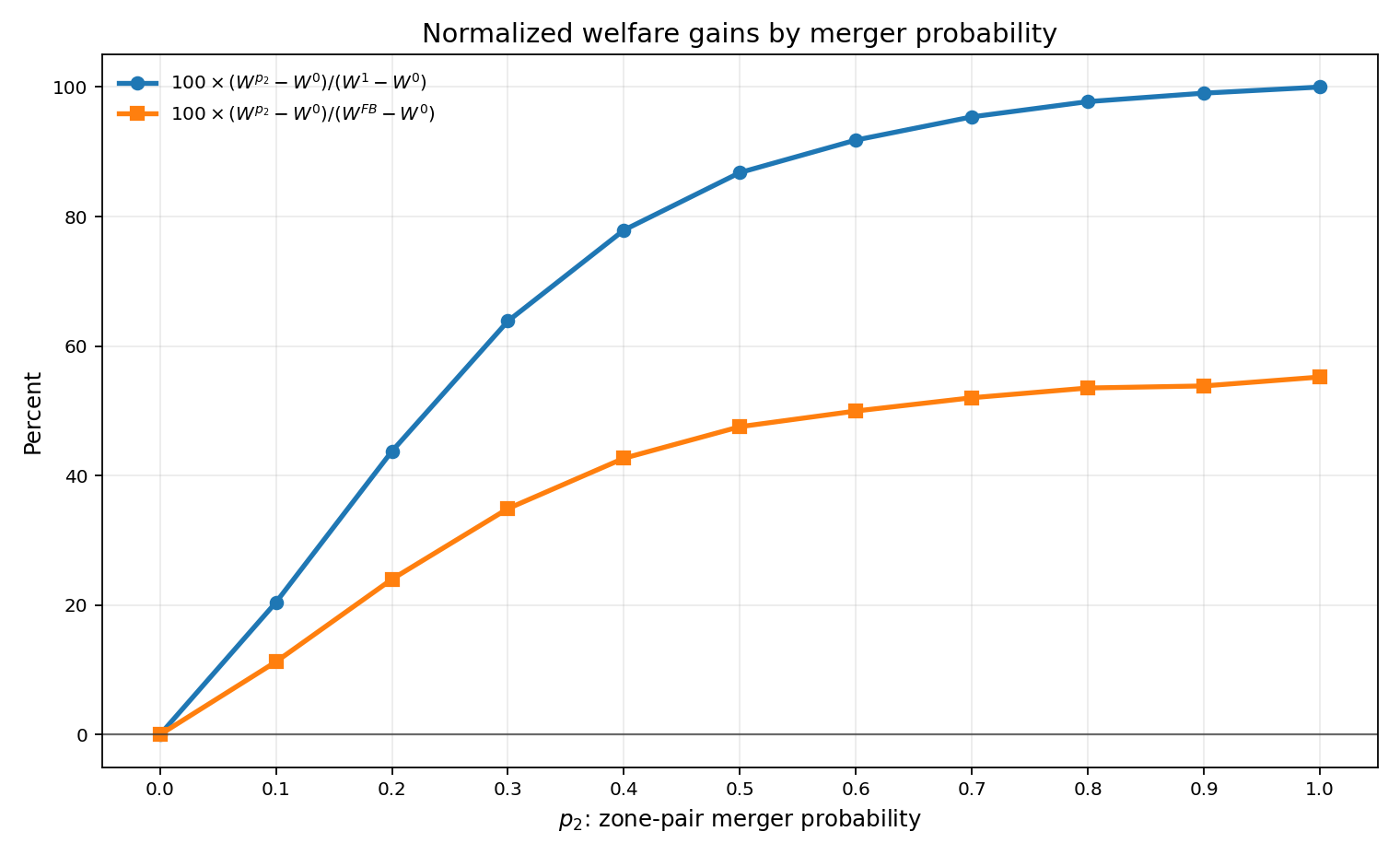}
\captionsetup{type=figure,font=footnotesize}
\captionof{figure}{Normalized welfare gains by merger probability. The blue series reports the average share of the full-merger gain recovered, $100\times(W^{p_2}-W^0)/(W^1-W^0)$. The orange series reports the average share of the first-best gap closed, $100\times(W^{p_2}-W^0)/(W^{FB}-W^0)$. Averages are taken over the nondegenerate values $p_1<1$.}
\label{fig:simulation-main-normalized}
\end{minipage}
\end{center}

\Cref{fig:simulation-main-normalized} shows two distinct margins. The orange curve is lower than the blue curve because even complete merging does not implement the first-best assignment. When $p_2=1$, all states are comparable and the mechanism becomes serial dictatorship under the random priority order, but serial dictatorship still respects priority rather than maximizing total welfare. Thus the remaining gap to $W^{FB}$ is the welfare cost of using a priority-respecting mechanism. At the same time, the figure shows that average welfare under the post-merger mechanism rises relative to the original zonal mechanism over the simulated range. The main takeaway is that, for a fixed policy objective motivating a visibly fair restriction of the message space, eliciting additional comparisons can deliver sizable welfare gains on average, even though the gain need not be monotone in every realization or beneficial for every officer.

The simulation therefore separates two losses. The gap $W^1-W^0$ is the informational loss from coarse zones: it is reduced by merging zones and eliciting more within-zone comparisons. The gap $W^{FB}-W^1$ is the residual loss from using a priority-respecting serial dictatorship rather than the first-best assignment. Visible fairness can make the first loss smaller by allowing a richer message space, but it does not remove the second loss.

\subsection{The limit: constrained Pareto efficiency}
\label{sec:constrained-efficiency-limit}

As the following example demonstrates, the Modular Priority Mechanism does not guarantee efficiency.

\begin{example}\label{ex:motivation_formal}
Consider two officers $I = \{i_1, i_2\}$ with priority $\pi(i_1)< \pi(i_2)$ and of the same type $t$, and two states $S = \{s_1, s_2\}$ with capacities $q_{s_1} = 2$ and $q_{s_2} = 1$. Let there be a modular upper-bound of one officer of type $t$ allowed in $s_1$. Thus, the constraint-induced message space partitions states into two zones: $z_1 = \{s_1\}$ and $z_2 = \{s_2\}$. Let the exogenous ranking of zones be $z_1 \blacktriangleright_i z_2$ for both officers.

Suppose that officers' true preferences are as follows:
\[
\begin{aligned}
\succ_{i_1}: \quad & s_2 \succ s_1, \\
\succ_{i_2}: \quad & s_1 \succ s_2.
\end{aligned}
\]

Under the Modular Priority Mechanism, the resulting allocation is
\[
a=(a_{i_1},a_{i_2})=(s_1,s_2).
\]
However, the allocation
\[
a'=(a'_{i_1},a'_{i_2})=(s_2,s_1)
\]
strictly Pareto dominates $a$: officer $i_1$ strictly prefers $s_2$ to
$s_1$, and officer $i_2$ strictly prefers $s_1$ to $s_2$. Moreover,
$a'$ respects the upper-bound, since exactly one type-$t$ officer
is assigned to $s_1$.
\end{example}\medskip

\medskip

Example~\ref{ex:motivation_formal} shows that the Modular Priority Mechanism can fail to identify a cap-respecting Pareto improvement. In this realized instance, the efficient cap-respecting allocation would require using officer $i_1$'s comparison between $s_1$ and $s_2$. Ex post, eliciting that comparison is harmless: assigning $i_1$ to either state can be made consistent with the upper-bound. The difficulty is ex ante. A static message space must be chosen before the realized preference profile is known, and comparisons that are harmless in one profile may create visible-fairness conflicts with the upper-bound in another.

We now introduce a ``second-best'' notion of efficiency tailored to settings with modular upper-bounds.

\begin{defi}\label{def:ConstrainedPE}
An allocation \(a\) is \textbf{constrained Pareto efficient} if:
\begin{enumerate}[i)]
    \item \(a\) respects the upper--bounds, and
    
    \item there is no other allocation $a' \in \mathcal{A}$ that respects the upper-bounds such that $a'_i \succsim_i a_i$ for every officer $i \in I$, with strict preference for at least one officer (which is automatic under strict preferences whenever $a' \neq a$).
\end{enumerate}
 A \textbf{mechanism} is constrained Pareto efficient if, for any problem, when agents submit any \emph{truthful message}, the outcome is constrained Pareto efficient.
\end{defi}

In other words, an allocation is {constrained Pareto efficient} if there is no alternative allocation that respects both capacity and modular upper-bound constraints and strictly Pareto improves upon $a$ (with respect to the agents' true preferences). A mechanism is constrained Pareto efficient if, whenever agents submit truthful messages, the resulting allocation is constrained Pareto efficient \emph{under their full (true) preferences}. The following result demonstrates that, in general, no mechanism can simultaneously achieve constrained Pareto efficiency and visible fairness.

\begin{thm}\label{thm:NoMechanismAllThree}
There exists a sequentially solvent modular upper-bound system $H$ for which no mechanism is simultaneously visibly fair, respects $H$ at every message profile, and is constrained Pareto efficient.
\end{thm}

The negative result in Theorem~\ref{thm:NoMechanismAllThree} is not, of course, universal. Standard SD satisfies those properties for modular upper-bound constraints that never bind. On the other hand, constant mechanisms can trivially respect constraints that result in a single feasible allocation which, by its uniqueness, would be constrained Pareto efficient.

Theorem~\ref{thm:NoMechanismAllThree}, however, underscores a key tension in designing visibly fair mechanisms that must also adhere to non-trivial constraints such as modular upper-bounds. At its core, this tension arises from the message-space design: enforcing both fairness and quota requirements demands that the space be curated to preempt scenarios where certain preference reports would necessarily violate upper-bounds. Consequently, to avoid such violations, the mechanism must collect more restricted preference information than might otherwise be desirable. This reduction in elicited information, while upholding fairness and preserving the bounds, can lead to efficiency losses because the mechanism may lack the information needed to detect and implement mutually beneficial reallocations that remain compliant with all constraints.\footnote{One way to resolve this tension is to condition each officer's message space on the allocations of higher-priority officers. Appendix~\ref{sec:appendix-efficiency} presents a dynamic sequential dictatorship that elicits each officer's message after earlier assignments have been fixed and satisfies the corresponding dynamic versions of visible fairness and strategy-proofness, while also being constrained Pareto efficient and respecting upper-bounds (\Cref{dynamicModPriorityMechVisiblyFairModularAndSP}).}

\section{The Indian Administrative Service}\label{sec:IAS}

Officers of the Indian Administrative Service (IAS) are assigned to state cadres. Candidates are strictly ranked according to their performance in the Civil Services Examination, yielding a single priority ordering over all candidates.

Prior to 2017, the assignment mechanism was essentially a serial dictatorship with some modifications. As documented by \cite{thakurMatchingCivilService2023}, this produced a pattern of \emph{homophily}: officers from northern states were disproportionately assigned to northern states, and similarly for southern states. This geographic clustering was deemed incompatible with the national integration mandate of the All India Services. In response, the government reformed the mechanism in 2017, introducing a zone-based allocation system that we now describe formally.\footnote{Throughout this section, we abstract away from two features of the actual cadre allocation policy. First, the policy includes an extensive system of affirmative action in the form of reserved seats for candidates belonging to certain social categories (Scheduled Castes, Scheduled Tribes, Other Backward Classes, and others). Second, each state maintains a distinction between insider slots (reserved for candidates whose home state coincides with that state) and outsider slots. We set aside reservations because they serve an objective that is conceptually distinct from the ones our framework addresses: reservations aim to determine the composition of the matched population—ensuring that candidates from specific demographic groups reach a minimum threshold in the allocation. The insider/outsider system, by contrast, is a distributional constraint of the kind our framework is designed to handle, and we incorporate it in our analysis in \Cref{sec:IASmodular}. For the purpose of the present discussion, however, we omit it to isolate the cross-region distributional objective—the central motivation behind the 2017 reform—and to evaluate whether visible fairness and modular upper bounds provide a suitable foundation for it. The interested reader is referred to \cite{thakurMatchingCivilService2023} for a detailed discussion of both features.}

\subsection{The IAS Mechanism}

We instantiate the model of \Cref{sec: Model and Definitions} to the IAS setting: the states $S$ correspond to the states being allocated, and the priority $\pi$ is determined by performance on the Civil Services Examination. The states are partitioned into five geographically contiguous zones $Z = \{z_1, z_2, z_3, z_4, z_5\}$; let $z(s)$ denote the zone containing state $s$.

\paragraph{Home state.} Each officer $i \in I$ has a \emph{home state} $h_i \in S$, determined by where they are from.

\paragraph{Message space.} Each officer's message space is a \emph{zonal message space with ranking over the zones}, as defined in Section~3.2.2. Concretely, each officer $i$ submits a message $m_i$ consisting of:
\begin{enumerate}[(i)]
    \item a strict ranking of all states within each zone $z_j\in Z$, and
    \item a strict ranking $\triangleright_{m_i}$ over the five zones $Z$.
\end{enumerate}

Within each zone, all states are fully comparable. Across zones, the zone ranking directly imposes only the top-of-higher-zone versus bottom-of-lower-zone comparisons described in \Cref{sec: Visibly Fair Mechanisms}; all remaining comparisons in $\succ_{m_i}$ are those forced by transitivity.\footnote{See Figure~\ref{fig:iasRank} for an example.}

\medskip

Given the message profile, the mechanism proceeds as follows.

\paragraph{Interleaved serial dictatorship.} All officers are processed one at a time in decreasing priority order. For each officer $i$, the mechanism determines an assignment using an \emph{interleaving} rule that cycles through the zones in the order $\triangleright_{m_i}$, considering in each ``round''~$r$ the $r$-th most preferred state in each zone. Formally, let the zones be labeled $z^1_i,z^2_i,\ldots,z^5_i$ according to $\triangleright_{m_i}$, and within each zone $z^j_i$, let $s^j_1,s^j_2,\ldots$ be the states listed in decreasing order of preference under $m_i$. The mechanism finds the first state with remaining capacity in the sequence
\[
s^1_1,\; s^2_1,\; s^3_1,\; s^4_1,\; s^5_1,\; s^1_2,\; s^2_2,\; s^3_2,\; s^4_2,\; s^5_2,\; s^1_3,\; \ldots
\]
and assigns officer $i$ to that state. That state's remaining capacity is then decremented, and the mechanism proceeds to the next officer.

We illustrate the IAS mechanism with a small instance, highlighting how interleaving serial dictatorship promotes geographic spread across zones.

\begin{example}\label{ex:IAS}
There are two zones $z_1=\{s_1,s_2\}$ and $z_2=\{s_3,s_4\}$, each state having capacity $q_s=1$. There are four officers $i_1,i_2,i_3,i_4$ (listed in priority order). Their messages are represented below, where each block is a zone, states are ranked vertically within each block (most preferred on top), and $\triangleright$ denotes the ranking over zones:

\medskip

\begin{center}
\begin{tikzpicture}[font=\small]


\node[font=\normalsize] at (-1.2, 1.1) {$i_1$:};
\draw[thick] (0,0) rectangle (1.8,2.2);
\draw[thick] (0,1.7) -- (1.8,1.7);
\node at (0.9,1.9) {\textbf{$z_1$}};
\node at (0.9,1.15) {$s_1$};
\node at (0.9,0.45) {$s_2$};
\node at (2.4,1.1) {\Large $\triangleright$};
\draw[thick] (3.0,0) rectangle (4.8,2.2);
\draw[thick] (3.0,1.7) -- (4.8,1.7);
\node at (3.9,1.9) {\textbf{$z_2$}};
\node at (3.9,1.15) {$s_3$};
\node at (3.9,0.45) {$s_4$};

\node[font=\normalsize] at (6.8, 1.1) {$i_2$:};
\draw[thick] (8.0,0) rectangle (9.8,2.2);
\draw[thick] (8.0,1.7) -- (9.8,1.7);
\node at (8.9,1.9) {\textbf{$z_2$}};
\node at (8.9,1.15) {$s_3$};
\node at (8.9,0.45) {$s_4$};
\node at (10.4,1.1) {\Large $\triangleright$};
\draw[thick] (11.0,0) rectangle (12.8,2.2);
\draw[thick] (11.0,1.7) -- (12.8,1.7);
\node at (11.9,1.9) {\textbf{$z_1$}};
\node at (11.9,1.15) {$s_2$};
\node at (11.9,0.45) {$s_1$};


\node[font=\normalsize] at (-1.2, -1.9) {$i_3$:};
\draw[thick] (0,-3) rectangle (1.8,-0.8);
\draw[thick] (0,-1.3) -- (1.8,-1.3);
\node at (0.9,-1.1) {\textbf{$z_1$}};
\node at (0.9,-1.85) {$s_1$};
\node at (0.9,-2.55) {$s_2$};
\node at (2.4,-1.9) {\Large $\triangleright$};
\draw[thick] (3.0,-3) rectangle (4.8,-0.8);
\draw[thick] (3.0,-1.3) -- (4.8,-1.3);
\node at (3.9,-1.1) {\textbf{$z_2$}};
\node at (3.9,-1.85) {$s_4$};
\node at (3.9,-2.55) {$s_3$};

\node[font=\normalsize] at (6.8, -1.9) {$i_4$:};
\draw[thick] (8.0,-3) rectangle (9.8,-0.8);
\draw[thick] (8.0,-1.3) -- (9.8,-1.3);
\node at (8.9,-1.1) {\textbf{$z_1$}};
\node at (8.9,-1.85) {$s_1$};
\node at (8.9,-2.55) {$s_2$};
\node at (10.4,-1.9) {\Large $\triangleright$};
\draw[thick] (11.0,-3) rectangle (12.8,-0.8);
\draw[thick] (11.0,-1.3) -- (12.8,-1.3);
\node at (11.9,-1.1) {\textbf{$z_2$}};
\node at (11.9,-1.85) {$s_3$};
\node at (11.9,-2.55) {$s_4$};

\end{tikzpicture}
\end{center}

\medskip

\noindent\textbf{Officer $i_1$.} The interleaving sequence under $m_{i_1}$ is: $s_1$ (1st in $z_1$), $s_3$ (1st in $z_2$), $s_2$ (2nd in $z_1$), $s_4$ (2nd in $z_2$). The first available state is $s_1$, so $\mathcal{C}_1^{\textup{IAS}}(S^1, m) = z_1$ and $i_1$ is assigned to $s_1$.

\noindent\textbf{Officer $i_2$.} The interleaving under $m_{i_2}$ is: $s_3$ (1st in $z_2$), $s_2$ (1st in $z_1$), $s_4$ (2nd in $z_2$), $s_1$ (2nd in $z_1$). With $s_1$ taken, the first available state is $s_3$, so $\mathcal{C}_2^{\textup{IAS}}(S^2, m) = z_2$ and $i_2$ is assigned to $s_3$.

\noindent\textbf{Officer $i_3$.} The interleaving under $m_{i_3}$ is: $s_1$ (1st in $z_1$), $s_4$ (1st in $z_2$), $s_2$ (2nd in $z_1$), $s_3$ (2nd in $z_2$). With $s_1$ and $s_3$ taken, the first available state is $s_4$, so $\mathcal{C}_3^{\textup{IAS}}(S^3, m) = z_2$ and $i_3$ is assigned to $s_4$.

\noindent\textbf{Officer $i_4$.} The interleaving under $m_{i_4}$ is: $s_1, s_3, s_2, s_4$. With $s_1$, $s_3$, and $s_4$ taken, the first available is $s_2$, so $\mathcal{C}_4^{\textup{IAS}}(S^4, m) = z_1$ and $i_4$ is assigned to $s_2$.

\medskip

The final allocation is:
\[
i_1\to s_1,\quad i_2\to s_3,\quad i_3\to s_4,\quad i_4\to s_2.
\]

\noindent Note the role of the interleaving zone selection rule: by cycling across zones rather than exhausting one zone before moving to the next, the mechanism steers officers toward their top choices \emph{in different zones}, thereby promoting geographic dispersion. Without interleaving, a ranked partitioned priority mechanism that always selects the highest-ranked zone with available states would assign $i_3$ to $z_1$ (since $z_1\triangleright z_2$ and $s_2$ remains available), yielding  $i_3\to s_2$ and $i_4\to s_4$. The interleaving rule instead pushes $i_3$ into $z_2$, achieving greater geographic spread.
\end{example}

\subsection{Visible Fairness and Incentives}
We now examine two further properties of the IAS mechanism. First, we show that the
mechanism is visibly fair. Second, we show that the interleaving rule is not strategy-proof. 

\begin{prop}\label{prop:IASVisiblyFair}
The mechanism induced by the interleaved serial dictatorship procedure is a ranked partitioned priority mechanism and is therefore visibly fair.
\end{prop}

The key step is to formally define the interleaving procedure as a ranked zone selection function and verify that it satisfies the two conditions in the definition.\footnote{The proof is given in \Cref{prop:IASVisiblyFairproof}.}

\begin{example}[Failure of strategy-proofness]\label{ex:IASstrategyproof}
Consider two zones $z_1=\{s_1,s_2\}$ and $z_2=\{s_3,s_4\}$, and each state having capacity $q_s=1$. There are three officers $i_1, i_2, i_3$ (in priority order) with messages:
\begin{align*}
i_1 &: \quad z_1 \triangleright z_2, \quad z_1\!: s_1 > s_2, \quad z_2\!: s_3 > s_4, \\
i_2 &: \quad z_2 \triangleright z_1, \quad z_2\!: s_3 > s_4, \quad z_1\!: s_1 > s_2.
\end{align*}
The interleaving assigns $i_1 \to s_1$ and $i_2 \to s_3$.

Suppose officer $i_3$ has true preferences $s_1 \succ s_2 \succ s_4 \succ s_3$. Since every state in $z_1$ is preferred to every state in $z_2$, the unique truthful message is
\[
m_{i_3}: \quad z_1 \triangleright z_2, \quad z_1\!: s_1 > s_2, \quad z_2\!: s_4 > s_3.
\]
The interleaving sequence under $m_{i_3}$ is $s_1, s_4, s_2, s_3$. With $s_1$ and $s_3$ taken, the first available state is $s_4$, so $i_3$ is assigned to~$s_4$.

Now consider the deviation $\hat{m}_{i_3}$: $z_1 \triangleright z_2$, $z_1\!: s_1 > s_2$, $z_2\!: s_3 > s_4$---which reverses the within-zone ranking in~$z_2$. The interleaving sequence becomes $s_1, s_3, s_2, s_4$. With $s_1$ and $s_3$ taken, the first available state is $s_2$, and $i_3$ is assigned to~$s_2$. Since $s_2 \succ s_4$, the deviation is profitable and the mechanism is not strategy-proof.

The failure is a direct consequence of the mechanism not satisfying expressiveness (\Cref{def:expressiveness}). Under the truthful message $m_{i_3}$, the states $s_4 = \psi(m)_{i_3}$ and $s_2 = \psi(\hat{m}_{i_3}, m_{-i_3})_{i_3}$ lie in different zones ($s_4 \in z_2$, $s_2 \in z_1$). The only cross-zone comparison under $m_{i_3}$ is $\max(z_1, m_{i_3}) = s_1 \succ_{m_{i_3}} \min(z_2, m_{i_3}) = s_3$, which links $s_1$ to $s_3$ but not $s_2$ to~$s_4$. Since $s_2$ and $s_4$ are incomparable under $m_{i_3}$, officer $i_3$'s interest in the state obtained through deviation is not expressed under her truthful message. By \Cref{thm:Truthful1}, this violation of expressiveness is exactly what permits the profitable manipulation.
\end{example}

\subsection{Distributional Objectives of the IAS}\label{sec:IASmodular}

The 2017 reform replaced a near-serial-dictatorship with a zone-based mechanism in order to reduce geographic clustering and promote national integration. However, the precise distributional objectives that motivated the design of the five-zone partition were never formally articulated. In particular, it is unclear whether the policymaker sought to limit same-zone assignments, to balance origin–destination flows more broadly, or to target proximity to the home state specifically. Each of these readings leads to a different modular upper-bound system and, through the theory developed in Section~\ref{sec:DistributionalObjectives}, to a different induced message space and mechanism.

We now present three such specifications. Each is grounded in a plausible interpretation of what the IAS designers had in mind, and each illustrates a distinct feature of the visibly fair and modular upper-bounds framework: minimal departures from serial dictatorship, endogenous recovery of the actual zone partition, and the capacity to encode graduated proximity concerns through nested constraints. Throughout this subsection, we write $q_{z} := \sum_{s \in z} q_s$ for the total capacity of zone $z$.

\subsubsection{Specification 1: Zonal Homophily Caps}\label{sec:IASspec1}

The most direct reading of the 2017 reform is that the government sought to limit the number of officers serving in their own geographic zone. This captures the documented concern about homophily---officers clustering near their home region---while imposing the lightest possible constraint on the allocation.

\paragraph{Types and upper-bounds.} Let each officer's type be the zone containing her home state, so $T = \{z_1,\ldots,z_5\}$ and $t_i = z(h_i)$. Define the modular upper-bound system
\[
H = \bigl\{\, (\{z_j\},\, z_j,\, k_j) : j = 1,\ldots,5 \bigr\},
\]
where $k_j$ caps the number of officers from zone~$j$ who may be assigned to states in zone~$j$. A natural calibration is $k_j = \lfloor \alpha \cdot q_{z_j} \rfloor$ for some policy parameter $\alpha \in (0,1)$; for instance, $\alpha = 0.5$ imposes a 50\% insider cap in every zone.

\paragraph{Induced message spaces.} For an officer of type~$z_j$, the upper-bound signature is
\[
H^{s,z_j} = \begin{cases} \{(\{z_j\}, z_j, k_j)\} & \text{if } s \in z_j, \\ \varnothing & \text{if } s \notin z_j. \end{cases}
\]
Thus, the constraint-induced partition for every type consists of exactly two zones:
\[
z^{z_j}_1 = z_j \quad\text{(home zone)}, \qquad z^{z_j}_2 = S \setminus z_j \quad\text{(away)}.
\]
Each officer ranks states fully within her home zone and fully within the complement, but cannot compare states across the two.

\paragraph{Exogenous zone ranking.} A natural choice is $z^{z_j}_1 \;\blacktriangleright_i\; z^{z_j}_2$ for each officer~$i$ of type~$z_j$, giving priority to the home zone. By \Cref{modulartheorem}, the resulting Modular Priority Mechanism is visibly fair, strategy-proof, and respects the insider caps.

\begin{example}[Zonal homophily caps]\label{ex:IASspec1}
Consider a simplified setting with four states $S=\{s_1,s_2,s_3,s_4\}$ partitioned into $z_1=\{s_1,s_2\}$ and $z_2=\{s_3,s_4\}$, each state having capacity $q_s=2$. There are 6 officers: $i_1,\ldots,i_6$ (in priority order), with $t_{1}=t_{2}=t_{3}=z_1$ and $t_{4}=t_{5}=t_{6}=z_2$.  The insider caps are $k_1 = k_2 = 2$, i.e., at most 2 officers from each zone may serve in their home zone (50\% of the zone's total capacity of 4).

The induced partition for a type-$z_1$ officer is $z^{z_1}_1=\{s_1,s_2\}$ and $z^{z_1}_2=\{s_3,s_4\}$; symmetrically for type~$z_2$. Each officer ranks states within each of these two groups and, with the exogenous ranking placing the home zone first, the Modular Priority Mechanism assigns officers in priority order: first to their preferred state in the home zone (if the insider cap has not yet bound), then to their preferred state in the away zone once it binds.

Suppose all type-$z_1$ officers prefer $s_1$ to $s_2$ and $s_3$ to $s_4$ within their respective zones, while all type-$z_2$ officers prefer $s_3$ to $s_4$ and $s_1$ to $s_2$. Then: $i_1 \to s_1$, $i_2\to s_1$ (filling $s_1$'s capacity), and $i_3$ faces a binding insider cap ($k_1=2$ is reached), so $i_3$ is redirected to the away zone and assigned $s_3$. Similarly, $i_4\to s_3$ (filling $s_3$'s capacity), $i_5\to s_4$, and $i_6$ faces the binding insider cap and is redirected to the away zone, receiving~$s_2$. The final allocation is
\[
i_1\to s_1,\quad i_2\to s_1,\quad i_3\to s_3,\quad i_4\to s_3,\quad i_5\to s_4,\quad i_6\to s_2,
\]
distributing officers evenly across zones despite their preference for home states.
\end{example}

This specification is the most parsimonious departure from serial dictatorship: only one constraint per officer type, yielding the coarsest possible non-trivial partition (two zones).

\subsubsection{Specification 2: Cross-Zone Balance Constraints}\label{sec:IASspec2}

A more ambitious interpretation of the ``national integration'' mandate is that the government wanted not only to limit insider assignments but to ensure \emph{balanced representation from each origin zone in each destination zone}. Under this reading, no destination zone should be dominated by officers from any single origin---a concern about cultural and administrative diversity that goes beyond the insider/outsider distinction.

\paragraph{Types and upper-bounds.} Types are again $T = \{z_1,\ldots,z_5\}$ (origin zone). For every pair of origin zone~$z_j$ and destination zone~$z_\ell$, impose an upper-bound:
\[
H = \bigl\{\, (\{z_j\},\, z_\ell,\, k_{j\ell}) : j,\ell \in \{1,\ldots,5\} \bigr\}.
\]
A natural calibration sets:
\begin{itemize}
    \item $k_{jj} = \lfloor \alpha \cdot q_{z_j} \rfloor$ for insiders (as in Specification~1), and
    \item $k_{j\ell} = \lceil q_{z_\ell}/4 \rceil$ for $j \neq \ell$, so that each outside zone supplies at most roughly one quarter of any destination zone's officers.
\end{itemize}

\paragraph{Induced message spaces.} For an officer of type~$z_j$, the upper-bound signature at state~$s$ depends only on which zone~$z_\ell$ contains~$s$:
\[
H^{s,z_j} = \bigl\{(\{z_j\}, z_\ell, k_{j\ell})\bigr\}, \quad\text{where } s\in z_\ell.
\]
Since all states within the same destination zone share the same signature, the constraint-induced partition for every officer type is precisely
\[
Z = \{z_1, z_2, z_3, z_4, z_5\}
\]
---the five geographic zones used in the actual IAS mechanism.

\paragraph{Exogenous zone ranking.} Since the distributional goal is balanced representation across \emph{all} zone pairs, a natural choice is a rotating order that shifts the starting zone by type. Fix a reference ordering of the zones, say $z_1, z_2, z_3, z_4, z_5$, and for each officer~$i$ of type~$z_j$ let the exogenous ranking cycle through the zones beginning at~$z_j$:
\[
z_j \;\blacktriangleright_i\; z_{j+1} \;\blacktriangleright_i\; z_{j+2} \;\blacktriangleright_i\; z_{j+3} \;\blacktriangleright_i\; z_{j+4},
\]
where indices are modulo~5. Thus officers from different home zones face different exogenous rankings: a type-$z_1$ officer has $z_1 \blacktriangleright z_2 \blacktriangleright z_3 \blacktriangleright z_4 \blacktriangleright z_5$, a type-$z_2$ officer has $z_2 \blacktriangleright z_3 \blacktriangleright z_4 \blacktriangleright z_5 \blacktriangleright z_1$, and so on. Each type's demand is directed first toward its home zone but then rotates through the remaining zones in a balanced fashion, echoing the interleaving logic of the actual IAS mechanism and spreading excess demand evenly across destinations.

This specification \emph{endogenously recovers} the actual five-zone message space as the constraint-induced partition. The 25 upper-bounds ($5\times 5$) each constrains a different type on a different subset of states---yet the framework handles them seamlessly (see \Cref{fig:crossZoneBalance}). By \Cref{modulartheorem}, the Modular Priority Mechanism under these bounds is visibly fair, strategy-proof, and respects all 25 constraints.

\begin{center}
\captionsetup{font=footnotesize}
\begin{tikzpicture}[
  font=\small,
  zone/.style={
    circle, minimum size=13mm, line width=0.5pt, inner sep=0pt
  },
  bgline/.style={draw=gray!20, line width=0.35pt},
  capline/.style={->, >=stealth, draw=red!35, line width=0.55pt,
    shorten >=2pt, shorten <=2pt},
  annot/.style={font=\scriptsize, text=gray!60!black},
  fml/.style={font=\scriptsize, text=red!40!black}
]

\node[zone, draw=red!40, fill=red!6]       (z1) at (90:2.3)  {$z_1$};
\node[zone, draw=blue!40, fill=blue!6]     (z2) at (162:2.3) {$z_2$};
\node[zone, draw=green!45, fill=green!6]   (z3) at (234:2.3) {$z_3$};
\node[zone, draw=violet!40, fill=violet!6] (z4) at (306:2.3) {$z_4$};
\node[zone, draw=teal!45, fill=teal!6]     (z5) at (18:2.3)  {$z_5$};

\draw[bgline] (z2) -- (z3);
\draw[bgline] (z2) -- (z4);
\draw[bgline] (z2) -- (z5);
\draw[bgline] (z3) -- (z4);
\draw[bgline] (z3) -- (z5);
\draw[bgline] (z4) -- (z5);

\draw[bgline, ->, >=stealth] (z2) to[out=185, in=139, looseness=4] (z2);
\draw[bgline, ->, >=stealth] (z3) to[out=257, in=211, looseness=4] (z3);
\draw[bgline, ->, >=stealth] (z4) to[out=329, in=283, looseness=4] (z4);
\draw[bgline, ->, >=stealth] (z5) to[out=41, in=-5, looseness=4] (z5);

\draw[capline] (z1) -- (z2);
\draw[capline] (z1) -- (z3);
\draw[capline] (z1) -- (z4);
\draw[capline] (z1) -- (z5);

\draw[capline, draw=orange!50] (z1) to[out=113, in=67, looseness=4] (z1);

\node[font=\scriptsize, text=orange!65!black, anchor=south]
  at ($(z1)+(0,1.15)$) {Insider cap: $k_{jj}$};
\node[font=\scriptsize, text=red!45!black, anchor=west]
  at ($(z1)!0.5!(z5)+(0.55,0.15)$) {Outsider cap: $k_{j\ell}$};

\node[annot, text width=8cm, align=center] at (0.8,-3.4) {%
  $5\times 5 = 25$ upper-bounds: $(\{z_j\},\, z_\ell,\, k_{j\ell})$ 
  for all $j,\ell \in \{1,\ldots,5\}$\\[2pt]
  Induced partition for every type: $Z = \{z_1, z_2, z_3, z_4, z_5\}$};

\end{tikzpicture}
\captionof{figure}{Cross-zone balance constraints (\Cref{sec:IASspec2}).  Each origin zone faces caps on assignments to \emph{every} destination zone, yielding $5\times 5 = 25$ non-laminar upper-bounds. Highlighted arrows illustrate the constraints faced by type-$z_1$ officers; light connections indicate the analogous structure for all other types. The induced partition for every type is the full five-zone partition $\{z_1,\ldots,z_5\}$.}
\label{fig:crossZoneBalance}
\end{center}

\begin{example}[Cross-zone balance]\label{ex:IASspec2}
Consider six states $S=\{s_1,\ldots,s_6\}$ partitioned into three zones: $z_1=\{s_1,s_2\}$, $z_2=\{s_3,s_4\}$, and $z_3=\{s_5,s_6\}$, each state with capacity $q_s=2$ (so $q_{z_1}=q_{z_2}=q_{z_3}=4$). There are 9 officers $i_1,\ldots,i_9$ (in priority order), three of each type: $t_{1}=t_{2}=t_{3}=z_1$, $t_{4}=t_{5}=t_{6}=z_2$, and $t_{7}=t_{8}=t_{9}=z_3$. The cross-zone balance constraints impose, for every pair $(j,\ell)$: $
k_{jj}=1$ (insider cap: at most 1 insider per zone), $k_{j\ell}=2$ for $j\neq\ell$ (outsider cap per destination), yielding $3\times 3=9$ upper-bounds.

For a type-$z_1$ officer, the signatures are $H^{s,z_1}=\{(\{z_1\},z_\ell,k_{1\ell})\}$ where $s\in z_\ell$. Since the signature differs across the three zones, the constraint-induced partition is $\{z_1,z_2,z_3\}$---three zones, not two.

With the cyclic exogenous ranking,
\[
z_1\blacktriangleright z_2 \blacktriangleright z_3,\qquad
z_2\blacktriangleright z_3 \blacktriangleright z_1,\qquad
z_3\blacktriangleright z_1 \blacktriangleright z_2.
\]
Suppose all officers prefer $s_1$ to $s_2$ within $z_1$, $s_3$ to $s_4$ within $z_2$, and $s_5$ to $s_6$ within $z_3$. Then the assignment unfolds type by type:

The first type-$z_1$ officer, $i_1$, receives $s_1$, so the insider cap $k_{11}=1$ immediately binds. The next two type-$z_1$ officers are redirected to $z_2$, where both prefer $s_3$ to $s_4$; since $q_{s_3}=2$, officers $i_2$ and $i_3$ are both assigned to $s_3$, filling its capacity. The flow from $z_1$ to $z_2$ now reaches $k_{12}=2$.

Next, among type-$z_2$ officers, $i_4$ stays in the home zone but finds $s_3$ full, and so receives $s_4$; this triggers $k_{22}=1$. Officers $i_5$ and $i_6$ are redirected to $z_3$, where both prefer $s_5$ to $s_6$ and so are both assigned to $s_5$, filling its capacity. Hence $k_{23}=2$ binds.

Finally, for type-$z_3$, officer $i_7$ stays in $z_3$ but finds $s_5$ full, and so receives $s_6$; this triggers $k_{33}=1$. The remaining two type-$z_3$ officers, $i_8$ and $i_9$, are redirected
to $z_1$, where both prefer $s_1$ to $s_2$; since $s_1$ is full they fill the two seats at $s_2$. Thus $k_{31}=2$ binds as well.

The final allocation is therefore
\[
z_1:\ \{i_1,i_8,i_9\},\qquad
z_2:\ \{i_2,i_3,i_4\},\qquad
z_3:\ \{i_5,i_6,i_7\}.
\]
Each destination zone receives one insider and two outsiders, and the rotation creates a clean cascade: type-$z_1$ overflow goes to $z_2$, type-$z_2$ overflow to $z_3$, and type-$z_3$ overflow to $z_1$.

Crucially, the three-zone partition is what makes this balance enforceable. Because the mechanism tracks assignments to each destination zone \emph{separately} for each type, it can detect when a particular origin--destination flow is becoming saturated and redirect officers to the next zone in the rotation. 
\end{example}

\subsubsection{Specification 3: Home-Proximity Caps}\label{sec:IASspec3}

A third---and arguably the most realistic---interpretation is that the policymaker's concern about homophily was not uniform across all states in the home zone but rather was \emph{strongest for the home state itself}, weaker for nearby states, and weakest for the broader home zone. Officers serving in their exact home state pose the greatest risk of local entrenchment; those in a neighboring state within the same zone still have strong local ties; while those elsewhere in the home zone are less embedded but still not fully ``outsiders.'' This graduated concern can be captured through nested proximity constraints.

\paragraph{Types and upper-bounds.} Let each officer's type encode her home state, so $T = S$ and $t_i = h_i$. For each home state~$h\in S$, define nested subsets capturing increasing levels of geographic proximity:
\begin{itemize}
    \item $N_h^0 = \{h\}$---the home state itself,
    \item $N_h^1$---the home state and its immediate neighbors (e.g., nearby states in the same zone or states sharing a border),
    \item $N_h^2 = z(h)$---all states in the home zone.
\end{itemize}
The modular upper-bounds for type~$h$ are:
\[
H^h = \Bigl\{\, (\{h\}, N_h^0, k_h^0),\quad (\{h\}, N_h^1, k_h^1),\quad (\{h\}, N_h^2, k_h^2) \,\Bigr\},
\]
with $k_h^0 \leq k_h^1 \leq k_h^2$. For instance, one could set $k_h^0 = 0$ (no officer may serve in her exact home state), $k_h^1 = 1$ (at most one officer from~$h$ in the neighborhood), and $k_h^2 = 2$ (at most two in the broader home zone).

\paragraph{Induced message spaces.} For an officer of type~$h$, the upper-bound signature at state~$s$ is
\[
H^{s,h} = \begin{cases}
\{(\{h\}, N_h^0, k_h^0),\; (\{h\}, N_h^1, k_h^1),\; (\{h\}, N_h^2, k_h^2)\} & \text{if } s = h, \\[2mm]
\{(\{h\}, N_h^1, k_h^1),\; (\{h\}, N_h^2, k_h^2)\} & \text{if } s \in N_h^1 \setminus \{h\}, \\[2mm]
\{(\{h\}, N_h^2, k_h^2)\} & \text{if } s \in N_h^2 \setminus N_h^1, \\[2mm]
\varnothing & \text{if } s \notin N_h^2.
\end{cases}
\]
Since each of these four cases yields a distinct signature, the constraint-induced partition for type~$h$ has (up to) four zones:
\[
z_1^h = \{h\}, \qquad z_2^h = N_h^1\setminus\{h\}, \qquad z_3^h = z(h)\setminus N_h^1, \qquad z_4^h = S\setminus z(h).
\]

\paragraph{Exogenous ranking over zones.}
The exogenous ranking over the induced zones is a policy choice. If the policymaker wants officers to be placed as close as possible to home subject to the caps, a natural ranking is
\[
z_1^h \blacktriangleright_i z_2^h \blacktriangleright_i z_3^h \blacktriangleright_i z_4^h.
\]
When the exact-home cap is $k_h^0=0$, however, the first zone is closed from the outset. In that case one may equivalently place the admissible proximity tiers first, for example
\[
z_2^h \blacktriangleright_i z_3^h \blacktriangleright_i z_4^h \blacktriangleright_i z_1^h,
\]
which implements the same substantive priority among states that can actually be assigned.

The officer ranks states fully within each of these four groups but cannot compare states across them. This partition encodes a gradient of ``local embeddedness'': the mechanism suppresses precisely those cross-group comparisons that would make it impossible to enforce the graduated proximity constraints while maintaining visible fairness (see \Cref{fig:homeProximity}).

\begin{center}
\begin{minipage}{\textwidth}
\centering
\captionsetup{font=footnotesize}
\begin{tikzpicture}[
  font=\small,
  lbl/.style={font=\scriptsize, anchor=west}
]

\fill[gray!6]   (0,0) circle (3.2);
\draw[gray!30, line width=0.4pt] (0,0) circle (3.2);

\fill[orange!6] (0,0) circle (2.3);
\draw[orange!30, line width=0.4pt] (0,0) circle (2.3);

\fill[teal!8]   (0,0) circle (1.4);
\draw[teal!35, line width=0.4pt] (0,0) circle (1.4);

\fill[red!8]    (0,0) circle (0.55);
\draw[red!40, line width=0.5pt] (0,0) circle (0.55);
\node[font=\small] at (0,0) {$h$};

\coordinate (R1) at (0:0.55);
\coordinate (R2) at (15:1.4);
\coordinate (R3) at (10:2.3);
\coordinate (R4) at (5:3.2);

\node[lbl, text=red!50!black]    (L1) at (4.6, 1.5) {$N_h^0 = \{h\}$\quad cap $k_h^0$};
\node[lbl, text=teal!55!black]   (L2) at (4.6, 0.5) {$N_h^1$\quad cap $k_h^1$};
\node[lbl, text=orange!55!black] (L3) at (4.6,-0.5) {$N_h^2 = z(h)$\quad cap $k_h^2$};
\node[lbl, text=gray!55!black]   (L4) at (4.6,-1.5) {$S\setminus z(h)$\quad (no caps)};

\draw[red!30,    line width=0.3pt] (R1) -- ++(0.3,0) |- (L1.west);
\draw[teal!30,   line width=0.3pt] (R2) -- ++(0.15,0) |- (L2.west);
\draw[orange!25, line width=0.3pt] (R3) -- ++(0.1,0) |- (L3.west);
\draw[gray!25,   line width=0.3pt] (R4) -- ++(0.05,0) |- (L4.west);

\end{tikzpicture}
\captionof{figure}{Home-proximity caps (\Cref{sec:IASspec3}). Concentric rings represent nested subsets of increasing distance from the home state~$h$, each carrying a progressively looser cap: $k_h^0 \leq k_h^1 \leq k_h^2$.}
\label{fig:homeProximity}
\end{minipage}
\end{center}

\begin{example}[Home-proximity caps]\label{ex:IASspec3}
Consider six states $S=\{s_1,\ldots,s_6\}$ in two zones: $z_1=\{s_1,s_2,s_3\}$ and $z_2=\{s_4,s_5,s_6\}$, each state with capacity~$2$. There are 6 officers $i_1,\ldots,i_6$ (in priority order), with home states $h_{i_1}=h_{i_2}=h_{i_3}=s_1$ and $h_{i_4}=h_{i_5}=h_{i_6}=s_4$. The neighborhood of~$s_1$ is $N_{s_1}^1=\{s_1,s_2\}$ and $N_{s_1}^2=z_1=\{s_1,s_2,s_3\}$. Similar definitions can be made for other states.

The nested upper-bounds for type~$s_1$ are:
\[
H^{s_1}=\bigl\{\,(\{s_1\},\,\{s_1\},\,0),\quad (\{s_1\},\,\{s_1,s_2\},\,1),\quad (\{s_1\},\,\{s_1,s_2,s_3\},\,2)\bigr\}.
\]
Since $N_{s_1}^1\subsetneq N_{s_1}^2$, all three constraints involve distinct sets and distinct caps, and the induced partition for type~$s_1$ has the full four zones:
\[
z_1^{s_1}=\{s_1\},\qquad z_2^{s_1}=\{s_2\},\qquad z_3^{s_1}=\{s_3\},\qquad z_4^{s_1}=\{s_4,s_5,s_6\}.
\]

With the exogenous ranking $z_2^{h}\;\blacktriangleright\; z_3^{h}\;\blacktriangleright\; z_4^{h}\;\blacktriangleright\; z_1^{h}$ (neighbor first, then broader home zone, then away, home state last), the graduated caps produce a layered cascade. Officer~$i_1$ is assigned to~$s_2$ (the neighbor; $k_{s_1}^1=1$ not yet binding). Officer~$i_2$ faces the binding neighborhood cap and moves to the next tier: thereby receiving ~$s_3$ (broader home zone; $k_{s_1}^2=2$ not yet binding). Officer~$i_3$ finds both inner tiers saturated and is redirected to the away zone, receiving~$s_4$. By symmetry, type-$s_4$ officers cascade through~$s_5$,~$s_6$, and finally~$s_1$. The final allocation is
\[
i_1\to s_2,\quad i_2\to s_3,\quad i_3\to s_4,\quad i_4\to s_5,\quad i_5\to s_6,\quad i_6\to s_1.
\]
No officer serves in her home state ($k_h^0=0$), at most one per type in the neighborhood ($k_h^1=1$), and at most two per type in the home zone ($k_h^2=2$). Each successive officer of a given type lands one proximity tier farther from home---a graduated dispersion that neither the binary home/away split of Specification~1 nor the zone-to-zone balancing of Specification~2 can produce.
\end{example}

\section{Conclusion}
\label{sec: Conclusion}

We explored how designing mechanisms that restrict the preferences participants can report helps reconcile multiple policy objectives—particularly distributional constraints—with fairness principles grounded in strict priority orders. We introduced the concept of visible fairness, in which a mechanism never produces an allocation that appears to violate a participant’s priority based on the partially observed preferences. Theorem 1 demonstrated that every visibly fair mechanism operates as a variant of serial dictatorship adapted to partial preferences. We further showed that this framework can accommodate diverse distributional objectives by employing modular upper-bounds, which limit how many individuals of certain types can be placed in specified subsets of positions. Central to this approach are constraint-induced message spaces, which prevent participants from specifying cross-group preference comparisons that would otherwise undermine these quotas.

We then characterized what makes these mechanisms incentive-compatible, by parsing out two critical conditions---expressiveness and (weak) availability---that together ensure no participant can profit by misreporting her partial preferences. \Cref{thm:NoMechanismAllThree} shows that, once modular upper-bounds bind, restricting the scope of preference reporting can make constrained Pareto efficiency incompatible with visible fairness. The loss is informational. Under simultaneous elicitation, visible fairness may require the mechanism to avoid eliciting preference comparisons that would reveal efficiency improvements, even in instances where implementing those improvements would itself be consistent with visible fairness. 

The simulation exercise complements this impossibility result. It shows that the information lost through coarse zones can be quantitatively important: when zones are merged in a zonal priority rule with exogenous zone orders, average welfare rises substantially across a broad range of preference environments. Thus the design problem is not whether preference elicitation matters, but which comparisons can be elicited without compromising the distributional objectives that motivated the restriction.

Finally, we applied the full framework to the state allocation mechanism for India's All India Services. We established that the 2017 zone-based mechanism is visibly fair when insider quotas are removed and that its interleaving rule precludes strategy-proofness. We also discussed why within-state insider quotas are not naturally accommodated by visibly fair mechanisms based on the corresponding one-state-per-zone message spaces. We then formulated three modular-upper-bound specifications---zonal homophily caps, cross-zone balance constraints, and graduated home-proximity caps---each capturing a different plausible reading of the reform's distributional objectives. 

Overall, this paper formalizes a design practice already present in priority-based allocation systems: the use of restricted message spaces to make policy objectives operational. Its contribution is to turn this practice into a systematic design framework, identifying how such restrictions can be chosen so that the resulting mechanisms satisfy visible fairness, strategy-proofness, and the relevant distributional constraints.

\bibliographystyle{ecca}
\bibliography{incontestable}
\newpage
\appendix
\section{Appendix: Simulation Design and Additional Results}
\label{app:simulations}

This appendix describes the simulation exercise reported in \Cref{sec:preference-elicitation-simulations} in greater detail. The purpose of the exercise is to isolate the informational effect of merging zones in a visibly fair zonal priority mechanism. Merging zones serves the role of representing the elicitation of finer information about preferences from officers, by enriching the message space: states that were previously in different zones become comparable because they now lie in the same merged zone.

\subsection{Environment}

The reported figures use $n=m=50$. Every state has capacity one, so each run produces a complete assignment of all officers to all states. Utilities are generated according to
\[
    u_{is}=\sqrt{p_1}\,c_s+\sqrt{1-p_1}\,\varepsilon_{is},
\]
where $c_s\sim N(0,1)$ is a state-specific common component and $\varepsilon_{is}\sim N(0,1)$ is an officer-state idiosyncratic component. All draws are independent. The square-root scaling keeps the variance of $u_{is}$ constant across specifications and makes $p_1$ itself equal to the across-officer correlation in utilities for a given state. The parameter $p_1\in\{0,0.1,\ldots,1\}$ therefore indexes the correlation in officers' utilities for a fixed state.

For each utility profile, the simulation draws a random partition of states into $k$ nonempty zones, where $k$ is drawn uniformly from $\{2,\ldots,8\}$. The partition is constructed by first assigning one state to each zone and then assigning the remaining states randomly across zones. The simulation also draws, for each officer, an exogenous ranking of the initial zones. This ranking is a random permutation of the zones and is independent of the officer's utility vector.

The zone order is an exogenous input to the mechanism. Hence the initial message space elicits only within-zone rankings: for each zone $z$, officer $i$ ranks the states in $z$ by utility. No comparison between states in different initial zones is elicited before zones are merged.

\subsection{Mechanisms compared}

Fix a utility profile, an initial partition, and the exogenous zone order for each officer. The simulation then repeats the assignment for independent priority orders and merger draws.

The baseline allocation, denoted $a^0$, is produced as follows. Officers are processed according to a random priority order. When officer $i$ is reached, the mechanism scans her exogenous zone ranking until it finds the first zone with at least one unassigned state. The officer is assigned her most-preferred unassigned state within that zone. This gives welfare $W^0=\sum_i u_{i,a_i^0}$.

For the post-merger allocation, each pair of initial zones is merged independently with probability $p_2\in\{0,0.1,\ldots,1\}$. The resulting merged zones are the connected components of the graph whose vertices are the initial zones and whose edges are the realized pairwise mergers. Within each merged zone, officer $i$ ranks all states by utility. Across merged zones, the order is inherited from the exogenous order over initial zones. If $\rho_i(z)$ is the rank of initial zone $z$ in officer $i$'s exogenous order, with rank $1$ highest, then merged zone $\widehat z$ receives rank
\[
    \widehat{\rho}_i(\widehat z)
    =
    \min\{\rho_i(z):z\text{ is an initial zone contained in }\widehat z\}.
\]
The same priority algorithm is then applied to the merged partition, producing allocation $a^{p_2}$ and welfare $W^{p_2}=\sum_i u_{i,a_i^{p_2}}$. When $p_2=0$, no zones are merged and $W^{p_2}=W^0$. When $p_2=1$, all initial zones are connected and the mechanism coincides with the full-merger benchmark described next.

The full-merger benchmark places all states in one zone. Since every officer then ranks all states by utility, the mechanism is ordinary serial dictatorship under the same random priority order. Let $a^1$ be the resulting allocation and $W^1=\sum_i u_{i,a_i^1}$.
This benchmark measures the welfare attained when all preference comparisons are elicited, while still respecting the priority order.

Finally, the first-best benchmark ignores priority and chooses the one-to-one assignment that maximizes total welfare: $W^{FB}=\max_{a\in A}\sum_i u_{i,a_i}$. In the code, this is computed using the Hungarian algorithm. The simulation also records the expected welfare of a uniformly random one-to-one assignment, but the figures reported here use the comparisons among $W^0$, $W^{p_2}$, $W^1$, and $W^{FB}$.

\subsection{Reported normalizations}

For each parameter cell $(p_1,p_2)$, the code averages over utility profiles, priority orders, and merger realizations. The summary file reports the cell means of $W^0$, $W^{p_2}$, $W^1$, and $W^{FB}$. The main text uses two derived measures.

The first is the percentage of the full-merger gain recovered at merger probability $p_2$, computed as $100\times (W^{p_2}-W^0)/(W^1-W^0)$. The second is the percentage of the first-best gap closed, computed as $100\times (W^{p_2}-W^0)/(W^{FB}-W^0)$. The former equals 0 at $p_2=0$ and 100 at $p_2=1$ except in degenerate cases where the denominator is zero; the latter is typically smaller whenever serial dictatorship under full elicitation remains below first best.

The main figure reports the averages of these two cell-level ratios over $p_1\in\{0,0.1,\ldots,0.9\}$.
The case $p_1=1$ is omitted because $u_{is}=c_s$ for every officer $i$. Since $n=m$ and every state is assigned exactly once, every complete assignment has welfare $\sum_s c_s$. Therefore $W^0=W^{p_2}=W^1=W^{FB}$ in every run, and both denominators are zero.

\subsection{Welfare levels by preference correlation}

\Cref{fig:simulation-appendix-levels} reports the underlying welfare levels. Each panel fixes $p_1$, the cross-officer utility correlation. The horizontal axis is $p_2$, the probability with which each pair of initial zones is merged. The vertical axis is mean total welfare. The four plotted series are: $W^0$, $W^{p_2}$, $W^1$, and $W^{FB}$. Thus the figure shows both the informational gain from merging zones and the remaining gap between priority-respecting full elicitation and the first best.

\begin{figure}[H]
\centering
\includegraphics[width=\textwidth]{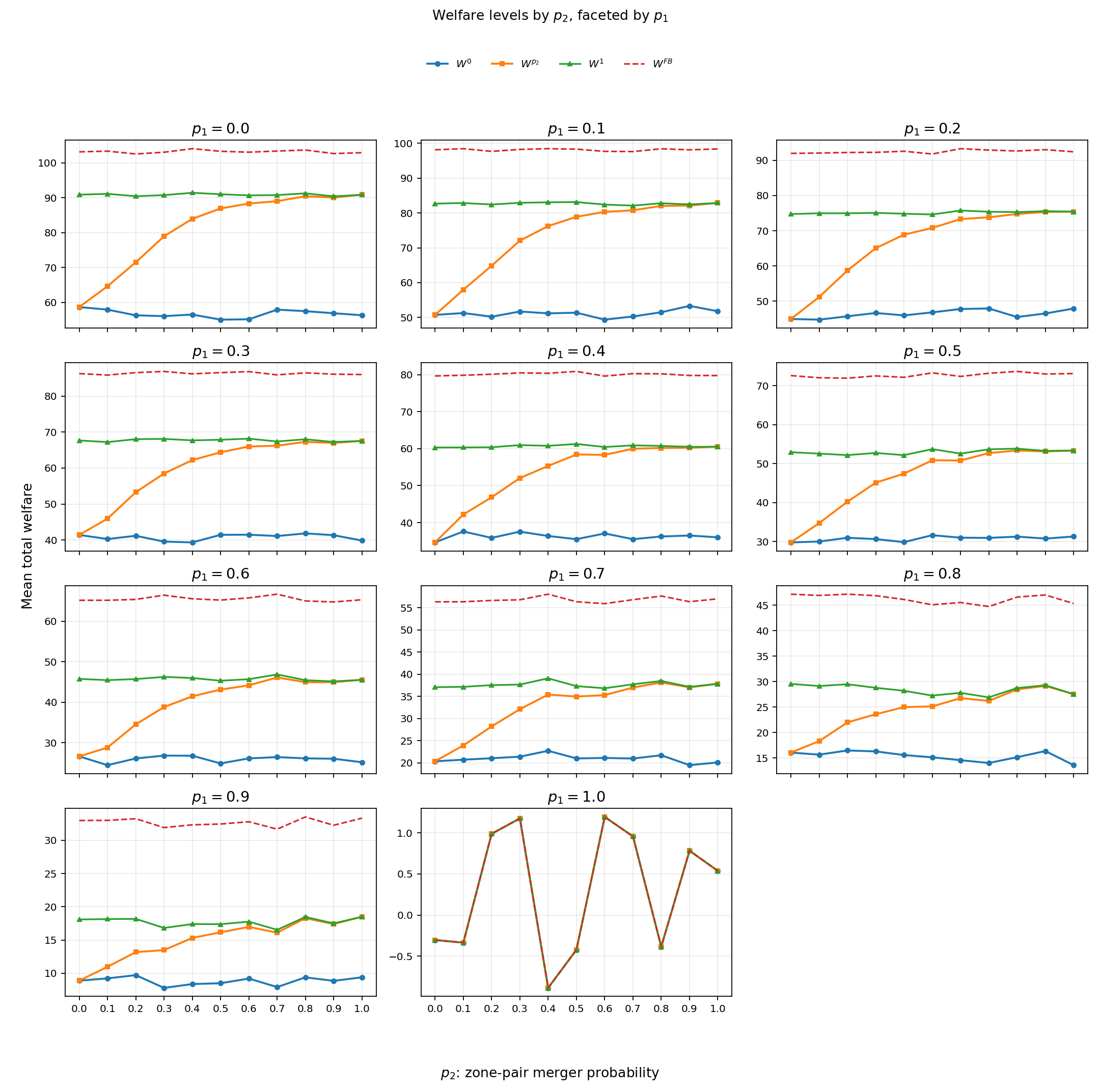}
\caption{Welfare levels by merger probability and utility correlation. Each panel fixes $p_1$. The blue line is the baseline welfare $W^0$ under the initial zonal message space. The orange line is post-merger welfare $W^{p_2}$. The green line is the full-merger serial-dictatorship benchmark $W^1$. The red dashed line is first-best welfare $W^{FB}$.}
\label{fig:simulation-appendix-levels}
\end{figure}

The panels show that $W^{p_2}$ generally rises with $p_2$ and approaches $W^1$. This is the welfare effect of eliciting additional comparisons by merging zones. The increase is especially visible for low and intermediate values of $p_1$, where there is enough heterogeneity in utilities for assignment choices to matter.

The figure also shows that $W^1$ remains below $W^{FB}$ for all nondegenerate values of $p_1$. Hence full merging eliminates the informational loss from coarse zones, but it does not eliminate the welfare loss associated with the priority order. The first-best assignment can reallocate states to maximize total welfare, while serial dictatorship must respect the random priority order.

As $p_1$ increases, welfare levels fall in absolute value because there is less idiosyncratic variation for the assignment to exploit. At $p_1=1$, all officers have the same utility for each state. Since all states are assigned, all complete assignments have the same total welfare. This is why the four welfare series coincide in the $p_1=1$ panel. The visible movement of the overlapping lines across $p_2$ in that panel is sampling variation from independent utility draws across parameter cells, not an allocation effect.

Small movements in $W^0$, $W^1$, and $W^{FB}$ across $p_2$ within a fixed $p_1$ panel should be interpreted similarly. Conceptually, these benchmarks do not depend on the merger probability. They are plotted at each $p_2$ because each parameter cell is simulated independently and the figure reports the corresponding cell means. The economically relevant $p_2$ movement is the rise of $W^{p_2}$ from the baseline $W^0$ toward the full-merger benchmark $W^1$.

\section{Appendix: Proofs}
\label{sec: Appendix Proofs}

\subsection{Proof of \Cref{thm:CharacterizationVisiblyFairMechanisms}}

\textit{A mechanism $\psi$ is visibly fair if and only if it is a $m$-queue allocation mechanism. }

\begin{proof}

\underline{Part 1}: A $m$-queue allocation mechanism $\psi$ is visibly fair.

\textbf{Case 1}. Consider any $m\in M$ and suppose by contradiction that for individual $i$ there exists $\psi(m)_j \succ_{m_i} \psi(m)_i$, where $\pi(i)< \pi(j)$.

By construction of the mechanism, suppose individual $i$ is assigned at step $k$; note that both $\psi(m)_j \in S^k$ and $\psi(m)_i \in S^k$ as $\pi(i)< \pi(j)$. It follows, by $\psi(m)_j \succ_{m_i} \psi(m)_i$ that $\psi(m)_i \not\in G(S^k, m_k)$ --- a contradiction.

\textbf{Case 2}. Consider any $m\in M$ and suppose by contradiction that for individual $i$ there exists $s \succ_{m_i} \psi(m)_i$, where $|\{i\in I: \psi(m)_i =s\}| < q_s$.

By construction of the mechanism, suppose individual $i$ is assigned at step $k$; note that both $s \in S^k$ and $\psi(m)_i \in S^k$ as $|\{i\in I: \psi(m)_i =s\}| < q_s$. It follows, by $s\succ_{m_i} \psi(m)_i$ (or $s\succ_{m_k} \psi(m)_i$ ) that $\psi(m)_i \not\in G(S^k, m_k)$ --- a contradiction.

    \noindent  \underline{Part 2}: A visibly fair mechanism is a $m$-queue allocation mechanism

Take an arbitrary visibly fair mechanism $\psi$ and fix a message profile
$m\in M$.  Denote the resulting allocation by
$a=\psi(m)$.
We now \emph{construct} for this profile the $m$-queue allocation procedure that replicates $a$.
\medskip

Define set $S^k$ for each officer $k\in I$ as follows:

\begin{itemize}
\item[] \textbf{Step $0$}:  Set $S^1=S$.

    \item[] \textbf{Step $k$ ($1\leq k\leq n$)}: Set $s^k=a_k$. If the number of officers assigned to $s^k$ reaches $q_{s^k}$, that is $|\{ i\in I:  i\leq k \text{ and } a_i=s^k\}|= q_{s^k}$, then $S^{k+1}\equiv S^{k}\backslash\{s^k\}$. Otherwise, $S^{k+1}= S^{k}$.

\end{itemize}

Define $g_k(S^k,m)=a_k$ on the realized pair $(S^k,m)$. For every nonempty pair $(X,m')$ not specified by this construction, define $g_k(X,m')$ arbitrarily by choosing an element of $G(X,m'_k)$, which is nonempty because $\succ_{m'_k}$ is a strict partial order and $X$ is finite and nonempty. Notice that for officer $k$, the set $S^k$ consists of only those states that still have remaining capacities to accommodate $k$ after the allocation of higher-ranking officers is taken into account. That is, each state $s \in S \setminus S^k$ is filled by higher ranking offers, $|\{ i\in I:  i< k \text{ and } a_i=s\}|= q_{s}$. Now we must have that $a_k \in G(S^k,m_k)$, otherwise there exists $s \in S^k$ such that $s \succ_{m_k} a_k$, which makes the outcome visibly unfair as a higher ranking offer does not occupy this state. 

An $m$-queue mechanism working down the priority list, using sets $\{S^k\}_{k=1}^n$ defined above and the procedure given in \Cref{def:mQueueMechanisms}, would give every officer the same assignment as $\psi$. Since the construction can be repeated for every profile $m$, $\psi$ is an $m$-queue allocation mechanism.
\end{proof}

\subsection{Proof of \Cref{thm:CharacterizationIncontestableMechanisms2}}

\textit{For a zonal message space $M$, $\psi$ is visibly fair if and only if it is a partitioned priority mechanism.}

\begin{proof}
\underline{Part 1}: A partitioned priority mechanism is visibly fair.

The proof follows the same logic as in Theorem 1, Part 1, and is omitted for brevity.

\noindent  \underline{Part 2}: A visibly fair mechanism is a partitioned priority mechanism.

Fix an arbitrary message profile \(m\in M\) and let
\(a=\psi(m)\) be the allocation produced by any visibly fair mechanism~\(\psi\).
We construct zone selection functions \(\{\mathcal{C}_k\}_{k=1}^n\) and sets \(\{S^k\}_{k=1}^n\) that
would reproduce the same allocation step by step.

\begin{itemize}
\item[] \textbf{Step \(0\)}:  Set \(S^1=S\).

\item[] \textbf{Step \(k\) (\(1\leq k\leq n\))}: Set \(s^k=a_k\). Let \(z_k(a_k)\) be the unique zone in \(Z_k\) containing \(a_k\), and set
\[
\mathcal{C}_k(S^k,m)=z_k(a_k).
\]
If the number of officers assigned to \(s^k\) reaches \(q_{s^k}\), that is
\[
\left|\{\ell\in\{1,\ldots,k\}: a_\ell=s^k\}\right|=q_{s^k},
\]
then \(S^{k+1}\equiv S^{k}\backslash\{s^k\}\). Otherwise, \(S^{k+1}=S^k\).
\end{itemize}

For every nonempty pair \((X,m')\) not specified by this construction, define \(\mathcal{C}_k(X,m')\) arbitrarily as any zone \(z\in Z_k\) with \(X\cap z\neq\emptyset\). Such a zone exists because \(Z_k\) partitions \(S\) and \(X\neq\emptyset\). Thus each \(\mathcal{C}_k\) is defined on its full domain.

Notice that for officer \(i_k\), the set \(S^k\) consists of only those states that still have remaining capacities after the allocation of higher-priority officers is taken into account. That is, each state \(s \in S \setminus S^k\) is filled by higher-priority officers, \(|\{\ell\in\{1,\ldots,k-1\}:a_\ell=s\}|=q_s\). Now we must have \(a_k\in G(S^k,m_k)\cap \mathcal{C}_k(S^k,m)\). If not, then either \(a_k\notin \mathcal{C}_k(S^k,m)\), contradicting the construction of \(\mathcal{C}_k\), or there exists \(s\in S^k\) such that \(s\succ_{m_k}a_k\), which makes the outcome visibly unfair because no higher-priority officer occupies \(s\).

A partitioned priority mechanism working down the priority list, using the zone selection functions \(\{\mathcal{C}_k\}_{k=1}^n\), the sets \(\{S^k\}_{k=1}^n\) defined above, and the procedure given in \Cref{def:PartitionPriorityMech}, would assign every officer the same post as \(\psi\).

Since the construction can be repeated for every profile \(m\), \(\psi\) is a partitioned priority mechanism.
\end{proof}

\subsection{Proof of \Cref{thm:CharacterizationIncontestableMechanisms3}}

\textit{For zonal message space with ranking over zones, $\psi$ is visibly fair if and only if it is a ranked partitioned priority mechanism.}

\begin{proof}
\underline{Part 1}: A ranked partitioned priority mechanism is visibly fair.

The proof follows the same logic as in Theorem 1, Part 1, and is omitted for brevity.
    
\noindent  \underline{Part 2}: A visibly fair mechanism is a ranked partitioned priority mechanism.

Fix an arbitrary message profile \(m\in M\) and let
\(a=\psi(m)\) be the allocation produced by any visibly fair mechanism~\(\psi\).
We construct ranked zone selection functions \(\{\mathcal{C}_k\}_{k=1}^n\) and sets \(\{S^k\}_{k=1}^n\) that
would reproduce the same allocation step by step.

\begin{itemize}
\item[] \textbf{Step \(0\)}:  Set \(S^1=S\).

\item[] \textbf{Step \(k\) (\(1\leq k\leq n\))}: Set \(s^k=a_k\). Let \(z_k(a_k)\) be the unique zone in \(Z_k\) containing \(a_k\), and set
\[
\mathcal{C}_k(S^k,m)=z_k(a_k).
\]
If the number of officers assigned to \(s^k\) reaches \(q_{s^k}\), that is
\[
\left|\{\ell\in\{1,\ldots,k\}:a_\ell=s^k\}\right|=q_{s^k},
\]
then \(S^{k+1}\equiv S^k\backslash\{s^k\}\). Otherwise, \(S^{k+1}=S^k\).
\end{itemize}

For every nonempty pair \((X,m')\) not specified by this construction, extend \(\mathcal{C}_k\) subject to the ranked-zone selection condition over \(Z_k\). If some zone \(z\in Z_k\) satisfies \(X\cap z\neq\emptyset\) and
\[
X\cap z\neq\{\min(z,m'_k)\},
\]
choose such a zone. Otherwise, every zone in \(Z_k\) that intersects \(X\) contains only its minimal state, and we choose a \(\triangleright_{m'_k}\)-maximal zone among the zones in \(Z_k\) that intersect \(X\). This choice satisfies the ranked-zone selection condition.

After performing the construction for every \(k\), we have defined all ranked zone selection functions \(\{\mathcal{C}_k\}_{k=1}^n\) and sets \(\{S^k\}_{k=1}^n\).

Notice that for officer \(i_k\), the set \(S^k\) consists of only those states that still have remaining capacities after the allocation of higher-priority officers is taken into account. That is, each state \(s\in S\setminus S^k\) is filled by higher-priority officers, \(|\{\ell\in\{1,\ldots,k-1\}:a_\ell=s\}|=q_s\). Now we must have \(a_k\in G(S^k,m_k)\cap\mathcal{C}_k(S^k,m)\). If not, then either \(a_k\notin \mathcal{C}_k(S^k,m)\), contradicting the construction of \(\mathcal{C}_k\), or there exists \(s\in S^k\) such that \(s\succ_{m_k}a_k\), which makes the outcome visibly unfair because no higher-priority officer occupies \(s\).

It remains to verify that the constructed zone selection function is a ranked zone selection function at the realized pair \((S^k,m)\). Suppose not. Then, for some step \(k\),
\[
S^k \cap \mathcal{C}_k(S^k,m)
=
\left\{\min\left(\mathcal{C}_k(S^k,m),m_k\right)\right\}
\]
and there is a zone \(z\in Z_k\) such that
\[
z\triangleright_{m_k}\mathcal{C}_k(S^k,m)
\quad\text{and}\quad
\max(z,m_k)\in S^k.
\]
By the definition of the expressed cross-zone comparison,
\[
\max(z,m_k)\succ_{m_k}\min\left(\mathcal{C}_k(S^k,m),m_k\right).
\]
Since \(a_k=\min\left(\mathcal{C}_k(S^k,m),m_k\right)\), officer \(i_k\) visibly prefers the available state \(\max(z,m_k)\) to her assignment, contradicting visible fairness.

A ranked partitioned priority mechanism working down the priority list, using the ranked zone selection functions \(\{\mathcal{C}_k\}_{k=1}^n\), the sets \(\{S^k\}_{k=1}^n\) defined above, and the procedure given in \Cref{def:RankedPartitionPriorityMech}, would assign every officer the same post as \(\psi\).

Since the construction can be repeated for every profile \(m\), \(\psi\) is a ranked partitioned priority mechanism.
\end{proof}

\subsection{Proof of \Cref{thm:Truthful1}}

\textit{ A visibly fair mechanism is strategy-proof if and only if it satisfies
expressiveness and weak availability.}

\begin{proof}
Fix a visibly fair mechanism $\psi$.

First, suppose that $\psi$ is strategy-proof. We show that it satisfies expressiveness and weak availability.

\emph{Expressiveness.} Suppose, toward a contradiction, that $\psi$ is not expressive. Then there exist an officer $i$, a message profile $m=(m_i,m_{-i})\in M$, and an alternative message $\hat{m}_i\in M_i$ such that, letting $x:=\psi(m_i,m_{-i})_i$ and $y:=\psi(\hat{m}_i,m_{-i})_i$, we have $x\neq y$, and $x$ and $y$ are not comparable under $m_i$.

Since $\succ_{m_i}$ is a strict partial order and $x$ and $y$ are incomparable, the relation $R:=\succ_{m_i}\cup\{(y,x)\}$ is acyclic. Indeed, any cycle in $R$ would have to use the added edge $(y,x)$, and would therefore imply a path from $x$ to $y$ under $\succ_{m_i}$, contradicting the incomparability of $x$ and $y$ under $m_i$.

By the finite linear-extension theorem, there exists a strict total order $\succ_i^*$ on $S$ extending $R$. Hence $y\succ_i^*x$, and, because $\succ_i^*$ extends $\succ_{m_i}$, the message $m_i$ is truthful for $\succ_i^*$, i.e., $m_i\in\mathcal{T}_i(\succ_i^*)$. But under the truthful message $m_i$, officer $i$ receives $x$, whereas by deviating to $\hat{m}_i$, officer $i$ receives $y$, and $y\succ_i^*x$. This contradicts strategy-proofness. Therefore $\psi$ is expressive.

\emph{Weak availability.} Suppose, toward a contradiction, that $\psi$ violates weak availability. Then there exist an officer $i$, a strict preference $\succ_i\in\mathcal{Q}_i$, a truthful message $m_i\in\mathcal{T}_i(\succ_i)$, a profile $m_{-i}\in M_{-i}$, and an alternative message $\hat{m}_i\in M_i$ such that, letting $x:=\psi(m_i,m_{-i})_i$ and $y:=\psi(\hat{m}_i,m_{-i})_i$, we have $y\succsim_i x$, but $y$ is not available to officer $i$ under $m=(m_i,m_{-i})$.

If $y=x$, then $y$ is available to $i$ under $m$, because $i$ herself is assigned to $y$ and feasibility implies that strictly fewer than $q_y$ higher-priority officers are assigned to $y$. Hence $y\neq x$. Since $\succ_i$ is strict and $y\succsim_i x$, it follows that $y\succ_i x$. Thus $\hat{m}_i$ is a profitable deviation from the truthful message $m_i$, contradicting strategy-proofness. Therefore $\psi$ satisfies weak availability.

Conversely, suppose that $\psi$ is expressive and satisfies weak availability. We show that $\psi$ is strategy-proof. Suppose, toward a contradiction, that it is not. Then there exist an officer $i$, a strict preference $\succ_i\in\mathcal{Q}_i$, a truthful message $m_i\in\mathcal{T}_i(\succ_i)$, a profile $m_{-i}\in M_{-i}$, and an alternative message $\hat{m}_i\in M_i$ such that
$$
\psi(\hat{m}_i,m_{-i})_i\succ_i\psi(m_i,m_{-i})_i.
$$
Let $x:=\psi(m_i,m_{-i})_i$ and $y:=\psi(\hat{m}_i,m_{-i})_i$. By expressiveness, $x$ and $y$ are comparable under $m_i$. Since $m_i$ is truthful for $\succ_i$ and $y\succ_i x$, it cannot be that $x\succ_{m_i}y$. Because $x\neq y$, comparability implies $y\succ_{m_i}x$. By weak availability, $y$ is available to officer $i$ under $m=(m_i,m_{-i})$, i.e.,
$$
\left|\{j\in I:\psi(m)_j=y \text{ and } \pi(j)<\pi(i)\}\right|<q_y.
$$

Consider the allocation $a=\psi(m)$. Since $a_i=x\neq y$, officer $i$ is not assigned to $y$. If fewer than $q_y$ officers are assigned to $y$ under $a$, then $y$ has spare capacity and $y\succ_{m_i}a_i$, so $a$ is visibly unfair by visible waste. If exactly $q_y$ officers are assigned to $y$, then, because fewer than $q_y$ higher-priority officers are assigned to $y$, some lower-priority officer $j$ with $\pi(i)<\pi(j)$ must be assigned to $y$. Since $y\succ_{m_i}a_i$, $a$ is visibly unfair by visible justified envy.

In either case, $\psi(m)$ is visibly unfair under $m$, contradicting visible fairness. Therefore $\psi$ is strategy-proof.
\end{proof}

\subsection{Proof of \Cref{modulartheorem}}

\textit{Suppose that $H$ is a sequentially solvent modular upper-bound system. Then the Modular Priority Mechanism is visibly fair, strategy-proof, and respects upper-bounds.}

\begin{proof}
     First, it should be clear that the Modular Priority Mechanism is a partitioned priority mechanism: each agent, following the priority order, is associated with a zonal message space, and is matched to the most-preferred state from a zone that still has states with spare capacity. Since the modular upper‑bound system \(H\) satisfies sequential solvency, at every step of the MPM an officer has a state available for which no upper-bound is binding and with spare capacity.

    Next, we show that, for any given set of type-specific modular upper-bound systems and exogenous rankings \( (\blacktriangleright_{i})_{i\in I} \), the Modular Priority Mechanism, represented by the function $\psi$, is strategy-proof. 
    
    Notice first that, by construction, an officer $i$ cannot, by submitting a different message, change the assignment of any officer $j<i$. Since $\psi$ only produces feasible outcomes, this implies that $\psi$ satisfies availability.
    
    Since the zone $z$ from which $i$'s assignment is drawn depends only on the assignments of officers with higher priority and on the exogenous ranking $\blacktriangleright_i$, it is unaffected by $m_i$. Hence the zone to which $i$ is assigned cannot be changed by changing her message. Let $s=\psi(m)_i$. For every $m_i'\in M_i$, we have $\psi(m_i',m_{-i})_i\in z$. Since both $s$ and $\psi(m_i',m_{-i})_i$ belong to the same zone $z$, they are comparable under $m_i$. Therefore, $\psi$ satisfies expressiveness, and by \Cref{coro:Truthful1}, strategy-proofness.

    Finally, we need to show that the Modular Priority Mechanism respects upper-bounds. That is, the final allocation \(\psi(m) = a\) satisfies, for every upper-bound \( (\Xi_h,S_h, k_h) \in H \):
\[
    \left| \left\{ i \in I : a_i \in S_h,\ t_i \in \Xi_h \right\} \right| \leq k_h.
\]

We prove this by induction on the steps \( k = 1, 2, \ldots, n \) of the mechanism. In order to facilitate notation and comprehension, we will denote by $a_k$ the ``tentative'' allocation $a$ by the end of step $k$ of the mechanism.

\medskip

\noindent\textbf{Base Case (\( k = 0 \))}:

At initialization, no officers have been assigned. That is, \( a^0_i = \emptyset \) for all \( i \in I \). Therefore, for any upper-bound \( (\Xi_h,S_h, k_h) \in H \), we have:
\[
    \left| \left\{ i \in I : a^0_i \in S_h,\ t_i \in \Xi_h \right\} \right| = 0 \leq k_h.
\]
Thus, the allocation trivially respects all modular upper-bounds at step \( k = 0 \). Moreover, the initialization-time call to the quotas-update procedure correctly registers any $h \in H$ with $k_h = 0$ as binding before Step~1.

\medskip

\noindent\textbf{Inductive Step}:

Assume that at step \( k - 1 \), the current assignment \( a^{k-1} \) respects all modular upper-bounds; that is, for every upper-bound \( (\Xi_h,S_h, k_h) \in H \):
\[
    N_{h}^{k-1} = \left| \left\{ i \in I : a_i^{k-1} \in S_h,\ t_i \in \Xi_h \right\} \right| \leq k_h.
\]

We need to show that after assigning officer \( i_k \) at step \( k \), the updated assignment \( a^{k} \) also respects all modular upper-bounds.

\medskip

Consider officer \( i_k \) of type \( t_{k} \). According to the mechanism, \( i_k \) is assigned to a state \( s^\ast\) within a zone \( z^{t_{k}}_{\ell^*} \) where the upper-bounds for type \( t_{k} \) involving states in that zone are not yet binding before $i_k$'s assignment---this is tracked by the flag \( B_{\ell^*}^{t_{k}} = \textsc{False} \).

That is, for every upper-bound $(\Xi_h, S_h, k_h) \in H$ for which
$s^* \in S_h$ \emph{and} $t_{k} \in \Xi_h$:
\[
\bigl|\{i \in I : a_i^{k-1} \in S_h,\ t_i \in \Xi_h\}\bigr| < k_h.
\]

To see this, suppose $N_h^{k-1} = k_h$ for some such $h$. By Definition~\ref{def:constraintInducedMessageSpaces}, all states in $z_{\ell^*}^{t_{k}}$ share the same upper-bound signature for type $t_{k}$, so $s^* \in S_h$ implies $z_{\ell^*}^{t_{k}} \subseteq S_h$. The most recent execution of the quotas-update procedure would then have set $B_{\ell^*}^{t_{k}} = \textsc{True}$, contradicting $B_{\ell^*}^{t_{k}} = \textsc{False}$ in Step~2.

\medskip
We now verify $N_h^k \le k_h$ for every $h \in H$:
\begin{itemize}
  \item If $t_{k} \notin \Xi_h$, then $N_h^k = N_h^{k-1} \le k_h$ by the
        inductive hypothesis.
  \item If $t_{k} \in \Xi_h$ but $s^* \notin S_h$, then again
        $N_h^k = N_h^{k-1} \le k_h$.
  \item If $t_{k} \in \Xi_h$ and $s^* \in S_h$, the strict slack established
        above gives $N_h^k = N_h^{k-1} + 1 \le k_h$.
\end{itemize}
Therefore, the resulting assignment $a_k$ respects all modular upper-bounds.

\medskip

\noindent\textbf{Conclusion}:

In both cases, the assignment \( a^{k} \) at step \( k \) respects all modular upper-bounds. By induction, the final allocation \( a^{n} \) produced by the mechanism satisfies:
\[
    \left| \left\{ i \in I : a_i^{n} \in S_h,\ t_i \in \Xi_h \right\} \right| \leq k_h,
\]
for every upper-bound \( h\in H \).

\end{proof}

\subsection{Proof of \Cref{thm:NoMechanismAllThree}}

\textit{There exists a sequentially solvent modular upper-bound system $H$ for which no mechanism is simultaneously visibly fair, respects $H$ at every message profile, and is constrained Pareto efficient.}

\begin{proof}
Consider the following problem. There are three officers $I=\{i_1,i_2,i_3\}$, with $\pi(i_1)<\pi(i_2)<\pi(i_3)$, and two states $S=\{s_1,s_2\}$, with $q_{s_1}=q_{s_2}=2$. All officers have the same type $t$. The modular upper-bound system consists of the single cap $H=\{(\{t\},\{s_1\},1)\}$. Thus, at most one officer may be assigned to $s_1$. Since all three officers must be assigned and $q_{s_2}=2$, every allocation that respects $H$ assigns exactly one officer to $s_1$ and two officers to $s_2$.

First note that $H$ is sequentially solvent. Given any feasible partial allocation of officers other than $i$, if the cap on $s_1$ is not binding and $s_1$ has spare capacity, then $i$ can be assigned to $s_1$; otherwise, because at most two other officers have been assigned and the partial allocation respects $H$, state $s_2$ has spare capacity and $i$ can be assigned to $s_2$.

Suppose, toward a contradiction, that a mechanism $\psi$ is visibly fair, respects $H$ at every message profile, and is constrained Pareto efficient.

Say that officer $i$ can reveal $s_1$ over $s_2$ if there exists a message $m_i\in M_i$ such that $s_1\succ_{m_i}s_2$. At most one officer can reveal $s_1$ over $s_2$. To see this, suppose two officers, say $i$ and $j$, have messages $m_i$ and $m_j$ with $s_1\succ_{m_i}s_2$ and $s_1\succ_{m_j}s_2$. Complete the profile arbitrarily. Since $\psi$ respects $H$, at most one officer is assigned to $s_1$. Hence at least one of $i$ and $j$ is assigned to $s_2$. But $s_1$ has capacity two and only one officer can be assigned there under $H$, so $s_1$ has spare capacity. The officer assigned to $s_2$ visibly prefers $s_1$ to $s_2$, yielding visible unfairness. This contradicts visible fairness.

Therefore at least two officers cannot reveal $s_1$ over $s_2$. Let these two officers be $r$ and $\ell$. Consider officer $r$. Since truthful messages exist, there is some message $m_r^0\in\mathcal{T}_r(s_1\succ_r s_2)$. Because $r$ cannot reveal $s_1$ over $s_2$, this message cannot satisfy $s_1\succ_{m_r^0}s_2$. Since it is truthful for $s_1\succ_r s_2$, it also cannot satisfy $s_2\succ_{m_r^0}s_1$. Therefore $s_1$ and $s_2$ are incomparable under $m_r^0$. Hence $m_r^0$ is truthful both for $s_1\succ_r s_2$ and for $s_2\succ_r s_1$. Analogously, there exists a message $m_\ell^0$ that is truthful for both strict preferences of officer $\ell$.

Let $p$ denote the remaining officer. Choose any message $m_p^-\in\mathcal{T}_p(s_2\succ_p s_1)$, which exists by the truthful-message existence assumption. Consider the fixed message profile $m^0=(m_r^0,m_\ell^0,m_p^-)$.

Now compare two true preference profiles. Under the first profile, officer $r$ is the only officer who prefers $s_1$: $s_1\succ_r s_2$, $s_2\succ_\ell s_1$, and $s_2\succ_p s_1$. The message profile $m^0$ is truthful for this preference profile. Constrained Pareto efficiency requires $r$ to be assigned to $s_1$. Indeed, if some other officer were assigned to $s_1$, then that officer would prefer $s_2$, while $r$ would prefer $s_1$; swapping their assignments would respect $H$ and Pareto improve the allocation.

Under the second profile, officer $\ell$ is the only officer who prefers $s_1$: $s_2\succ_r s_1$, $s_1\succ_\ell s_2$, and $s_2\succ_p s_1$. The same message profile $m^0$ is truthful for this second preference profile. By the same argument, constrained Pareto efficiency requires $\ell$ to be assigned to $s_1$.

But $\psi(m^0)$ is a single allocation, and any allocation respecting $H$ assigns only one officer to $s_1$. It cannot assign both $r$ and $\ell$ to $s_1$. This contradiction proves the result.
\end{proof}

\subsection{Proof of \Cref{prop:IASVisiblyFair}}\label{prop:IASVisiblyFairproof}

\textit{The mechanism induced by the interleaved serial dictatorship procedure is a ranked partitioned priority mechanism and is therefore visibly fair.}

\begin{proof}
For officer~\(i_k\) processed at step~\(k\), given the set of states with remaining capacity \(X=S^k\) and the message profile~\(m\), define \(\mathcal{C}_k^{\textup{IAS}}(X,m)\) as follows. Let \(z^1,z^2,\ldots,z^\ell\) be the zones in \(Z_k\) ordered by \(\triangleright_{m_k}\), and for each zone \(z^j\), let \(s^j_1,s^j_2,\ldots\) be the states in \(z^j\) in decreasing order of preference under~\(m_k\). The \textbf{IAS interleaving zone selection function} \(\mathcal{C}_k^{\textup{IAS}}(X,m)\) returns the zone containing the first state in the interleaving sequence
\[
s^1_1,\; s^2_1,\; \ldots,\; s^\ell_1,\; s^1_2,\; s^2_2,\; \ldots,\; s^\ell_2,\; s^1_3,\; \ldots
\]
that belongs to~$X$.

We verify that $\mathcal{C}_k^{\textup{IAS}}$ is a ranked zone selection function and that the resulting allocation coincides with the interleaved assignment.

Fix a step~$k$, an available set $X = S^k \neq \emptyset$, and a message profile~$m$. Let $z^j = \mathcal{C}_k^{\textup{IAS}}(X, m)$, and let $s^j_r$ be the first state in the interleaving sequence that belongs to~$X$ (so $s^j_r \in X \cap z^j$, where $r$ and $j$ identify the round and zone in which $s^j_r$ appears).

\smallskip

\noindent\textit{Condition~1:} $X \cap \mathcal{C}_k^{\textup{IAS}}(X,m) \neq \emptyset$.

By construction, $s^j_r \in X \cap z^j$, so the intersection is non-empty.

\smallskip

\noindent\textit{Condition~2:} Either \(X \cap z^j \neq \bigl\{\min(z^j,m_k)\bigr\}\), or there is no zone \(z'\in Z_k\) with \(z'\triangleright_{m_k}z^j\) and \(\max(z',m_k)\in X\).

We show the second condition holds. Take any zone $z^{j'}$ with $j' < j$ (i.e., $z^{j'} \triangleright_{m_k} z^j$). state $s^{j'}_1 = \max(z^{j'}, m_k)$ appears at position~$j'$ in round~1 of the interleaving sequence, whereas $s^j_r$ appears at position $(r{-}1)\ell + j$. Since $j' < j \leq (r{-}1)\ell + j$, state $s^{j'}_1$ precedes $s^j_r$ in the sequence. Because $s^j_r$ is the first element of the sequence belonging to~$X$, we have $s^{j'}_1 \notin X$. Hence no zone ranked higher than $z^j$ under $\triangleright_{m_k}$ has its most-preferred state in~$X$.

When $j=1$ (the top-ranked zone), no zone ranks higher, so the condition holds vacuously.

\smallskip

\noindent\textit{Correctness of the allocation.}
At step~$k$, the ranked partitioned priority mechanism assigns $a_k = G(S^k, m_k) \cap z^j$. Since states $s^j_1, \ldots, s^j_{r-1}$ all precede~$s^j_r$ in the interleaving sequence and hence do not belong to $S^k$, the state~$s^j_r$ is the $m_k$-maximal element of $S^k \cap z^j$, so $G(S^k, m_k) \cap z^j = \{s^j_r\}$. This coincides with the state assigned by the interleaving rule.

Since $\mathcal{C}_k^{\textup{IAS}}$ is a ranked zone selection function and the allocation matches the ranked partitioned priority mechanism, visible fairness follows from \Cref{thm:CharacterizationIncontestableMechanisms3}.
\end{proof}

\section{Proofs that examples satisfy sequential solvency}\label{app:seq-solvency}

Throughout this appendix we verify the current definition of sequential solvency. Thus, in each proof we fix an arbitrary officer $i\in I$, an arbitrary subset $J\subseteq I\setminus\{i\}$, and an arbitrary feasible partial allocation $a^J=(a_j)_{j\in J}$ that respects $H$. We then exhibit a state with spare capacity at which officer $i$ can be placed without violating any modular upper-bound relevant to her type.

\bigskip

\textbf{Proof that Example \ref{ex:example2} satisfies sequential solvency.}

The three doctor types are symmetric up to relabeling of regions, so fix an arbitrary doctor $i$ with $t_i=1$. Let
\[
    B=\{s_4,s_7\}
\]
be the rural hospitals outside $R_1$, and let
\[
    b=\bigl|\{j\in J:a_j\in B\}\bigr|,
    \qquad
    u=\bigl|\{j\in J:a_j\in U\}\bigr|.
\]
The set $B$ has total capacity $8$, and no type--$1$ upper-bound applies to states in $B$.

Suppose, toward a contradiction, that no state is admissible for $i$. If $b<8$, then some rural hospital outside $R_1$ has spare capacity and is unconstrained for type~$1$, contradicting the supposition. Hence $b=8$.

If the urban cap is binding, then $u=19$. Since $B\cap U=\emptyset$, this would imply
\[
    |J|\geq b+u=8+19=27,
\]
contradicting $|J|\leq |I|-1=26$. Therefore the urban cap is not binding, so $u<19$.

Because the urban cap is not binding, any vacant urban hospital outside $R_1$ would be admissible for $i$; the local cap for type~$1$ does not apply outside $R_1$. Hence all $16$ seats in $U\setminus R_1$ must be occupied. Together with $b=8$, this means all $24$ seats outside $R_1$ are occupied. Since $|J|\leq 26$, at most two officers in $J$ are assigned to $R_1$. Consequently the local cap $(\{1\},R_1,6)$ is not binding, and $R_1$ has spare capacity. As the urban cap is also not binding, any vacant state in $R_1$ is admissible for $i$, again a contradiction.

Thus an admissible state always exists for a type--$1$ doctor. The cases $t_i=2$ and $t_i=3$ follow by relabeling the regions. Hence Example~\ref{ex:example2} satisfies sequential solvency.

\bigskip

\textbf{Proof that Example \ref{ex:motivation_formal} satisfies sequential solvency.}

Fix any officer $i$. There are only two officers, so $J$ is either empty or contains the other officer. If $J=\emptyset$, state $s_1$ is admissible. If $J=\{j\}$ and $a_j=s_1$, then the cap on $s_1$ is binding, but $s_2$ has spare capacity and no upper-bound applies there. If $J=\{j\}$ and $a_j=s_2$, then $s_2$ is full, but $s_1$ has spare capacity and its cap is not binding. Hence an admissible state exists in all cases, so Example~\ref{ex:motivation_formal} satisfies sequential solvency.

\bigskip

\textbf{Proof that Example \ref{ex:IASspec1} satisfies sequential solvency.}

The two types play symmetric roles, so fix an arbitrary officer $i$ with $t_i=z_1$. The away zone $z_2$ has total capacity $4$ and no type--$z_1$ upper-bound applies there.

If any seat in $z_2$ is vacant, it is admissible for $i$. Otherwise $z_2$ is full, so four officers in $J$ are assigned to $z_2$. Since $|J|\leq |I|-1=5$, at most one officer in $J$ is assigned to the home zone $z_1$. Therefore $z_1$ has spare capacity and the insider cap $(\{z_1\},z_1,2)$ is not binding. A vacant seat in $z_1$ is then admissible for $i$.

The case $t_i=z_2$ is symmetric. Hence Example~\ref{ex:IASspec1} satisfies sequential solvency.

\bigskip

\textbf{Proof that Example \ref{ex:IASspec2} satisfies sequential solvency.}

The three types and zones are symmetric up to cyclic relabeling, so fix an arbitrary officer $i$ with $t_i=z_1$. For each destination zone $z_\ell$, let
\[
    o_\ell=\bigl|\{j\in J:a_j\in z_\ell\}\bigr|,
    \qquad
    x_\ell=\bigl|\{j\in J:a_j\in z_\ell,\ t_j=z_1\}\bigr|.
\]
Each zone has capacity $4$, and there are only two type--$z_1$ officers in $I\setminus\{i\}$, so
\[
    x_1+x_2+x_3\leq 2.
\]

Suppose, toward a contradiction, that $i$ has no admissible state. Since total capacity is $12$ and $|J|\leq 8$, at least four seats are vacant. For every zone with a vacant seat, the corresponding type--$z_1$ cap must therefore be binding; otherwise $i$ could be placed at a vacant seat in that zone. Thus:
\[
    o_1<4 \Rightarrow x_1=1,
    \qquad
    o_2<4 \Rightarrow x_2=2,
    \qquad
    o_3<4 \Rightarrow x_3=2.
\]
These implications can hold for at most one zone: two vacant zones would require at least three type--$z_1$ officers in $J$, contradicting $x_1+x_2+x_3\leq 2$.

Hence, under the no-admissible-state supposition, at most one zone can have vacancies. But there are at least four vacant seats and each zone has capacity exactly four. Therefore either more than four vacancies would have to lie in one zone, which is impossible, or exactly one zone is empty. In the latter case the relevant type--$z_1$ cap in that zone is not binding because $x_\ell=0$, again contradicting the no-admissible-state supposition.

Therefore some admissible state exists for $i$. The other types follow by cyclic relabeling. Hence Example~\ref{ex:IASspec2} satisfies sequential solvency.

\bigskip

\textbf{Proof that Example \ref{ex:IASspec3} satisfies sequential solvency.}

Fix an arbitrary officer $i$ with $t_i=s_1$. For type~$s_1$, every state in the away zone
\[
    z_2=\{s_4,s_5,s_6\}
\]
is unconstrained: for every $s\in z_2$, $H^{s,s_1}=\emptyset$. The total capacity of $z_2$ is $6$, while $|J|\leq |I|-1=5$. Hence at least one seat in $z_2$ is vacant, and that seat is admissible for $i$.

The same argument, with $s_1$ and $s'$ interchanged, applies to officers of type~$s'$. Hence Example~\ref{ex:IASspec3} satisfies sequential solvency.

\bigskip

\textbf{Sequential solvency of the upper-bound system used in the proof of \Cref{thm:NoMechanismAllThree}.}

The construction in the proof of \Cref{thm:NoMechanismAllThree} has three officers of a single type $t$, two states $s_1,s_2$ with $q_{s_1}=q_{s_2}=2$, and the single upper-bound $(\{t\},\{s_1\},1)$. Fix any officer $i$, any $J\subseteq I\setminus\{i\}$, and any feasible partial allocation $a^J$ respecting this bound. Since $|J|\leq 2$, either no officer in $J$ is assigned to $s_1$, in which case $s_1$ has spare capacity and the bound on $s_1$ is not binding, or exactly one officer in $J$ is assigned to $s_1$, in which case the bound on $s_1$ is binding but $s_2$ has spare capacity and no upper-bound applies to $s_2$. In either case, officer $i$ has an admissible state.

\section{Appendix: Additional Efficiency Results}
\label{sec:appendix-efficiency}

As shown in \Cref{thm:NoMechanismAllThree}, visibly fair mechanisms cannot, in general, elicit sufficient information to implement constrained Pareto efficient allocations while respecting modular upper-bounds. Nonetheless, if each officer's message space can be conditioned on the assignments of higher-priority officers, constrained Pareto efficiency can be recovered. To see that the impossibility is purely informational, it suffices to consider the following sequential procedure. The more challenging design problem---and the focus of the main text---is how to construct the message spaces for a \emph{static} mechanism that simultaneously upholds visible fairness and respects distributional constraints, a task addressed by the Modular Priority Mechanism and the constraint-induced message spaces developed in \Cref{sec:DistributionalObjectives}.

\begin{defi}\label{def:DynamicModularPriorityMech}
The \textbf{Dynamic Modular Priority Mechanism} $\psi$ operates as
follows.

\smallskip
\noindent\textbf{Quotas-update procedure.}
Let $(B^{t})_{t\in T}$ denote the current binding sets.
For each $h=(\Xi_{h},S_{h},k_{h})\in H$:
\begin{itemize}
  \item Let
        $N_{h}=\bigl|\{\,i\in I:a_{i}\in S_{h}\text{ and }
        t_{i}\in\Xi_{h}\,\}\bigr|$.
  \item If $N_{h}=k_{h}$, then for every $t\in\Xi_{h}$, add $h$ to
        $B^{t}$.
\end{itemize}

\smallskip
\noindent\textbf{Initialization.}
\begin{itemize}
  \item For each officer $i\in I$, set $a_{i}=\varnothing$.
  \item For each type $t\in T$, initialize $B^{t}=\varnothing$.
  \item Apply the quotas-update procedure (so that any upper-bound
        $h\in H$ with $k_{h}=0$ is registered as binding).
\end{itemize}

\smallskip
\noindent\textbf{Sequential assignment.}
Process the officers in the order $(i_{1},\dots,i_{n})$. For each
$k=1,\dots,n$:
\begin{enumerate}
  \item Construct the zonal message space $M_{i_{k}}$ that partitions
        $S$ into
        \[
          z_{1}\;=\;\Bigl\{\,s\in S\;:\;
              \bigl(\forall (\Xi_{h},S_{h},k_{h})\in B^{t_{{k}}},
              \;s\notin S_{h}\bigr)
              \text{ and }
              \bigl|\{\,j<k:a_{i_{j}}=s\,\}\bigr|<q_{s}\,\Bigr\},
        \]
        \[
          z_{2}\;=\;S\setminus z_{1}.\footnote{Apply sequential solvency to officer $i_k$ and the partial allocation $(a_{i_1},\dots,a_{i_{k-1}})$ of already assigned officers. It yields a state with strictly positive residual capacity that lies outside every $S_h$ whose upper-bound is already binding for type $t_{k}$. Hence this state belongs to $z_1$, and therefore $z_1\neq\varnothing$. If $z_2=\varnothing$, the offered message space is the one-zone zonal message space with zone $z_1=S$.}
        \]

  \item Elicit from officer $i_{k}$ a message $m_{i_{k}}\in M_{i_{k}}$.
        Let $s^{\ast}$ be the $m_{i_{k}}$-most-preferred state in
        $z_{1}$.
  \item Set $a_{i_{k}}=s^{\ast}$.
  \item Apply the quotas-update procedure.
\end{enumerate}

\smallskip
\noindent\textbf{Outcome.}
After processing $i_{1},\dots,i_{n}$, the mechanism returns the
allocation $a=(a_{i})_{i\in I}$.
\end{defi}

The Dynamic Modular Priority Mechanism, therefore, sequentially assigns officers, dynamically tailoring the menu of available states for each officer based on previous assignments. At each officer’s turn, the mechanism constructs a zone containing all states with remaining capacity that do not belong to any subset of states where a modular upper-bound for the officer’s type has already become binding. The officer is then matched to her most-preferred state among these feasible alternatives. By construction, this ensures visible fairness since each officer always obtains their top choice from that zone. Furthermore, by dynamically adjusting the zones and enabling officers to fully express their preferences within these constraints, the mechanism attains constrained efficiency: any allocation that improves an officer’s assignment without harming others would necessarily violate at least one modular upper-bound. 

\paragraph{Dynamic interpretation of visible fairness and feasibility.}
For the Dynamic Modular Priority Mechanism, visible fairness is evaluated with respect to the realized history and the realized messages. At the history $h_k$ faced by officer $i_k$, the mechanism offers the zonal message space $M_{i_k}(h_k)$ and records the submitted message $m_{i_k}\in M_{i_k}(h_k)$. The final allocation is visibly fair if no officer $i_k$ visibly envies a lower-priority officer and no officer visibly prefers a state with unused physical capacity, where both comparisons are evaluated using the realized relation $\succ_{m_{i_k}}$. Respecting modular upper-bounds is evaluated on the final allocation in the same sense as in \Cref{def:respectsModularUpperBounds}.

\paragraph{Game-theoretic formulation.}
Because the mechanism is sequential, its incentive properties are naturally formulated in terms of histories and strategies.  Let $h_k=(a_{i_1},\dots,a_{i_{k-1}})$ denote the \emph{history} at step~$k$ (with $h_1=\emptyset$).  The message space $M_{i_k}(h_k)$ presented to officer~$i_k$ depends on~$h_k$ through the zone~$z_1$ defined in step~1 of the mechanism.  A \emph{strategy} for officer~$i_k$ is a mapping $\sigma_{i_k}$ that selects, for every history~$h_k$, a message $\sigma_{i_k}(h_k)\in M_{i_k}(h_k)$.  The strategy is \emph{truthful} if $\sigma_{i_k}(h_k)\in\mathcal{T}_{i_k}(\succ_{i_k})$ for every~$h_k$.  Because the zone~$z_1$ is determined entirely by~$h_k$ and does not depend on~$m_{i_k}$, officer~$i_k$'s assignment is a function of $(h_k,m_{i_k})$ alone.  Writing it as $\psi_k(h_k,m_{i_k})$, the mechanism is \emph{strategy-proof in dominant strategies} if, for every officer~$i_k$, every preference $\succ_{i_k}\!\in\mathcal{Q}_{i_k}$, every history~$h_k$, every truthful message $m_{i_k}\in\mathcal{T}_{i_k}(\succ_{i_k})\cap M_{i_k}(h_k)$, and every alternative message $m'_{i_k}\in M_{i_k}(h_k)$:
\[
  \psi_k(h_k,\,m_{i_k})\;\succsim_{i_k}\;\psi_k(h_k,\,m'_{i_k}).
\]
\begin{thm}
\label{dynamicModPriorityMechVisiblyFairModularAndSP}
Suppose that $H$ is a sequentially solvent modular upper-bound system. Then the Dynamic Modular Priority Mechanism is visibly fair, constrained Pareto efficient, strategy-proof, and respects upper-bounds. 
\end{thm}

\begin{proof}
Fix a problem, a modular upper-bound system $H$ satisfying sequential solvency, and a profile of truthful preferences. Let
\[
    a=(a_i)_{i\in I}
\]
be the allocation produced by the Dynamic Modular Priority Mechanism.

We first record two observations about the construction of the feasible zone. At the beginning of the step in which officer $i_k$ is processed, let
\[
    a^{k-1}=(a_{i_1},\ldots,a_{i_{k-1}})
\]
denote the partial allocation of the officers already assigned. For officer $i_k$, the mechanism defines
\[
z_1^k =
\left\{
s\in S:
\begin{array}{l}
| \{j<k:a_{i_j}=s\}|<q_s, \text{ and}\\
s\notin S_h \text{ for every } h=(\Xi_h,S_h,k_h)\in B^{t_{k}}_{k-1}
\end{array}
\right\},
\]
and
\[
    z_2^k=S\setminus z_1^k.
\]
Thus $z_1^k$ is the set of states that are currently feasible for officer $i_k$: they have spare capacity and are not contained in any currently binding upper-bound relevant to type $t_{k}$. The complementary zone $z_2^k$ contains states that are unavailable either because their physical capacity is exhausted or because assigning officer $i_k$ there would violate a binding upper-bound.

Sequential solvency implies that $z_1^k\neq\emptyset$ at every step. Indeed, the partial allocation $a^{k-1}$ respects $H$ by the induction argument below. Applying sequential solvency to officer $i_k$ and the partial allocation $a^{k-1}$ gives a state $s$ such that $s$ has spare capacity and, for every upper-bound $h\in H_{s,t_{k}}$, the count of previously assigned officers covered by $h$ is strictly below $k_h$. Equivalently, $s$ is not in any set $S_h$ whose upper-bound is already binding for type $t_{k}$. Hence $s\in z_1^k$.

This observation also covers initially binding constraints. If some upper-bound has $k_h=0$, then it is binding at the empty allocation. Therefore, for every type $t\in \Xi_h$, every state in $S_h$ is placed in the initially unavailable zone $z_2^1$, not in $z_1^1$. Sequential solvency guarantees that, even after excluding all such initially binding sets, each first officer still has at least one admissible state in $z_1^1$.

\medskip

\noindent\textbf{Step 1: The mechanism respects upper-bounds.}

We prove by induction on the assignment order that every partial allocation generated by the mechanism respects $H$.

At step $0$, no officer is assigned. Hence, for each upper-bound $h=(\Xi_h,S_h,k_h)$,
\[
    \bigl|\{i:a_i\in S_h,\ t_i\in\Xi_h\}\bigr|=0\leq k_h.
\]
If $k_h=0$, then $h$ is already binding at this initial history, and the mechanism records this by excluding $S_h$ from the feasible zone of every type in $\Xi_h$.

Now suppose that, before officer $i_k$ is assigned, the partial allocation $a^{k-1}$ respects $H$. Officer $i_k$ is assigned to some state
\[
    a_{i_k}\in z_1^k.
\]
Consider any upper-bound $h=(\Xi_h,S_h,k_h)$.

If $t_{k}\notin\Xi_h$ or $a_{i_k}\notin S_h$, then the count for $h$ is unchanged. If $t_{k}\in\Xi_h$ and $a_{i_k}\in S_h$, then, because $a_{i_k}\in z_1^k$, the upper-bound $h$ was not binding before assigning $i_k$. Therefore,
\[
    \bigl|\{j<k:a_{i_j}\in S_h,\ t_{j}\in\Xi_h\}\bigr|<k_h.
\]
After assigning $i_k$ to $a_{i_k}$, this count increases by exactly one and therefore remains weakly below $k_h$. Thus no upper-bound is violated.

By induction, the final allocation $a$ respects all modular upper-bounds.

\medskip

\noindent\textbf{Step 2: The mechanism is visibly fair.}

Fix officer $i_k$. At the history faced by officer $i_k$, the mechanism partitions $S$ into two zones, $z_1^k$ and $z_2^k$, where $z_1^k$ contains exactly the states that are feasible for $i_k$ given prior assignments and currently binding upper-bounds relevant to her type. The message space offered to $i_k$ is zonal with respect to this partition, and the mechanism assigns $i_k$ her $m_{i_k}$-maximal state in $z_1^k$.

Hence no state in $z_1^k$ is ranked above $a_{i_k}$ under $m_{i_k}$, and no state in $z_2^k$ is comparable to $a_{i_k}$ under $m_{i_k}$. Therefore there is no state $s$ with $s\succ_{m_{i_k}}a_{i_k}$ that is either assigned to a lower-priority officer or left with spare capacity. This holds for every realized history and every officer, so the final allocation is visibly fair.

\medskip

\noindent\textbf{Step 3: The mechanism is strategy-proof.}

Fix an officer $i_k$ and fix any history generated prior to reporting. This history consists of the assignments of officers $i_1,\ldots,i_{k-1}$ and the resulting set of binding upper-bounds. Both are determined before $i_k$ sends her message. Hence the partition
\[
    \{z_1^k,z_2^k\}
\]
that officer $i_k$ faces is independent of her message.

Given this fixed partition, the mechanism assigns $i_k$ the state that her message ranks highest within $z_1^k$. A truthful message ranks states within $z_1^k$ according to her true preference. Hence truthful reporting gives her her truly best state in the fixed feasible set $z_1^k$. Any misreport can only select another state in the same set $z_1^k$, which cannot be strictly preferred to the truthful assignment.

This argument holds after every possible history and for every behavior of the other officers. Therefore truthful reporting is a weakly dominant strategy, and the Dynamic Modular Priority Mechanism is strategy-proof.

\medskip

\noindent\textbf{Step 4: The mechanism is constrained Pareto efficient.}

Suppose, toward a contradiction, that there exists another allocation
\[
    a'=(a'_i)_{i\in I}
\]
that respects all modular upper-bounds and Pareto dominates $a$, with at least one officer strictly better off. Let $i_k$ be the highest-priority officer who is strictly better off under $a'$:
\[
    a'_{i_k}\succ_{i_k} a_{i_k},
\]
and for every $\ell<k$,
\[
    a'_{i_\ell}=a_{i_\ell}.
\]
Such an officer exists because $a'$ strictly Pareto improves upon $a$ for at least one officer.

Consider the step at which $i_k$ was assigned by the mechanism. We show that $a'_{i_k}\in z_1^k$ at that step.

First, $a'_{i_k}$ had spare physical capacity at the time $i_k$ was processed. If not, then all seats of $a'_{i_k}$ were already occupied in $a$ by officers with priority higher than $i_k$. Since $a'$ agrees with $a$ for all higher-priority officers, those same seats would also be occupied by higher-priority officers in $a'$, leaving no capacity for $i_k$ at $a'_{i_k}$, contradicting feasibility of $a'$.

Second, $a'_{i_k}$ was not excluded by any binding upper-bound relevant to type $t_{k}$. Suppose instead that there exists an upper-bound
\[
    h=(\Xi_h,S_h,k_h)
\]
such that $t_{k}\in\Xi_h$, $a'_{i_k}\in S_h$, and $h$ was already binding at the beginning of step $k$. Then exactly $k_h$ higher-priority officers of types in $\Xi_h$ were already assigned to states in $S_h$ under $a$. Since $a'$ agrees with $a$ for all higher-priority officers, the same $k_h$ officers are assigned to $S_h$ under $a'$. Assigning $i_k$ also to $S_h$ under $a'$ would raise the count to at least $k_h+1$, contradicting the assumption that $a'$ respects $H$.

Therefore $a'_{i_k}\in z_1^k$. But the Dynamic Modular Priority Mechanism assigned $i_k$ her most-preferred state in $z_1^k$ under a truthful message. Hence it cannot be that
\[
    a'_{i_k}\succ_{i_k} a_{i_k},
\]
a contradiction.

Thus no allocation that respects the upper-bounds can Pareto dominate the mechanism's outcome. Therefore the outcome is constrained Pareto efficient.

Combining the four steps, the Dynamic Modular Priority Mechanism is visibly fair, respects upper-bounds, strategy-proof, and constrained Pareto efficient.
\end{proof}

The Dynamic Modular Priority Mechanism is, by design, a standard serial dictatorship adapted to account for binding upper-bounds; its role here is to diagnose the source of the static impossibility rather than to serve as a practical recommendation.  It is also worth noting that its sequential elicitation contrasts with the single round required by a static mechanism. While this added complexity may be negligible in small markets, it can become impractical in larger settings where the number of officers is substantial.

\subsection{Visible efficiency and true Pareto efficiency}

This subsection explores the relationship between visible efficiency and fairness in allocation mechanisms. An allocation is visibly efficient under message profile $m$ if no alternative allocation visibly Pareto dominates it (meaning no reallocation would make all affected agents better off according to their reported preferences). The key results show that visibly fair allocations are always visibly efficient, but the converse does not hold. When comparing visible efficiency to true Pareto efficiency (based on actual preferences), every Pareto efficient allocation is visibly efficient under truthful messages, but visible efficiency does not guarantee Pareto efficiency—even visibly fair allocations can be Pareto inefficient. Furthermore, when message $\hat{m}$ contains more preference information than $m$, allocations that are visibly fair or efficient under $\hat{m}$ remain so under $m$, but not vice versa, demonstrating that visible fairness and efficiency depend critically on the message spaces.

An allocation $a^\prime$ \textbf{visibly Pareto dominates} allocation $a$ under message profile $m$ if $a^\prime\neq a$ and, for every officer $i\in I$ such that $a^\prime_i\neq a_i$, we have $a^\prime_i\succ_{m_i}a_i$. An allocation $a$ is \textbf{visibly efficient under $m$} if there is no allocation $a^\prime\in\mathcal{A}$ that visibly Pareto dominates $a$ under $m$. A mechanism $\psi$ is \textbf{visibly efficient} if $\psi(m)$ is visibly efficient under $m$ for every $m\in M$.

\begin{thm}
\label{thm:Efficiency}
For any message profile $m\in M$, the following statements are true.
\begin{enumerate}
    \item Every visibly fair allocation is visibly efficient.
    \item A visibly efficient allocation may not be visibly fair. 
    
\end{enumerate}

\end{thm}

\begin{proof}
       \underline{Statement 1}: Fix a message profile \(m\in M\). Suppose \(a \in \mathcal{A}\) is a visibly fair allocation that is not visibly efficient under $m$. Then there exists another allocation \(a^{\prime} \in \mathcal{A}\) such that for all $i \in I$ with $a^\prime_i \not= a_i$, we have   $a^\prime_i \succ_{m_i} a_i$. 
    
Let $i^*$ be the highest ranking officer among the ones that are allocated to a different state under $a$ and $a'$ and let $s^*$ be her assigned state under $a'$. That is, for \[ \bar{I}:=\left \{ i\in I:a_i\ne a^\prime_i \right \}, \qquad i^* := \arg\min_{i\in \bar{I}} \pi(i),  \quad \text{and} \quad a'_{i^*}=s^*.\]
 
The allocation $a$ must be visibly unfair under $m$. This is because either $a_i=s^*$ for some $i \in \bar{I}$, or $|\{ i\in I: a_i=s^*\}|< q_{s^*}$. Yet $s^* \succ_{m_{i^*}} a_{i^*}$.

    \underline{Statement 2}:     Let $I=\{i_1,i_2\}$, $S=\{s_1,s_2\}$, and $q_{s_1}=q_{s_2}=1$. Additionally, assume zonal message space for both officers, $Z=\{z_1\}$ with $z_1 = \{s_1,s_2\}$.  For preferences ${s_1} \succ_{i_1} s_2$ and $s_1 \succ_{i_2} {s_2} $, the allocation $(a_1, a_2) =(s_2,s_1)$ is visibly efficient but not visibly fair under the truthful message.

\end{proof}

In contrast, an allocation $a^\prime$ \textbf{Pareto dominates} allocation $a$ if $a^\prime\neq a$ and, for every officer $i\in I$ such that $a^\prime_i\neq a_i$, we have $a^\prime_i\succ_i a_i$. An allocation $a$ is \textbf{Pareto efficient} if there is no allocation $a^\prime\in\mathcal{A}$ that Pareto dominates $a$. A mechanism $\psi$ is \textbf{Pareto efficient} if $\psi(m)$ is Pareto efficient for every $m\in M$.

\begin{thm}
\label{thm:Efficiency2}
For any truthful message profile $m\in M$, the following statements are true.
\begin{enumerate}
    \item Every Pareto efficient allocation is visibly efficient.
    \item A visibly efficient allocation may not be Pareto efficient. 
    
    \item A visibly fair allocation may not be Pareto efficient.
\end{enumerate}

\end{thm}

\begin{proof}
    \underline{Statement 1}: Consider an allocation \(a \in \mathcal{A}\) that is not visibly efficient for some truthful message \(m\in M\). This implies there is another allocation \(a^\prime \in \mathcal{A}\) such that for all $i \in I$ with $a^\prime_i \not= a_i$, we have $a^\prime_i \succ_{m_i} a_i$, and therefore $a^\prime_i \succ_i a_i$ ($m$ is a truthful message). Thus, \(a\) is not Pareto efficient.

    \underline{Statement 2 and 3}:     Let $I=\{i_1,i_2\}$, $S=\{s_1,s_2,s_3\}$, and $q_{s_1}=q_{s_2}=q_{s_3}=1$. Additionally, assume a zonal message space, $Z=\{z_1, z_2\}$ with  $z_1=\{s_1,s_2\}$ and $z_2=\{s_3\}$. For preferences $s_3 \succ_{i_1} {s_1} \succ_{i_1} s_2$ and $s_1 \succ_{i_2} {s_3} \succ_{i_2} s_2$, the allocation $(a_1, a_2) =(s_1,s_3)$ is visibly fair and visibly efficient under the truthful message, as $s_1$ and $s_3$ are not comparable. However, it is not Pareto efficient.
\end{proof}

Message profile $\hat{m}$ \textbf{contains more preference information than} message profile $m$ if, for every officer $i\in I$, the induced expressed relation under $m_i$ is contained in the induced expressed relation under $\hat{m}_i$:
$$
\succ_{m_i}\subseteq\succ_{\hat{m}_i}.
$$
Equivalently, for all $i\in I$ and all $s,s'\in S$,
$$
s\succ_{m_i}s' \quad\implies\quad s\succ_{\hat{m}_i}s'.
$$

\begin{thm}
\label{thm:Efficiency3}
Suppose the message profile $\hat{m}$ contains more preference information than the message profile $m$. Then the following statements are true. 

\begin{enumerate}
    \item Every visibly fair allocation under $\hat{m}$ is also visibly fair under $m$.
    \item A visibly fair allocation under ${m}$ may not be visibly fair under $\hat{m}$.
    \item Every visibly efficient allocation under $\hat{m}$ is also visibly efficient under $m$.
    \item A visibly efficient allocation under ${m}$ may not be visibly efficient under $\hat{m}$.
\end{enumerate}
\end{thm}

\begin{proof}
    \underline{Statement 1}: There are two cases: (i) Allocation \(a\) is not visibly fair under \(m\) because there exist some \(i\in I\) such that there is a $j\in I$ such that $a_i\not=a_j$, $\pi(i) <  \pi(j)$, and $a_j \succ_{m_i} a_i$. Since $a_j \succ_{\hat{m}_i} a_i$, $a$ cannot be visibly fair under $\hat{m}$.

    (ii) Allocation \(a\) is not visibly fair under \(m\) because there exist some \(i\in I\) such that there is a $s\in S$ such that $a_i\not=s$, $|\{ i\in I: a_i=s\}|<q_s$, and $s \succ_{m_i} a_i$. Since $s \succ_{\hat{m}_i} a_i$,  $a$ cannot be visibly fair under $\hat{m}$.

\underline{Statement 3}: We use the contrapositive. If allocation \(a\) is not visibly efficient under \(m\), then there exists another allocation \(a^\prime\) such that for all $i \in I$ with $a^\prime_i \not= a_i$, we have   $a^\prime_i \succ_{m_i} a_i$. Since $a^\prime_i \succ_{\hat{m}_i} a_i$ for all such $i\in I$, $a$ cannot be visibly efficient under \(\hat{m}\). 

\underline{Statement 2 and 4}: Let $I=\{i_1,i_2\}$, $S=\{s_1,s_2,s_3\}$, and $q_{s_1}=q_{s_2}=q_{s_3}=1$. For each officer, let the message space contain both the partial message induced by the zonal partition $Z=\{z_1,z_2\}$, with $z_1=\{s_1,s_2\}$ and $z_2=\{s_3\}$, and the complete message induced by the one-zone partition $Z=\{S\}$. For preferences $s_3 \succ_{i_1} s_1 \succ_{i_1} s_2$ and $s_1 \succ_{i_2} s_3 \succ_{i_2} s_2$, the allocation $(a_1,a_2)=(s_1,s_3)$ is visibly fair and visibly efficient under the truthful partial message profile $m$, because $s_1$ and $s_3$ are not comparable under $m$.

Let $\hat m=(\succ_i)_{i\in I}$ be the truthful complete message profile. Then $\hat m$ contains more preference information than $m$. However, the allocation $(a_1,a_2)=(s_1,s_3)$ is neither visibly fair nor visibly efficient under $\hat m$.
\end{proof}

\section{Appendix: Indirect Message Spaces}\label{appendixindirect}
\subsection{Preference Elicitation in All India Services}
The 2017 Cadre Allocation Policy for India's All India Services, including the Indian Administrative Service (IAS), Indian Police Service (IPS), and Indian Forest Service (IFoS), introduces a zonal system that divides all states and union territories into five geographical zones, requiring candidates to first indicate their zone preferences in descending order, followed by state preferences within each preferred zone. For illustrative purposes, we include a screenshot of the submitted preferences from 2020 IFoS examination.

\begin{figure}[!ht]
    \centering
    \includegraphics[width=1\linewidth]{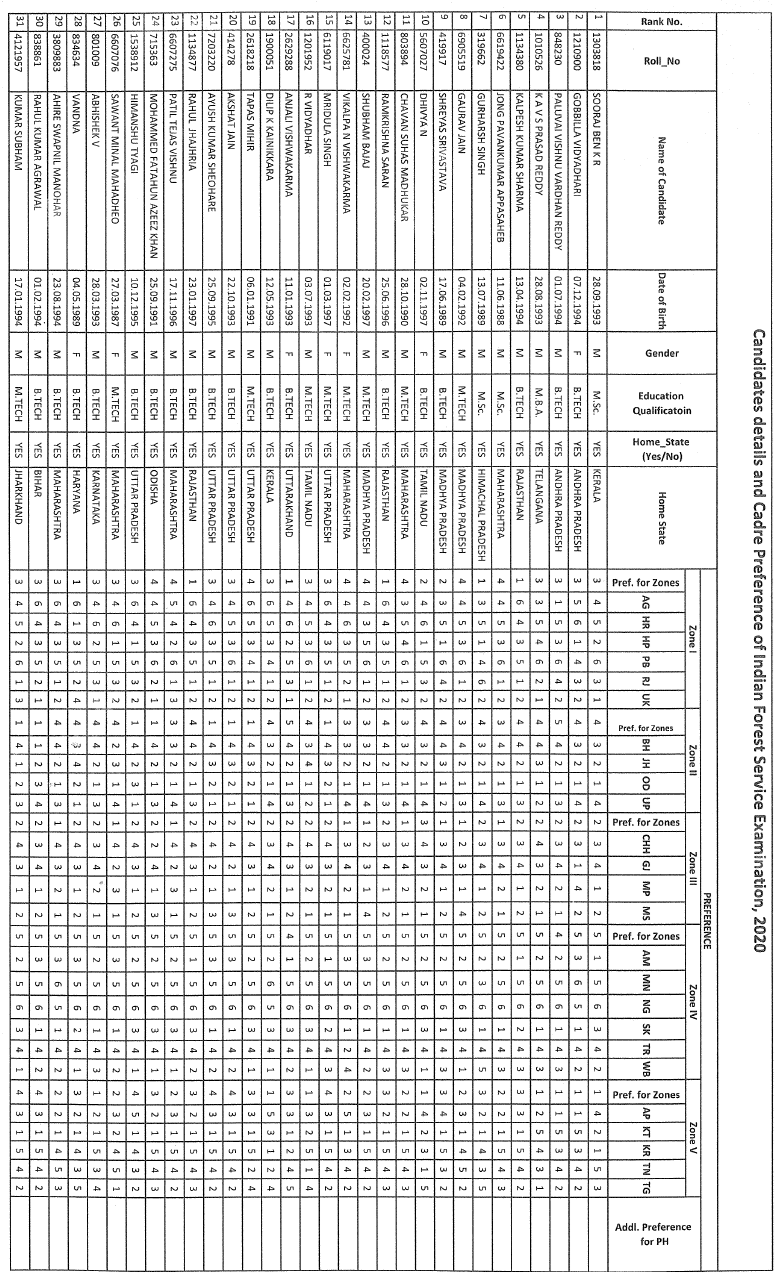}
    \caption{Example preference form for the 2020 IFoS cadre allocation.}
    \label{fig:ifos_form}
\end{figure}

\subsection{Rank-order lists in Chinese College Admissions}
23 out of 31 provinces in China implement the structured rank-order list system, in which majors are effectively nested under colleges, as noted by \cite{huchinesecolleges}. These provinces include: Shanghai, Beijing, Tianjin, Hainan, Jiangsu, Fujian, Hubei, Hunan, Guangdong, Heilongjiang, Gansu, Jilin, Anhui, Jiangxi, Guangxi, Shanxi, Henan, Shaanxi, Ningxia, Sichuan, Yunnan, Tibet, and Xinjiang. For illustrative purposes, we include a screenshot of the official college-major preference form from Shanghai.

\begin{figure}[!ht]
    \centering
    \includegraphics[width=1\linewidth]{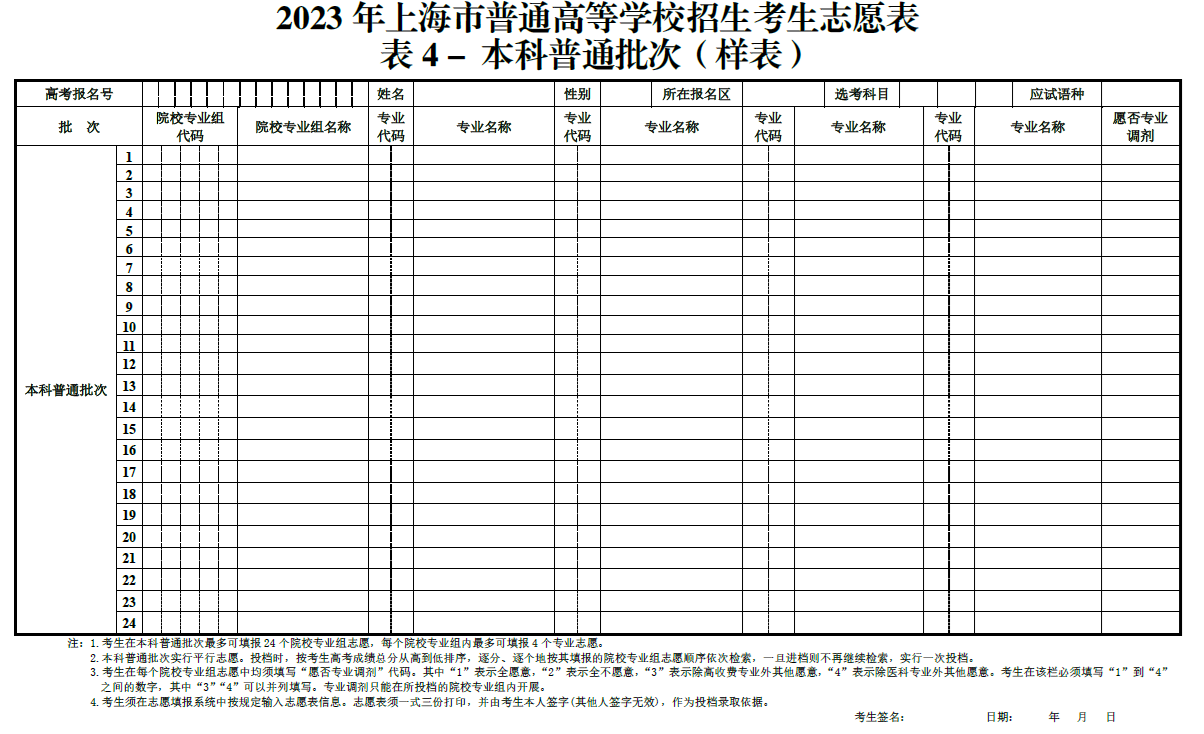}
    \caption{Example college-major preference form from Shanghai.}
    \floatfoot{\footnotesize Source: \url{https://www.shmeea.edu.cn/page/08000/20230407/17353.html}}
    \label{fig:shanghai_form}
\end{figure}

\subsection{Reserve Officer Training Corps (ROTC) Mechanism }

\cite{sonmez2013bidding}'s model of cadet-branch matching problem consists of
\begin{enumerate}
    \item a finite set of cadets $I = \{i_1, i_2, \dots, i_n\}$,
    \item a finite set of branches $B = \{b_1, b_2, \dots, b_m\}$,
    \item a vector of branch capacities $q = (q_b)_{b \in B}$,
    \item a set of ``terms" $T = \{t_1, \dots, t_k\}$,
    \item a list of cadet preferences $P = (P_i)_{i \in I}$ over $(B \times T) \cup \{\emptyset\}$, and
    \item a list of base priority rankings $\pi = (\pi_b)_{b \in B}$.
\end{enumerate}

The ROTC mechanism is not direct. Instead, each cadet submits a ranking of branches $\succ_i'$, and he can sign a branch-of-choice contract for any of his top three choices under $\succ_i'$.
\end{document}